%% file: CSDA_Combined.tex
\documentclass[a4paper,fleqn]{cas-sc}

\usepackage[authoryear,longnamesfirst]{natbib}
\usepackage{bibunits}
\defaultbibliographystyle{cas-model2-names}
\defaultbibliography{allpapers}
\usepackage{amssymb,amsmath,amsthm,mathtools}
\usepackage[textfont=it,tableposition=above,font=small]{caption}
\usepackage{xcolor}
\usepackage[inline]{enumitem}
\usepackage{booktabs} 
\usepackage{subcaption}
\usepackage{multirow}
\usepackage{float}

\newtheorem{theorem}{Theorem}
\newtheorem{lemma}[theorem]{Lemma}
\theoremstyle{definition}
\newtheorem*{definition}{Definition}

\newcommand{\RPD}{\textsf{\textbf{RPD}}} 
\newcommand{\RPDOD}{\textsf{\textbf{RPD-OD}}} 
\newcommand{\RHD}{\textsf{\textbf{RHD}}} 
\newcommand{\FD}{\textsf{\textbf{FD}}} 
\newcommand{\ID}{\textsf{\textbf{ID}}} 
\newcommand{\MBD}{\textsf{\textbf{MBD}}} 
\newcommand{\VV}{\mathcal{V}} 
\newcommand{\HH}{\mathbb{H}} 
\newcommand{\R}{\mathbb{R}} 
\newcommand{\N}{\mathbb{N}} 
\newcommand{\argmax}{\operatornamewithlimits{argmax}} 
\newcommand{\pr}{\mathbb{P}} 
\newcommand{\E}{\mathbb{E}} 
\newcommand{\var}{\mathsf{var}} 
\newcommand{\MAD}{\mathsf{MAD}} 
\newcommand{\med}{\mathsf{med}} 
\newcommand{\set}[1]{\left\lbrace #1 \right\rbrace} 
\renewcommand{\SS}{\mathbb{S}} 
\renewcommand{\P}[1]{\mathcal{P}({#1})} 
\newcommand{\dd}{\,\mathrm{d}} 
\newcommand{\abs}[1]{\left|{#1}\right|} 
\newcommand{\inner}[1]{\left\langle#1\right\rangle} 
\newcommand{\norm}[1]{\left\lVert#1\right\rVert} 
\newcommand{\Clo}[1]{\operatorname{cl}\left({#1}\right)} 
\newcommand{\spn}[1]{\operatorname{span}\left(#1\right)} 
\newcommand{\st}{\,\colon\,} 
\newcommand{\eqdis}{\overset{d}{=}} 
\newcommand{\tr}{^\top} 
\newcommand{\indic}{\mathbf{1}} 
\newcommand{\xasto}[1][]{\xrightarrow[\scalebox{0.7}{$#1$}]{\text{a.s.}}} 

\begin{document}
\gdef\lastpage{\pageref*{main:LastPage}}
\begin{bibunit}

\let\WriteBookmarks\relax
\def\floatpagepagefraction{1}
\def\textpagefraction{.001}

\shorttitle{Projection depth for functional data: Practical issues, computation and applications}    

\shortauthors{F. Bo\v{c}inec, S. Nagy and H. Yeon}

\title [mode = title]{Projection depth for functional data: Practical issues, computation and applications}  

\author[1]{Filip Bo\v{c}inec}[orcid=0009-0000-4415-9302]
\ead{bocinec@karlin.mff.cuni.cz}
\credit{Conceptualization, Formal analysis, Software, Writing – original draft}
\cormark[1]
\cortext[1]{Corresponding author}

\author[1]{Stanislav Nagy}[orcid=0000-0002-8610-4227]
\ead{nagy@karlin.mff.cuni.cz}
\credit{Conceptualization, Methodology, Supervision, Writing – review \& editing}

\author[2]{Hyemin Yeon}[orcid=0009-0008-3585-2103]
\ead{hyeon1@kent.edu}
\credit{Methodology, Software, Supervision, Writing – review \& editing}

\affiliation[1]{organization={Faculty of Mathematics and Physics, Charles University},
            addressline={Ke Karlovu 2027/3}, 
            city={Prague},
            postcode={12116}, 
            country={Czech Republic}}

\affiliation[2]{organization={Department of Mathematical Sciences, Kent State University},
            addressline={233 Summit Street}, 
            city={Kent},
            postcode={44242}, 
            state={OH},
            country={USA}}

\nonumnote{This work was supported by the ERC CZ grant LL2407 of the Ministry of Education, Youth and Sport of the Czech Republic.}

\begin{abstract}
Statistical analysis of functional data is challenging due to their complex patterns, for which functional depth provides an effective means of reflecting their ordering structure. 
In this work, we investigate practical aspects of the recently proposed regularized projection depth (RPD), which induces a meaningful ordering of functional data while appropriately accommodating their complex shape features. 
Specifically, we examine the impact and choice of its tuning parameter, which regulates the degree of effective dimension reduction applied to the data, and propose a random projection–based approach for its efficient computation, supported by theoretical justification. 
Through comprehensive numerical studies, we explore a wide range of statistical applications of the RPD and demonstrate its particular usefulness in uncovering shape features in functional data analysis.
This ability allows us to propose a novel RPD-based outlier detection method, and to demonstrate that the RPD outperforms competing depth-based approaches in tasks such as functional classification and two-sample hypothesis testing.
\end{abstract}


\maketitle

\section{Introduction: Regularized projection depth}

Over the past two decades, the ordering structure of functional data has been extensively studied, leading to many effective applications in Statistics and Machine Learning. 
A central concept is functional depth, an extension of multivariate depth functions such as halfspace  \citep{Tukey1975}, simplicial \citep{Liu1990}, spatial \citep{Chaudhuri1996}, and zonoid \citep{Mosler2002} depths, 
which measures centrality versus outlyingness for functional data.
Examples of functional depths include integrated \citep{Fraiman_Muniz2001}, (modified) band \citep{Lopez_Romo2009}, infimal \citep{Mosler2013}, spatial \citep{Chakraborty_Chaudhuri2014} and spherical \citep{Elmore_etal2006, Mendros_Nagy2025} depths.
All these depths naturally induce rankings and support a variety of statistical tools for functional data analysis \citep{SunGenton2011, Chakraborty_Chaudhuri2014b, Hubert_etal2017, RamsayChenouri2024}. 

An important distinction from multivariate data is that functional data often exhibit complex shape patterns.
Consequently, in the functional setting, centrality depends not only on magnitude, location, or pointwise variability, but crucially on the shape of the data itself.
Although functional depth provides an important framework for ordering and summarizing functional observations, 
relatively few general-purpose tools are designed to be explicitly sensitive to shape. 
This complexity motivates the development and practical assessment of depth measures that can effectively capture shape-related features~\citep{Hlubinka_etal2015, Nagy_etal2017, Helander_etal2020, Yeon_etal2025}.

Toward this aim, \cite{Bocinec_etal2026} recently extended the prominent projection depth on the Euclidean space \citep{Zuo2003} to a real infinite-dimensional separable Hilbert space $\HH$.
The projection depth builds on the paradigm of projection pursuit \citep{Huber1985} and the Stahel-Donoho outlyingness \citep{Stahel1981, Donoho1982}.
To explain, let $\inner{\cdot, \cdot}$ be the inner product for $\HH$ with induced norm $\norm{\cdot}$, and denote by $\P{\HH}$ the set of Borel probability measures on $\HH$.
For $x\in \HH$ and a unit vector $v \in \SS = \set{y \in \HH \st \norm{y}=1}$, the outlyingness of $x$ with respect to $X \sim P_X \in \P{\HH}$ in direction $v$ is
\begin{align}\label{eqOut}
	O_v(x;P_X) = \frac{\abs{\inner{x, v} - \med[\inner{X,v}]}}{\MAD[\inner{X,v}]},
\end{align}
where $\med[U]$ and $\MAD[U]$ denote the median and median absolute deviation of a real-valued variable $U$.\footnote{We define the median of a real-valued random variable~$U$, denoted by~$\med[U]$, as the midpoint of all values~$m \in \R$ satisfying $\min\left\{ \pr(U \leq m), \pr(U \geq m) \right\} \geq \tfrac{1}{2}$. 
The median absolute deviation $\MAD[U]$ is then defined as $\med[\abs{U - \med[U]}]$.}
As observed in \citet{Chakraborty_Chaudhuri2014} and \citet{Bocinec_etal2026}, the maximal outlyingness (over all directions $v \in \SS$) can be infinite for almost every $x \in \HH$ under standard models.
This yields a degenerate projection depth in $\HH$, which is not statistically useful. 
To address this degeneracy, the \textit{regularized projection depth} (\RPD{}) was proposed by incorporating a regularization that ensures well-defined and strictly positive values, as defined below.

\begin{definition}[Regularized projection depth;~\citet{Bocinec_etal2026}]
	For $\beta \geq 0$ and $X\sim P_X \in \P{\HH}$, the \emph{regularized set of directions} is
	\begin{equation}\label{eqRegDir}
		\VV_{\beta} = \VV_{\beta}(P_X) = \set{v\in\SS\st \MAD[\inner{X,v}]\geq \beta}.
	\end{equation}
	The \emph{regularized projection depth (\RPD{})} of $x \in \HH$ with respect to $P_X$ is defined as
	\begin{align}
		D_{\beta}(x;P_X)
		= \inf_{v \in \VV_{\beta}} \left( 1 + O_v(x; P_X) \right)^{-1}, \label{eqRPD}
	\end{align}
	where the projection outlyingness $O_v$ is given in~\eqref{eqOut}.
\end{definition}

The previous study~\citep{Bocinec_etal2026} explored the theoretical properties of \RPD{}. Importantly, the non-degeneracy of \RPD{} was established, meaning that for $\beta>0$, the depth function $D_\beta(\cdot; P_X)$ is strictly positive. Further, it was shown that \RPD{} preserves several properties of the usual multivariate projection depth, including quasiconcavity and monotonicity, and is well-behaved under symmetry. 
\RPD{} also exhibits other desirable properties:
(i) it is invariant under orthogonal transformations and shifts, 
(ii) for $\beta>0$, \RPD{} is $1/\beta$-Lipschitz, hence continuous, and 
(iii) under mild conditions, it also satisfies the vanishing-at-infinity property.
Importantly, (iv) uniform consistency of the sample \RPD{} on bounded sets is established, and (v) a high breakdown point of the induced median (any point in $\HH$ that maximizes \RPD{}) is proved.
Nevertheless, the practical performance and applications of \RPD{} remain unexplored.

The primary goal of this paper is to fill this gap and explore practical aspects of the \RPD{} in functional data analysis. 
We first consider the selection of the regularization parameter $\beta$ in the \RPD{} definition.
Following the quantile-based choice suggested in~\citet{Bocinec_etal2026}, 
we take $\beta$ as a quantile at level $u \in (0,1)$ of the MADs of the randomly drawn projections. 
Although this choice was proposed in~\cite{Bocinec_etal2026}, 
its theoretical justification was not established there.
Importantly, we show that this choice preserves key theoretical properties of the \RPD{}, 
including non-degeneracy and consistency. 
Using a quantile-based $\beta$ also enhances data adaptivity; 
for instance, smaller values of $u$ help highlight shape features in functional data. 
In practice, this selection can be implemented through the same random projection approximation used for computing the \RPD{}, whose theoretical validity is established in this paper.

Building on these foundations, this paper provides a systematic numerical investigation of the \RPD{} in central tasks of functional data analysis: outlier detection, classification, and hypothesis testing. 
Specifically, we propose a novel \RPD{}-based outlier detection method, referred to as \RPDOD{}. 
This approach introduces a methodological innovation by applying a mixture model directly to depth values -- a strategy not previously considered in the literature. 
Crucially, the reliability of this procedure stems from the \RPD{}'s inherent sensitivity to shape features; applying the same mixture principle to traditional depths (e.g., integrated depth) would be ineffective, as they often fail to capture structural anomalies.
We demonstrate that \RPDOD{} accurately detects outliers that differ in shape, even in challenging cases where they overlap with regular curves, whereas competing depths show limited effectiveness.

While this novel outlier detection approach is a significant contribution, the versatility of the \RPD{} extends well beyond it. 
When used in depth-depth classification \citep{Li_etal2012}, the \RPD{} also performs exceptionally well in scenarios where groups differ by shape.
Moreover, the \RPD{}-rank-based test is powerful under alternatives involving differences in mean magnitude, mean shape, or covariance structure.
The \RPD{} also shows competitive performance in robust location estimation. Overall, these studies demonstrate the practical effectiveness and versatility of the \RPD{}, with particular strength in shape-focused inferential tasks in functional data analysis.

The remainder of the paper is organized as follows.
Section~\ref{sec: SelectionTuningParameter} proposes a strategy for selecting the regularization parameter $\beta$ and examines which theoretical results for \RPD{} continue to hold when $\beta$ is taken in this data-adaptive manner. 
In Section~\ref{sec: Computation}, we describe an efficient and simple method for approximate computation of the (sample) \RPD{}, 
establish the consistency of this approximation, 
and analyze its convergence properties.
Our statistical contributions are collected in Section~\ref{sec: Simulations}, which mainly addresses three tasks: (i) outlier detection and the introduction of the~\RPDOD{} method (Section~\ref{sec: OutlierDetection}),
(ii)~supervised classification (Section~\ref{sec: Classification}), and
(iii) hypothesis testing (Section~\ref{sec: HypothesisTesting}). 
For each task, we demonstrate the applicability of~\RPD{} and evaluate its performance relative to other, both classical and cutting-edge depth notions across a variety of scenarios. 
Finally, Section~\ref{sec: RealData} demonstrates the application of~\RPD{} to real-world data analysis.
Detailed proofs of all theoretical results (Section~\ref{Supplement: Proofs}), additional simulation studies concerning robust location estimation (Section~\ref{Supplement: LocationEstimation}), supplementary simulation results (Section~\ref{Supplement: Simulations}), and complete \textsf{R} codes of the performed numerical studies (including an \textsf{R} package for computing \RPD{}) are given in an online supplementary material.

\subsection*{Notations}

We consider a real separable (typically infinite-dimensional) Hilbert space $\HH$ with inner product $\inner{\cdot, \cdot}$ and the induced norm $\norm{\cdot}$.
Let $\SS = \set{v \in \HH \st \norm{v}=1}$ denote the unit sphere in $\HH$. For a set $A \subseteq \HH$, we use $\Clo{A}$ and $\spn{A}$ to denote the closure and linear span of $A$, respectively.
All random variables are defined on a common probability space $\left( \Omega, \mathcal F, \pr \right)$.
Throughout the paper, “a.s.” stands for “almost surely”.
The set of all Borel probability measures on a metric space $\mathbb{M}$ is denoted by $\mathcal{P}(\mathbb{M})$; 
we write $X \sim P_X \in \P{\mathbb{M}}$ for a random variable $X$ in $\mathbb{M}$ whose distribution is $P_X$.
The symbol \scalebox{0.8}{$\eqdis$} denotes equality in distribution. 
We write $Z_n \xasto[n\to\infty] Z$ to denote the almost sure convergence of random variables $\{Z_n\}_{n=1}^\infty$ to $Z$ as $n \to \infty$.

\section{Selection of tuning parameter}\label{sec: SelectionTuningParameter}

The theory of the~\RPD{} was originally developed under the assumption of a fixed parameter $\beta>0$~in~\eqref{eqRegDir}. 
In practice, however, the parameter $\beta$ is taken as the $u$-quantile of $\MAD[\inner{X,V}]$, 
where $V$ is independent of $X$ and supported on~$\SS$; see also~\citet[Section~4]{Bocinec_etal2026}. 
Formally, we make the following definition.

\begin{definition}[Regularization based on quantiles]
Let $u \in [0,1)$ and
$V$ be a random element supported on $\SS$.\footnote{By the support of a random element $V$ we mean the smallest closed set $A\subseteq \HH$ such that $\pr(V\in A)=1$.} 
Consider the univariate random variable $\MAD[\inner{X,V}](\omega) = \MAD[\inner{X,V(\omega)}]$, $\omega\in\Omega$, and define
\begin{equation}\label{eqBetaU}
    \beta(u) = q_u[\MAD[\inner{X,V}]]
\end{equation}
be the $u$-quantile\footnote{Let $Z$ be a non-negative real-valued random variable with distribution function $F$. The $u$-quantile of $Z$ is defined as $q_u=\inf\set{x\geq 0 \st F(x)\geq u}$. In particular, we define $q_0=0$.}
of $\MAD[\inner{X,V}]$. 
\footnote{Note that we treat $\MAD[\inner{X,V}]$ in~\eqref{eqBetaU} as a random variable with respect to the random element $V$ (considering $X$ as fixed), and therefore, when computing quantiles, we take into account the randomness induced by $V$. We will adhere to this convention throughout the paper, whenever writing $\MAD[\inner{X,V}]$.} Using $\beta(u)$, we set 
\begin{equation}\label{eqRPDu}
    \VV_{(u)} = \VV_{\beta(u)}\quad\text{and}\quad D_{(u)}(\cdot; P_X) = D_{\beta(u)}(\cdot; P_X).
\end{equation}
\end{definition}

The depth~$D_{(u)}$ is the \RPD{} based on a regularization by the $u$-quantile of $\MAD[\inner{X,V}]$ that is used in practice. 
Intuitively, in~\eqref{eqRPDu} we restrict attention to the $(1-u)$-fraction of directions $v\in\SS$ with the largest values of $\MAD[\inner{X,v}]$ in~\eqref{eqRegDir}.
Smaller values of $u$ correspond to weaker regularization, with $\VV_{(0)}=\SS$ in the boundary case. 
Note that the regularization in the definition of~\RPD{} can be viewed as a form of effective dimension reduction. 
Choosing the regularization parameter $\beta(u)$ as a quantile, therefore, allows us to quantify the degree of effective dimension reduction by means of a scalar $u \in [0,1)$.

The choice of $\beta$ as a $u$-quantile guarantees that $\VV_{(u)}\neq\emptyset$. 
As shown in the following result, for $u>0$, we obtain that $\beta(u)>0$ under mild assumptions.
Recall that a \emph{hyperplane} in $\HH$ is a set of the form $\left\{ x \in \HH \colon \inner{x, v} = c \right\}$ for some $v \in \SS$ and $c \in \R$.

\begin{lemma}\label{lemma: betaPos}
    Let $X \sim P_X \in \P{\HH}$ and $u\in(0,1)$. Suppose that $V\sim\nu\in\P{\SS}$ is smooth in the sense that $\nu(H) = 0$ for every hyperplane $H \subset \HH$ passing through the origin. If $P_X$ has no atom with probability at least $1/2$, i.e.
    \begin{equation}\label{eq: atomAssumpt}
        \sup_{x\in\HH}P_X(\set{x}) < \frac{1}{2},
    \end{equation}
    then $\beta(u)>0$.
\end{lemma}

The proof of Lemma~\ref{lemma: betaPos}, as well as all the other proofs from this paper, is provided in Section~\ref{Supplement: Proofs} of the online supplementary material.
Throughout this paper, we assume that the conditions of Lemma~\ref{lemma: betaPos} are satisfied.
These assumptions are natural, since one may choose $\nu$ as the distribution of $Y/\norm{Y}$, where $Y$ is (any) centered, non-degenerate Gaussian element in $\HH$.
Moreover, assumption~\eqref{eq: atomAssumpt} on $P_X$ is standard and is typically satisfied.

Note that the choice of the regularization parameter as a quantile does not alter most of the main theoretical results of~\citet{Bocinec_etal2026}.
The departures from the theory are summarized below:

\begin{itemize}
    \item \textbf{Non-degeneracy.} The non-degeneracy of the~\RPD{} from~\citet[Theorem~2]{Bocinec_etal2026} holds provided that $u\in [0,1)$ is such that $\beta(u)>0$. Lemma~\ref{lemma: betaPos} above shows that this is satisfied for $u>0$ under mild assumptions.
    
    \item \textbf{Orthogonal invariance.} The orthogonal invariance property of the~\RPD{} from~\citet[Theorem~5]{Bocinec_etal2026} no longer holds. 
    If $\mathcal{S}\colon \HH \to \HH$ is a unitary operator,\footnote{A bounded linear operator $\mathcal{S}\colon \HH \to \HH$ is called \emph{unitary} if and only if $\mathcal{S}\mathcal{S}^*=\mathcal{S}^*\mathcal{S}=\mathcal{I}$, where $\mathcal{I}$ denotes the identity operator and $\mathcal{S}^*$ the adjoint operator of~$\mathcal{S}$.} then, in general, the random variable $\MAD[\inner{\mathcal{S}X,V}](\omega)=\MAD[\inner{X,\mathcal{S}^*V}](\omega)$ has a different distribution from $\MAD[\inner{X,V}](\omega)$. Consequently, a unitary transformation $\mathcal{S}$ leads to a different value of $\beta(u)$. 
    This leads to the following result.
    
    \begin{theorem}\label{thm: invariance}
        Let $\mathcal{T}x = \mathcal{S}x + e$ for all $x \in \HH$, where $\mathcal{S} \colon \HH \to \HH$ is a unitary operator and $e \in \HH$ is a fixed vector. Then $D_{(u)}(\mathcal{T}x; P_{\mathcal{T}X}) = D'_{(u)}(x; P_X)$ for all $x\in\HH$. Here, $D'_{(u)}$ denotes the~\RPD{} where $\beta(u)$ in~\eqref{eqBetaU} is replaced by $q_u[\MAD[\inner{X,\mathcal{S}^*V}]]$.
    \end{theorem}
    
    A natural idea is to consider a unitarily invariant $V$, i.e., $\mathcal{S}V \eqdis V$ for all unitary $\mathcal{S}$. Unitarily invariant non-degenerate distributions, however, do not exist when $\dim(\HH)=\infty$~\citep[p.~1]{Kuo1975}. 
    In practical applications, this issue does not arise because the observed functions are represented in a discretized form with finitely many values. Consequently, the dimension is finite, and $V$ can naturally be chosen uniformly on the unit sphere of the corresponding finite-dimensional space.
    
    Although orthogonal invariance is lost, an immediate consequence of Theorem~\ref{thm: invariance} is that the~\RPD{} is shift invariant:
    \begin{equation}\label{eq: shiftInvariance}
        D_{(u)}(x+e; P_{X+e}) = D_{(u)}(x; P_X) \quad \text{for all } e \in \HH \text{ and } x \in \HH.
    \end{equation}
        
    \item \textbf{Consistency.} The quantile $\beta(u)$ depends on the underlying distribution $P_X$. Hence, it must be approximated from a random sample. The sample version consistency of \RPD{} remains to hold under mild assumptions; we elaborate on this in~Theorem~\ref{thm: ConsistRegQuantile} in Section~\ref{subsec: consistency} below.
    
    \item \textbf{Robustness of the induced median.} Theorem~10 in~\citet{Bocinec_etal2026} shows that for fixed $\beta>0$, the breakdown point of the median induced by the~\RPD{} equals~$1/2$. The case where $\beta=\beta(u)$ is chosen as in~\eqref{eqBetaU} is discussed in Theorem~\ref{thm: BP} in Section~\ref{subsec: robustness} below.
\end{itemize}

\subsection{Consistency with \texorpdfstring{$\beta$}{beta} chosen as a quantile} \label{subsec: consistency}

For $X \sim P_X \in \P{\HH}$, consider a random sample $\set{X_i}_{i=1}^n$ of size $n$, where $X_i \eqdis X$. Denote by $\widehat{P}_n$ the corresponding empirical distribution, $\widehat{P}_n = n^{-1}\sum_{i=1}^n\delta_{X_i}$, where $\delta_x$ is the Dirac measure at $x \in \HH$. For $v \in \SS$, define
$\widehat{\med}\left[\set{\inner{X_i, v}}_{i=1}^n\right]$ and
$\widehat{\MAD}\left[\set{\inner{X_i, v}}_{i=1}^n\right]$ 
as the sample median and sample MAD based on the projections $\set{\inner{X_i,v}}_{i=1}^n$, respectively.\footnote{For a sample $\set{Z_i}_{i=1}^n \subset \R$ with order statistics $Z_{(1)} < \cdots < Z_{(n)}$, the sample median is defined as $\widehat{\med}\left[ \set{Z_i}_{i=1}^n\right] = Z_{((n+1)/2)}$ if $n$ is odd and $\widehat{\med}\left[ \set{Z_i}_{i=1}^n\right] = (Z_{(n/2)} + Z_{(n/2+1)})/2$ otherwise. The sample MAD is then defined by $\widehat{\MAD}\left[\set{Z_i}_{i=1}^n\right] = \widehat{\med}\left[\set{\abs{Z_i - \widehat{\med}\left[\set{Z_j}_{j=1}^n\right]}}_{i=1}^n\right]$.}

With this notation, the regularization parameter $\beta(u)$ is estimated by its sample analog, that is, we define $\widehat{\beta}_n(u)$ as the $u$-quantile of $\widehat{\MAD}\left[\set{\inner{X_i, V}}_{i=1}^n\right]$. Note that this quantile is taken conditionally on $\set{X_i}_{i=1}^n$, thus $\widehat{\beta}_n(u)$ remains random since it depends on the realization of the sample $\set{X_i}_{i=1}^n$. As shown in the following lemma, under mild assumptions, $\widehat{\beta}_n(u)$ is a consistent estimator of $\beta(u)$.

\begin{lemma}\label{lemma: betaConv}
    Let $X \sim P_X \in \P{\HH}$ have a contiguous support.\footnote{We say that a random variable $X \sim P_X \in \P{\HH}$ has contiguous support \citep{Kong_Zuo2010} if for each $v \in \SS$, the support of the real random variable $\inner{X,v}$ is connected.}
    Further assume that $\beta(u)$ is a point of continuity of the distribution function of $\MAD[\inner{X,V}]$. Then $\widehat{\beta}_n(u)\xasto[n\to\infty]\beta(u)$.
\end{lemma}

The sample version of the~\RPD{} is defined as
\begin{equation}
    D_{(u)}(x; \widehat{P}_n)=\inf_{v\in\widehat{\VV}_{n(u)}}\left(1+O_v(x; \widehat{P}_n)\right)^{-1}
    =\left(1+\sup_{v\in\widehat{\VV}_{n(u)}}\frac{\abs{\inner{x,v}-\widehat{\med}\left[\set{\inner{X_i, v}}_{i=1}^n\right]}}{\widehat{\MAD}\left[\set{\inner{X_i, v}}_{i=1}^n\right]}\right)^{-1}, \label{eqRPDsample}
\end{equation}
where 
\begin{equation*}
    \widehat{\VV}_{n(u)}
    = \set{v \in \SS \st \widehat{\MAD}\left[\set{\inner{X_i, v}}_{i=1}^n\right] \geq \widehat{\beta}_n(u)}.
\end{equation*}

Next, we show that defining $\beta(u)$ as a quantile does not affect the consistency of~\RPD{}.
To this end, we introduce the maximal outlyingness function $\Gamma_x$ at $x \in \HH$ over $\VV_t$ (from~\eqref{eqRPDu} with $\beta(u)=t$) as
\begin{align}\label{eqMaxOut}
    \Gamma_x(t) = \sup_{v \in \VV_{t}} O_v(x; P_X), \quad t \in [0, \infty).
\end{align}

\begin{theorem}\label{thm: ConsistRegQuantile}
    Suppose that $X \sim P_X \in \P{\HH}$ has contiguous support and let $u \in [0,1)$ be such that~$\beta(u) > 0$. 
    Assume that $\beta(u)$ is a point of continuity of the distribution function of $\MAD[\inner{X,V}]$.
    Furthermore, assume that $\mathcal{F} \subset \HH$ is a bounded set for which the family $\set{\Gamma_x : x \in \mathcal{F}}$ of maximal outlyingness functions~\eqref{eqMaxOut} is equicontinuous at $\beta$, i.e.,
    \begin{equation}\label{consAssump}
        \lim_{t \to \beta} \sup_{x \in \mathcal{F}} \abs{\Gamma_x(t) - \Gamma_x(\beta)} = 0.
    \end{equation}
    Then, the sample \RPD{} in~\eqref{eqRPDsample} is uniformly consistent for the population \RPD{} in~\eqref{eqRPDu} over $\mathcal{F}$, namely
    \begin{equation*}
        \sup_{x \in \mathcal{F}} \abs{D_{(u)}(x; \widehat{P}_n) - D_{(u)}(x; P_X)} \xasto[n \to \infty] 0.
    \end{equation*}
\end{theorem}

Theorem~\ref{thm: ConsistRegQuantile} states that the consistency of the sample version $D_{(u)}(\cdot, \widehat{P}_n)$ is essentially equivalent to condition~\eqref{consAssump}, that is, to the uniform continuity of the maximal outlyingness function $\Gamma_x(\cdot)$. This continuity is, in fact, required already for the consistency of the standard \RPD{} with a fixed parameter $\beta$. It can be shown that~\eqref{consAssump} is satisfied, for instance, in the case of elliptically symmetric distributions on~$\HH$~\citep[Theorems~8 and 9]{Bocinec_etal2026}.

\subsection{Robustness of the induced median with \texorpdfstring{$\beta$}{beta} chosen as a quantile}\label{subsec: robustness}

Robustness analysis plays a central role in the study of depth-based functionals.
An estimator is said to be robust 
if small perturbations in the data distribution do not lead to substantial changes in its value.
We assess robustness through the breakdown point, which represents the maximum fraction of contamination that can be introduced before the estimator becomes arbitrarily unreliable. 
It is a crucial measure of the robustness of any statistical functional \citep{Huber_Ronchetti2009}.

In this section, we study the robustness of the \RPD{}-induced median using a modified breakdown point designed for infinite-dimensional spaces.
For $X \sim P_X \in \P{\HH}$ and $u \in (0,1)$,
the \RPD{} median of $P_X$ is defined as
\begin{equation} \label{eq: RPD median}
    \theta(P_X) = \theta_{(u)}(P_X)
    = \argmax\set{D_{(u)}(\theta; P_X)\st \theta \in \Clo{\spn{\VV_{(u)}}}}.
\end{equation}
We restrict our focus to points $\theta \in \Clo{\spn{\VV_{(u)}}}$ because all relevant information about $P_X$ in $D_{(u)}(\cdot; P_X)$ is contained in this subspace. 
More precisely, $D_{(u)}(x; P_X) = D_{(u)}(\widehat{x}; P_X)$ for any $x \in \HH$ and its orthogonal projection $\widehat{x}$ onto $\Clo{\spn{\VV_{(u)}}}$~\citep[Theorem~5d]{Bocinec_etal2026}.

To define our breakdown point notion, 
we first introduce contaminated distributions and their related quantities.
Assuming a smooth (as in Lemma~\ref{lemma: betaPos}) distribution $V \sim \nu \in \P{\SS}$ and an atomless $P_X$,
for $\varepsilon \in [0,1]$ and $Q \in \P{\HH}$, 
we define the $\varepsilon$-contaminated distribution 
$P_{(Q, \varepsilon)} \in \P{\HH}$ by $Q$ as
\begin{equation*}
    P_{(Q, \varepsilon)}(A)
    = (1-\varepsilon)P_X(A) + \varepsilon\, Q(A),
    \quad \text{for all Borel sets } A \subseteq \HH.
\end{equation*}
The contaminated regularization parameter
$\beta_{(Q,\varepsilon)}(u)$ is the $u$-quantile of
\begin{equation*}
    \MAD[\inner{X', V}], \quad \text{where } X' \sim P_{(Q, \varepsilon)}.
\end{equation*}
Under our assumptions, Lemma~\ref{lemma: betaPos} ensures that $\beta(u) > 0$, 
and for $\varepsilon < 1/2$, $\beta_{(Q,\varepsilon)}(u) > 0$ for any $Q \in \P{\HH}$.
The subsequent contaminated direction set 
$\VV_{(u)}(Q, \varepsilon)$ is defined by
\begin{equation*}
    \VV_{(u)}(Q, \varepsilon)
    = \set{v \in \SS \st \MAD[\inner{X', v}] \geq \beta_{(Q,\varepsilon)}(u)},
    \quad \text{where } X' \sim P_{(Q, \varepsilon)},
\end{equation*}
which is non-empty for any $Q \in \P{\HH}$ by construction.
Finally, tailored to our infinite-dimensional setup, the breakdown point of the \RPD{} median $\theta = \theta(P_X)$
is defined as\footnote{
    Since $\VV_{(u)} \neq \emptyset$ and $\VV_{(u)}(Q,\varepsilon) \neq \emptyset$, 
    the sets of medians of $P_X$ and $P_{(Q,\varepsilon)}$ are non-empty, see~\citet[Theorem~6]{Bocinec_etal2026}. 
    The particular choice of representatives $\theta(P_X)$ and $\theta(P_{(Q,\varepsilon)})$ does not affect our results.}
\begin{equation} \label{eq: BP}
    \varepsilon^*_{(u)}(\theta; P_X)
    = \inf\set{\varepsilon \in [0,1] \st 
        \sup_{Q \in \P{\HH}} 
        \sup_{v \in \VV_{(u)}(Q, \varepsilon)}
        \abs{\inner{\theta(P_{(Q,\varepsilon)}) - \theta(P_X), v}} = \infty}.
\end{equation}

The distance $\left\Vert \theta(P_{(Q,\varepsilon)}) - \theta(P_X) \right\Vert$ originally used in the classical breakdown point~\citep{Huber_Ronchetti2009} is replaced in~\eqref{eq: BP} by the maximum inner product over the contaminated direction set $\VV_{(u)}(Q, \varepsilon)$.
The intuition for this change comes from
the Cauchy-Schwarz inequality \citep[Theorem~5.3.3]{Dudley2002}, which provides
\begin{equation} \label{eq: Cauchy Schwarz}
    \sup_{v \in \SS}\abs{\inner{\theta(P_{(Q,\varepsilon)}) - \theta(P_X), v}} = \left\Vert \theta(P_{(Q,\varepsilon)}) - \theta(P_X) \right\Vert.
\end{equation}
Since the entire Hilbert sphere $\SS$ is too rich,
the supremum in \eqref{eq: BP} is taken over a smaller set $\VV_{(u)}(Q, \varepsilon)$.

The breakdown point analysis for the \RPD{} median with $\beta(u)$ chosen as a quantile depends on the smoothness of both $P_X$ and $V \sim \nu$.  
Let $F_v$ denote the distribution function of $\inner{X, v}$ for $v \in \SS$, and define its modulus of continuity
\begin{equation*}
    \omega_v(\delta) = \sup_{\abs{t_1 - t_2} \leq \delta} \abs{F_v(t_1) - F_v(t_2)}, \quad \delta\geq 0.
\end{equation*}
For $\varepsilon \in (0, 1/2)$, consider the technical condition
\begin{equation}\label{technAssBP} 
     \exists\gamma>0\st \nu\left(\set{v \in \SS \st 
        \omega_v(\gamma) < \frac{1/2 - \varepsilon}{1 - \varepsilon}}\right) 
        \geq 1 - u.
\end{equation}
This condition essentially states that there exist sufficiently many directions $v \in \SS$ (in terms of the measure $\nu$) for which the projection $\inner{X, v}$ is sufficiently scattered.
Note that the function $\omega_v(\cdot)$ is non-decreasing.
Consequently, if~\eqref{technAssBP} holds for some $\gamma$, it also holds for all $0<\gamma'\leq \gamma$.

\begin{theorem}\label{thm: BP}
    Consider $X \sim P_X \in \P{\HH}$ with no atoms, $u \in (0,1)$, and assume that $V \sim \nu \in \P{\SS}$ is smooth as in Lemma~\ref{lemma: betaPos}. 
    If assumption~\eqref{technAssBP} holds for some $\varepsilon \in (0, 1/2)$, then
    \begin{equation*}
        \sup_{Q \in \P{\HH}} 
        \sup_{v \in \VV_{(u)}(Q, \varepsilon)}
        \abs{\inner{\theta(P_{(Q,\varepsilon)}) - \theta(P_X), v}} < \infty.
    \end{equation*}
    Consequently, if~\eqref{technAssBP} holds for every $\varepsilon \in (0,1/2)$, 
    then $\varepsilon^*_{(u)}(\theta; P_X) = 1/2$.
\end{theorem}

This result is remarkable, as it shows that under quite mild conditions, the infinite-dimensional \RPD{} median achieves the highest possible breakdown point of $1/2$, attainable by location-equivariant estimators \citep{Huber_Ronchetti2009}.
Note that the \RPD{} median inherently satisfies this equivariance requirement due to the shift invariance property~\eqref{eq: shiftInvariance}.

A natural question is when the technical condition~\eqref{technAssBP} holds. Consider, for instance, Gaussian distributions in~$\HH$.
That is, assume that $X$ has a centered Gaussian distribution with a positive definite covariance operator $\Sigma \colon \HH \to \HH$~\citep[Definition~7.2.3]{Hsing_Eubank2015}.
This implies that for any $v\in\SS$, the random variable $\inner{X,v}$ follows a univariate Gaussian distribution with zero mean and variance $\inner{v, \Sigma v}$. 
In this setting, we can bound the modulus of continuity
\begin{equation*}
    \omega_v(\delta) = \sup_{\abs{t_1 - t_2} \leq \delta} \abs{F_v(t_1) - F_v(t_2)} \leq \delta\,\sup_{t\in \R}f_v(t) \leq \frac{\delta}{\sqrt{2\pi \inner{v, \Sigma v}}},
\end{equation*}
where $f_v$ denotes the probability density function of $\inner{X,v}$. 
It follows that if there exists $\gamma>0$ such that 
\begin{equation}\label{eq: BPGauss}
\begin{aligned}
    \nu\left(\set{v \in \SS \st 
    \omega_v(\gamma) < \frac{1/2 - \varepsilon}{1 - \varepsilon}}\right)\geq\nu\left(\set{v \in \SS \st 
    \frac{\gamma}{\sqrt{2\pi \inner{v, \Sigma v}}} < \frac{1/2 - \varepsilon}{1 - \varepsilon}}\right) \\=
    \nu\left(\set{v \in \SS \st 
    \inner{v, \Sigma v} > \left(\frac{(1 - \varepsilon)\gamma}{(1/2 - \varepsilon)\sqrt{2\pi}}\right)^2}\right)
    \geq 1 - u,
\end{aligned}
\end{equation}
then~\eqref{technAssBP} is satisfied. 
Moreover, note that for $V \sim \nu$, the random variable $\inner{V, \Sigma V}(\omega)=\inner{V(\omega), \Sigma V(\omega)}$, $\omega \in \Omega$, is strictly positive. 
Consequently, the $u$-quantile of $\inner{V, \Sigma V}$, for any $u>0$, is always positive, implying that for any $\varepsilon \in (0, 1/2)$, there exists a $\gamma>0$ such that~\eqref{eq: BPGauss} holds. 
Theorem~\ref{thm: BP} then implies that $\varepsilon^*_{(u)}(\theta; P_X) = 1/2$ for any Gaussian distribution $P_X$ with a positive definite covariance operator.
The same argument applies to elliptically symmetric distributions in~$\HH$~\citep{Boente_etal2014} with a positive definite scatter operator, for which all projections $\inner{X,v}$, $v\in\SS$, possess bounded densities. 

Note that the condition~\eqref{technAssBP} is more strict for smaller values of $u>0$.
This indicates that stronger regularization leads to more robust estimators of location.
This is illustrated in our numerical examples in Section~\ref{Supplement: LocationEstimation} of the online supplementary material.

\section{Computation of~\RPD{}}\label{sec: Computation}

The definition of the \RPD{} in~\eqref{eqRPD} requires computing an infimum over the set of directions $\VV_{(u)}\subseteq\SS$. 
Although the \RPD{} can be interpreted as a conditional version of the classical projection depth \citep{Zuo2003}, directly computing this infimum in a (typically) infinite-dimensional Hilbert space $\HH$ poses a significant challenge.
While sophisticated algorithms exist for projection depth in finite-dimensional settings \citep[e.g.,][]{Liu_Zuo2014b, Shao_etal2022}, they do not scale well to functional data.

To address this issue, we employ a stochastic approximation algorithm inspired by \citet[Section~6]{Dyckerhoff2004}. 
Instead of evaluating the infimum over the entire set $\VV_{(u)}$, we minimize the outlyingness over a finite set of directions drawn independently from a probability distribution supported on $\VV_{(u)}$. 
First, we sample $L\in\N$ independent directions $\set{u_1,\ldots,u_L} \in \SS$ from $\nu$ and approximate $\beta(u)$ with the sample $u$-quantile $\beta^{(L)}(u)$ of $\set{\MAD[\inner{X,u_i}]}_{i=1}^L$. 
Next, we sample $M\in\N$ independent directions $\set{v_1, \dots, v_M} \in \SS$ from the distribution $\nu$ conditioned on $\VV^{(L)}_{(u)}=\VV_{\beta^{(L)}(u)}$.
Finally, using the approximate regularization $\beta^{(L)}(u)$, and directions $\set{v_1, \dots, v_M}$,
The approximated random \RPD{} is defined as
\begin{equation} \label{eq: random RPD}
    D_{(u)}^{(M)}(x; P_X) = \min_{m=1,\dots,M} \left(1 + O_{v_m}(x; P_X) \right)^{-1}.
\end{equation}
This approach reduces the complex optimization problem to a straightforward evaluation of univariate outlyingness functions. 
The validity of this approximation is ensured by the following consistency result.

\begin{theorem}[Consistency of the random \RPD{}] \label{thm: consistency}
Consider $x \in \HH$ and $P_X\in\P{\HH}$ with either finite or contiguous support. Let $\nu \in \P{\SS}$ be such that its support is the whole set $\SS$. Assume that $\beta(u)$ in~\eqref{eqBetaU} is uniquely defined, i.e., the distribution function of $\MAD[\inner{X,V}]$ is strictly increasing at $\beta(u)$. Additionally, assume that $\Gamma_x(\cdot)$ in~\eqref{eqMaxOut} is continuous at $\beta(u)$. Then, the approximated random~\RPD{}~\eqref{eq: random RPD} converges almost surely to the theoretical \RPD{} as the number of directions increases, that is $\lim_{M, L\to \infty} D_{(u)}^{(M)}(x; P_X) = D_{(u)}(x; P_X)$ a.s.
\end{theorem}

For practical implementation, we leverage the computational efficiency of \textsf{C++} integrated into the \textsf{R} environment via the \textsf{RcppArmadillo} package \citep{RcppArmadillo}. 
The algorithm underlying~\eqref{eq: random RPD} is available in the software package \textsf{RPD}, freely available online.\footnote{\url{https://github.com/NagyStanislav/RPD}}

The choice of the numbers of directions $L$ and $M$ presents a trade-off between computational efficiency and approximation accuracy. 
While larger values improve the quality of the approximation, they also increase the computational cost. 
This issue becomes particularly pronounced for functional data with high effective dimension, which may make the stochastic approximation computationally intensive.
Consequently, the main drawback of the \RPD{} is that in practice, it is computed approximately. 
Unless stated otherwise, in all subsequent computations we set $L = 1\,000$ and $M = 10\,000$.
These choices ensure that, even for datasets of moderate size (in the order of several thousand observations) and moderate discretization resolution (around $100$ grid points), the resulting computations can be performed within at most a few seconds.

In practice, the approximation of the regularization parameter $\beta^{(L)}(u)$ is less critical than the approximation of the depth itself.
Therefore, it is generally more important to choose a sufficiently large value of $M$ to ensure an accurate evaluation of the \RPD{}.  

Figure~\ref{fig: convergence rate} illustrates the convergence behavior of the stochastic approximation for representative functional observations.
Empirical evidence indicates that the stochastic approximation converges more rapidly for outlying functions than for more central ones.
This phenomenon can be attributed to the fact that, for outliers, it is typically easier to identify a projection direction in which the function appears extreme, whereas for more central functions the search for the most extreme direction is inherently more challenging.

\begin{figure}
    \centering
    \begin{subfigure}{0.33\textwidth}
        \centering
        \includegraphics[width=\linewidth]{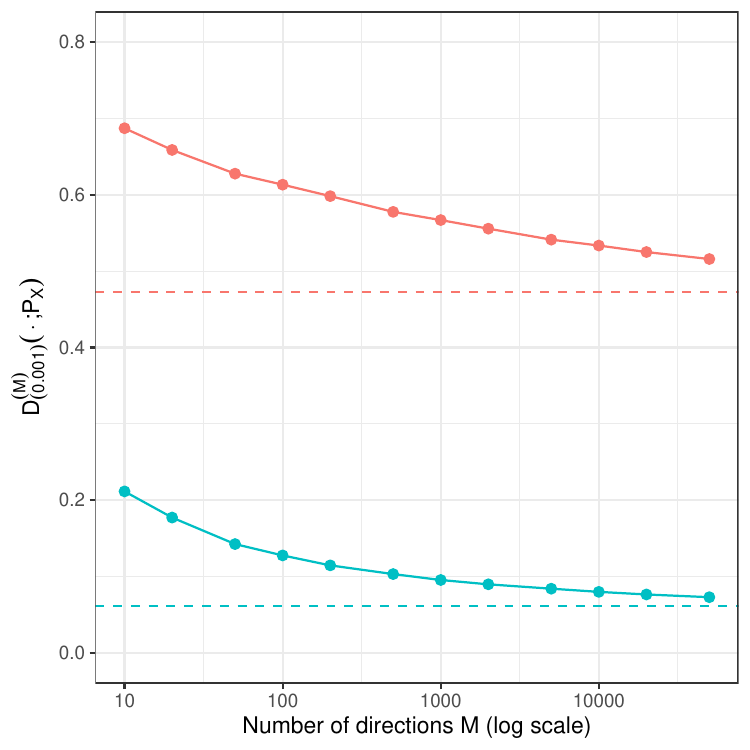}
    \end{subfigure}
    \begin{subfigure}{0.33\textwidth}
        \centering
        \includegraphics[width=\linewidth]{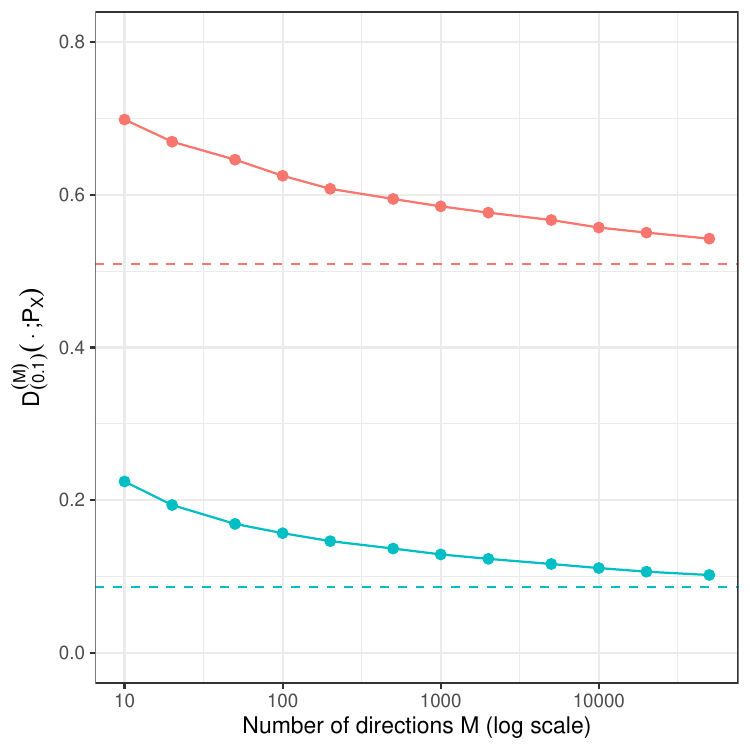}
    \end{subfigure}

    \caption{Convergence rate plots of $D_{(u)}^{(M)}$ with $u = 0.001$ (left) and $u = 0.1$ (right). The depth is evaluated for a central function $x(t)=0$, $t \in [0,1]$ (red), and for a peripheral function $x(t)=1.5\,\sin(2\pi t)$, $t \in [0,1]$ (blue). The reference dataset consists of $n=50$ independent centered Gaussian trajectories with covariance function $\Sigma(s,t)=\exp\left(-(s-t)^2/0.32\right)$, discretized on a grid of $T=51$ points. The plots show $D_{(u)}^{(M)}(x;\widehat{P}_n)$ as a function of $M$, averaged over $100$ Monte Carlo simulations, with dashed horizontal lines corresponding to $M=10^7$. The stochastic approximation of \RPD{} converges more slowly for central functions than for peripheral ones.}
    \label{fig: convergence rate}
\end{figure}

A related phenomenon was studied by~\citet{Briend_etal2025} in the context of the random approximation of the classical multivariate halfspace (Tukey) depth.
The authors show that, for points whose depth is bounded away from both $0$ and $1/2$, the random halfspace depth provides an inefficient approximation of the exact halfspace depth.
An analogous effect is observed in the present setting for the randomized~\RPD{}.

\section{Simulations}\label{sec: Simulations}

We assess the practical utility of \RPD{} in three statistical tasks:
(i) outlier detection, (ii) supervised classification, and (iii) hypothesis testing for the equality of two distributions.
Additional simulation results regarding robust location estimation are in Section~\ref{Supplement: LocationEstimation} of the online supplementary material. 
Throughout the paper we compare \RPD{} with the following functional depths: 
\begin{itemize}
    \item \RHD{}: the regularized halfspace depth~\citep{Yeon_etal2025} with a relatively high and fixed dimension $J=6$. This choice is motivated by the simulation results reported in~\citet{Bocinec_etal2026}, where \RHD{} with $J=6$ was observed to perform better. We also considered the default choice of $J$, which led to inferior performance, and is therefore not reported here;
    \item \FD{}: the integrated halfspace depth \citep{Fraiman_Muniz2001, Nagy_etal2016};
    \item \MBD{}: the modified band depth~\citep{Lopez_Romo2009}; and
    \item \ID{}: the infimal halfspace depth \citep{Mosler2013}, representing infimal/extremal depths from the literature \citep{Narisetty_Nair2016}.
\end{itemize}
All our simulations are performed in the statistical software \textsf{R}. 
We employ the \textsf{RPD} package for \RPD{}, the \textsf{RHD} package for \RHD{},\footnote{Available online at \url{https://github.com/luckyhm1928/RHD}.} and \textsf{ddalpha} \citep{ddalpha} for \FD{}, \ID{}, and \MBD{}. 
All results below are estimated based on Monte Carlo simulations with $1\,000$ replications.

For \RPD{}, the regularization parameter is chosen based on the $u$-quantile as defined in~\eqref{eqBetaU}, and to ensure consistency with the definition used in this paper (lower values of $u$ mean weaker regularization), we use the $(1-u)$-quantile for the \RHD{} regularization parameter. 
This will be denoted as $\RPD{}_u$ and $\RHD{}_u$. 
Several values of $u$ are considered and their impact is discussed. 

We consider the separable Hilbert space $\HH = L_2([0,1])$ of square-integrable functions on $[0,1]$, equipped with the standard inner product $\inner{x, y} = \int_0^1 x(t)y(t) \dd t$ and the induced norm $\norm{x} = \left(\int_0^1 x(t)^2 \dd t\right)^{1/2}$ for $x, y \in \HH$. 
In practice, each function $x \in \HH$ is observed on a discrete grid of $T=101$ equidistant points $t_1,\ldots,t_T \in [0,1]$, so that the $L_2$-norm is approximated by its discrete counterpart $\norm{x} \approx \left(\sum_{k=1}^T x(t_k)^2/T\right)^{1/2}$.

\subsection{Outlier detection}\label{sec: OutlierDetection}\label{sec:4.1}

Recall that the task of outlier detection can be viewed essentially as a binary classification problem, where the outlier class is typically very small or potentially empty. 
Our classification procedure consists of two steps: first, we order the data points according to their depth values, and subsequently, we select a threshold below which points are classified as outliers.
We divide our performance analysis into two parts. Initially, we evaluate the discriminative power of the various depth measures introduced above for outlier detection. 
Subsequently, by introducing a data-driven threshold selection mechanism, we propose a novel~\RPD{}-based outlier detection method and compare its performance with other established techniques.

\subsubsection{Discriminative power of~\RPD{}}\label{subsec: discriminative}
For outlier detection, we consider $n$ independent random functions $\set{X_i(t)}_{i=1}^n$ defined on $t \in [0,1]$, among which $m$ are outliers. This corresponds to a contamination level of $\varepsilon = m/n$. We examine a wide spectrum of functional outliers. 
In the following models, the error terms $\set{e_i(t)}_{i=1}^n$ are independent realizations of a centered Gaussian process with the covariance function $\Sigma(s, t) = \lambda \exp\left( - \left\vert s - t \right\vert/\lambda \right)$ for $s, t \in [0,1]$ and a scale parameter $\lambda > 0$. 

\begin{enumerate}[label=\textbf{Model~(D$_\arabic*$)}, ref=(D$_\arabic*$), leftmargin=*]
    \item\label{ModelD1} \emph{Magnitude outliers:} The non-contaminated observations are generated as $X_i(t) = 4t + e_i(t), \, t \in [0,1]$, where $e_i(t)$ is taken with $\lambda = 1$. The outlying observations follow the model $X_i(t) = 4t + 8 k_i + e_i(t)$, where $k_i \in \set{-1,1}$ is independent of $e_i$ with $\pr(k_i=1)=\pr(k_i=-1)=0.5$.
    
    \item\label{ModelD2} \emph{Location difference model:} The clean observations are generated according to $X_i(t) = 30\,t(1-t)^{1.5} + e_i(t), \, t \in [0,1]$, with $\lambda = 0.3$. The outlying observations are $X_i(t) = 30\,t^{1.5}(1-t) + e_i(t)$.

    \item\label{ModelD3} \emph{Covariance structure difference model:} The clean observations are $X_i(t) = 4t + e_i(t), \, t \in [0,1]$, with $\lambda = 1$. The outlying observations follow the model $X_i(t) = 4t + f_i(t)$, where $\set{f_i(t)}_{i=1}^m$ are independent centered Gaussian processes with covariance function $\Sigma(s,t) = 5\exp\left(-2\abs{t-s}^{0.5}\right),\, s,t \in [0,1]$.

    \item\label{ModelD4} \emph{Shifted periodic outliers:} The clean observations are $X_i(t) = 2\sin(15\pi t) + e_i(t), \, t \in [0,1]$, with $\lambda = 1$. The outlying observations are $X_i(t) = 2\sin(15\pi t+2) + e_i(t)$.
    
    \item\label{ModelD5} \emph{Central high-frequency outliers:} The clean observations are $X_i(t) = e_i(t), \, t \in [0,1]$, with $\lambda = 1$. The outlying observations follow the model $X_i(t) = \sin(20(t + \theta_i)\pi) + f_i(t)$, where $\set{\theta_i}_{i=1}^m$ are independent and uniform on $(0.25, 0.5)$ and $\set{f_i(t)}_{i=1}^m$ are independent centered Gaussian processes with covariance $\Sigma(s,t) = 0.1\exp\left(-\abs{t-s}^{0.1}/4\right),\, s,t \in [0,1]$.
    
    \item\label{ModelD6} \emph{Central shape outliers:} Take a basis of orthogonal polynomials $\phi_1, \dots, \phi_6$ on $[0,1]$ of maximum degree $5$ (with $\phi_1$ denoting the constant function of degree $0$, etc.). The non-contaminated observations are generated as $X_i(t) = \sum_{j=1}^6 c^j_i \phi_j(t), \, t \in [0,1]$, where $\set{c_i}_{i=1}^n = \set{(c_i^1, \dots, c_i^6)\tr}_{i=1}^n$ are independent centered Gaussian random vectors with covariance matrix $\Sigma \in \R^{6 \times 6}$ having unit diagonal and all off-diagonal elements equal to $0.95$. The outlying observations follow the model $X_i(t) = \sum_{j=1}^6 b^j_i \phi_j(t), \, t \in [0,1]$, where $\set{b_i}_{i=1}^m = \set{(b_i^1, \dots, b_i^6)\tr}_{i=1}^m$ are independent Gaussian random vectors centered at $(1,1,\dots,1)\tr \in \R^6$ with covariance matrix $\Sigma^{-1}/100$.
\end{enumerate}

Models \ref{ModelD1}--\ref{ModelD5} were inspired by~\citet{Cuevas_etal2007, Arribas_Romo2014, Ojo_etal2022} and \citet{Jimenez-Varon_etal2024}. Model~\ref{ModelD6} was used in~\citet{Bocinec_etal2026}.
Examples of simulated datasets from \ref{ModelD1}--\ref{ModelD6} are in Figure~\ref{fig: OutlierDetDatasets}.

\begin{figure}
    \centering
    \includegraphics[width=0.9\textwidth]{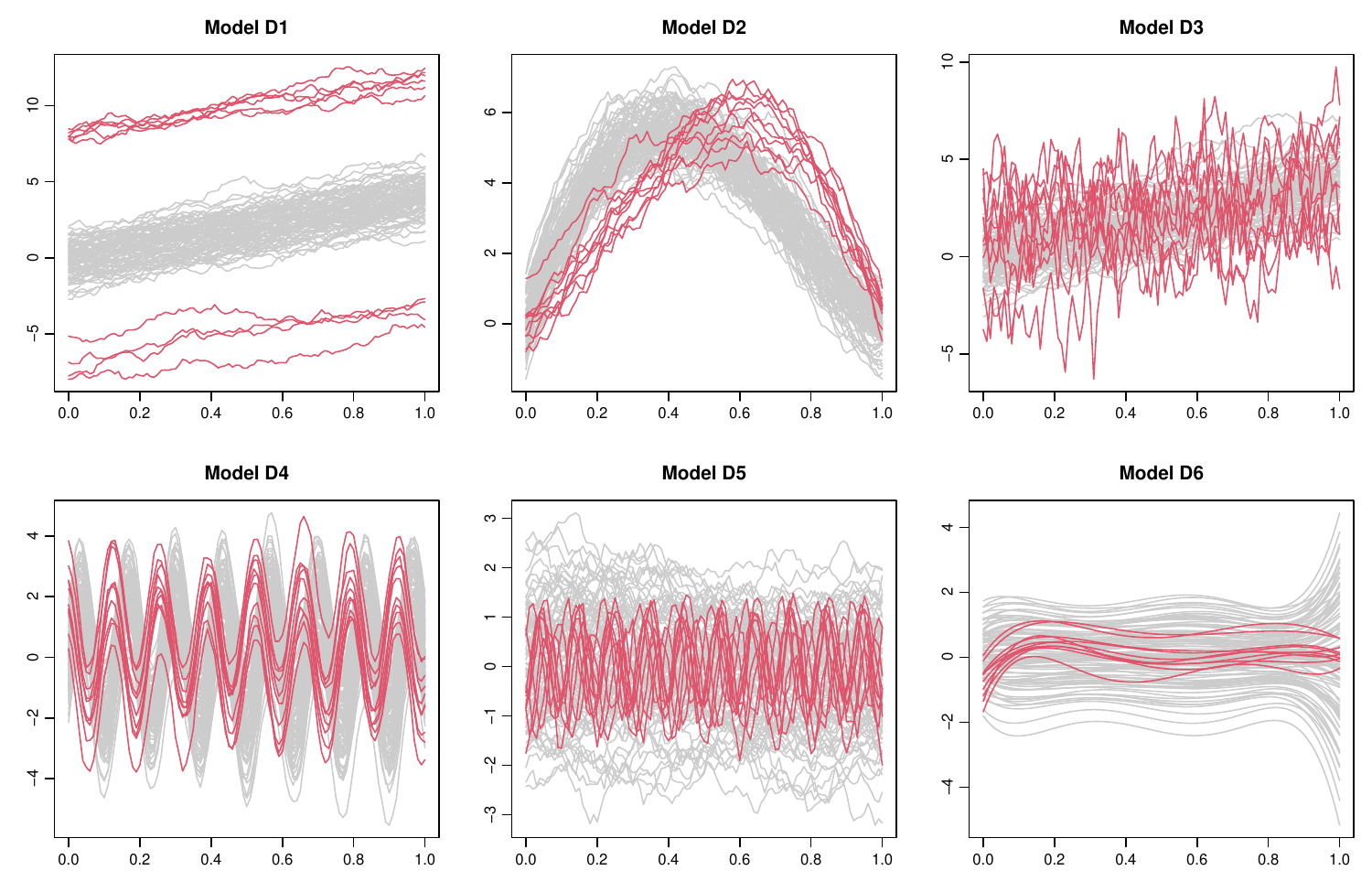}
    \caption{Outlier detection: Simulated datasets generated under the contamination Models~\ref{ModelD1}--\ref{ModelD6} described in Section~\ref{sec: OutlierDetection}. Each panel contains $100$ functions, of which $10$ are contaminated. The non-contaminated functions are gray, while the contaminating functions are highlighted in red.
    \label{fig: OutlierDetDatasets}}
\end{figure}

As the initial step of our analysis, we examine the ability of the depth measures to separate true outliers from non-contaminated curves, without any thresholding. 
The $n$ functions in the pooled sample are ranked based on their depth values, with a normalized rank of $1/n \approx 0$ assigned to the least deep function and a rank of $1$ to the deepest one. 
Functions sharing the same depth value are assigned their average rank.
The discriminative power is then evaluated using the normalized mean rank of the $m$ true outliers. 
The resulting average normalized mean ranks and their corresponding standard deviations, obtained across $1\,000$ independent Monte Carlo replications, are summarized in Table~\ref{tab: OutlierDetection_n500} for $n=500$ and $m\in\set{5,25,50}$, which correspond to contamination levels of $\varepsilon\in\set{0.01,0.05,0.10}$. 
The results for the same study with $n=100$ and $n=1\,000$ are provided in Tables~\ref{tab: OutlierDetection_n100} and~\ref{tab: OutlierDetection_n1000} in Section~\ref{Supplement: Simulations} of the online supplementary material.

\input{Tables_applications/OutlierDetectionRanks_n500}

The results show that \RPD{} consistently delivers near-perfect performance across all considered models.
In the simple Model~\ref{ModelD1}, all depth measures are able to accurately detect the shift in outlying functions. 
Models~\ref{ModelD2}, \ref{ModelD3}, and~\ref{ModelD4} involve location, covariance-structure, and shifted periodic outliers, respectively.
These types of outliers are more challenging to identify; however, they still manifest themselves as pointwise vertical deviations, which explains why the remaining depth measures (not necessarily projection-based) continue to exhibit reasonably good performance.
For instance, in Model~\ref{ModelD2}, \FD{} (and the closely related \MBD{}) performs comparably to \RPD{}. 
In Model~\ref{ModelD3}, good detection results are obtained by \RHD{}, and particularly by \ID{}.
In contrast, for Model~\ref{ModelD4}, all considered depth measures except \RPD{} yield a large average rank of the outlying curves, indicating that the set of lowest-depth curves already contains a substantial number of non-outlying observations.
Even at this stage, \RPD{} maintains consistently strong performance.
The main strength of \RPD{} becomes evident in Models~\ref{ModelD5} and~\ref{ModelD6}. 
In these settings, the outliers tend to remain within the central region of the data cloud while having markedly different shapes.
Here, the advantage of \RPD{} stems from its projection-based nature, whereas depths such as \FD{}, \ID{}, or \MBD{} rely solely on pointwise vertical deviations and therefore fail to adequately capture such shape outlyingness. 
The only competing depth that still gives reasonably good performance is \RHD{}, and only when using a small regularization parameter $u=0.001$.
In contrast, the~\RPD{} achieves very good results in both models~\ref{ModelD5} and~\ref{ModelD6} for $u=0.1$, and performs nearly perfectly for $u=0.001$.

For Models~\ref{ModelD5} and~\ref{ModelD6}, we observe that \RPD{} with regularization parameter $u=0.1$ performs worse than \RPD{} with $u=0.001$. 
This behavior is explained by the fact that stronger regularization leads to reduced sensitivity of the depth. 
Consequently, for the purpose of outlier detection, we generally recommend using a very small value of the regularization parameter $u$, such as $u=0.001$, as employed here.

A visual comparison of the performance of the depth measures under Model~\ref{ModelD5} is also in Figure~\ref{fig: OutlierDetDepths}. 
While \RPD{} successfully detects the central shape outlying curves, \MBD{} primarily identifies curves located at the periphery of the data cloud as outliers, because it attaches more importance to location (horizontal shift) than severe differences in shape. 
Analogous figures for the other models are provided in Figures~\ref{fig: OutlierDetDepths_D1}--\ref{fig: OutlierDetDepths_D6} in Section~\ref{Supplement: Simulations} of the online supplementary material.

\begin{figure}
    \centering
    \includegraphics[width=0.9\textwidth]{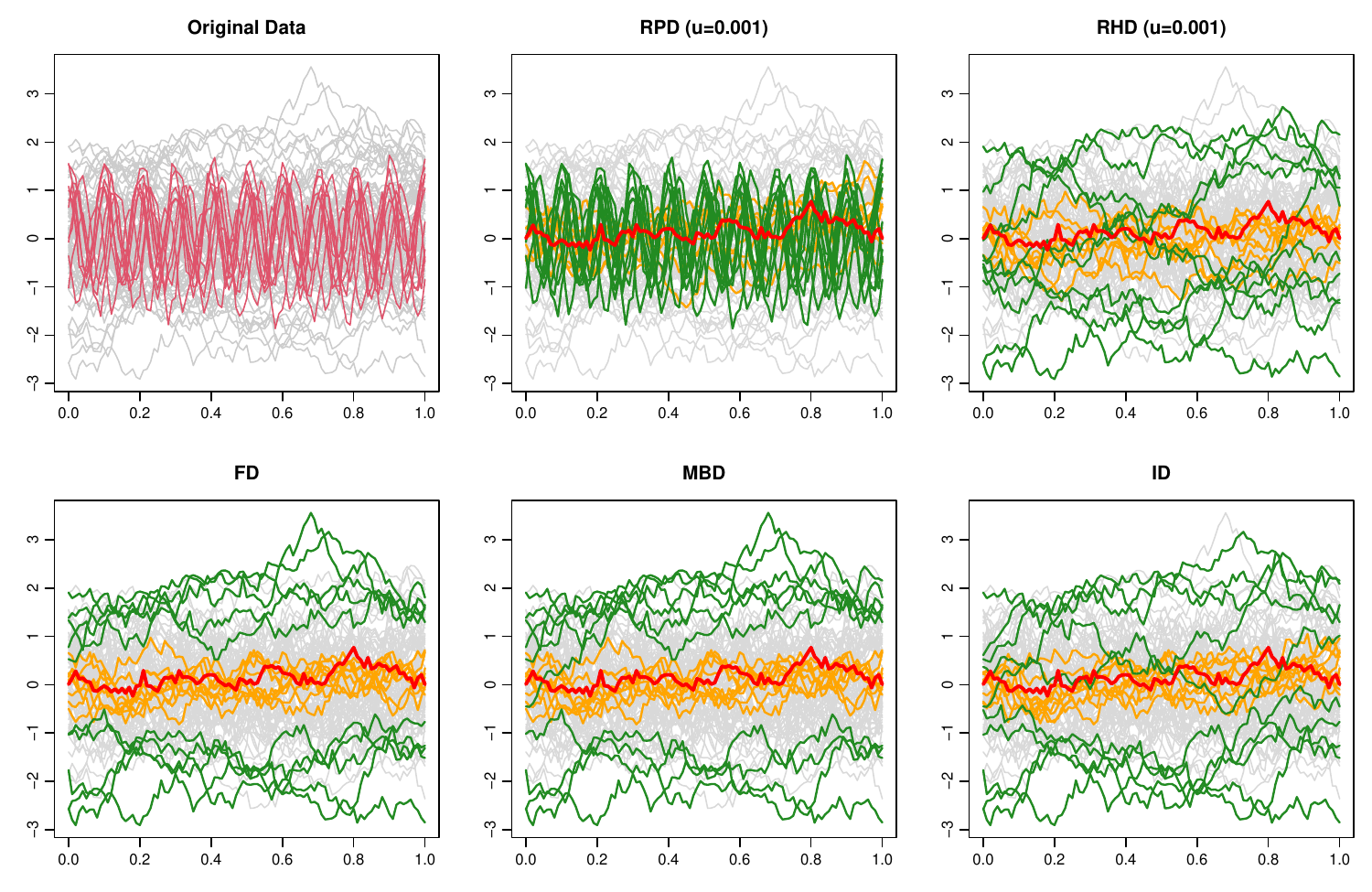}
    \caption{Outlier detection: A single random sample of functional observations (gray) generated from Model~\ref{ModelD5} with sample size $n = 100$, containing $m = 10$ outlying curves (red) (top left). 
    The same sample is shown with the sample median curve highlighted in red, the $10\%$ deepest observations in orange, and the $10\%$ least deep observations in green, based on \RPD{} and \RHD{} with $u = 0.001$, \FD{}, \MBD{} and \ID{} (remaining panels).
    }
    \label{fig: OutlierDetDepths}
\end{figure}

\subsubsection{Threshold selection} \label{subsec: threshold}
While Section~\ref{subsec: discriminative} established the near-perfect discriminative power of the~\RPD{} in separating outliers, this capability cannot be fully leveraged without a robust mechanism for selecting a threshold to flag outlying observations.
Determining such a threshold is challenging because the exact distribution of depth values is generally unknown and depends on the underlying data-generating process.
A common approach in the literature relies on functional boxplots~\citep{SunGenton2011}. 
Nevertheless, since shape outliers can effectively remain within the central envelope formed by non-outlying functions (see Models~\ref{ModelD5} and~\ref{ModelD6}), such methods often fail to distinguish them from the rest of the sample, see also~\citet{Nagy_etal2025}.
To address this, we operate directly in the $[0,1]$ space, transforming the functional data into univariate depth values and performing classification based solely on these measures.

As shown in Section~\ref{subsec: discriminative}, the \RPD{} distinguishes outlying functions almost flawlessly by assigning them significantly lower depth values.
This behavior is visually apparent in Figure~\ref{fig: OutlierDetectionDensities}, where the kernel density estimates of the resulting logit-depths---obtained via the transformation $\text{logit}(t) = \log(t / (1 - t))$---exhibit a clear bimodal character.
The primary mode, located at higher depth values, corresponds to non-contaminated curves, whereas the secondary mode at lower values corresponds to outliers.
This bimodality stems from two key properties of \RPD{}. 
First, \RPD{} is defined using $\med$ and $\MAD$ estimators, both possessing a high breakdown point.
Consequently, $D_{(u)}(x; P_X)$ is robust (as a function of $x$), meaning that the depth values of clean (that is, non-contaminated) observations are not much influenced by contamination. 
Second, \RPD{} effectively isolates outliers by pushing their depth values toward zero. 

\begin{figure}
    \centering
    \includegraphics[width=0.8\linewidth]{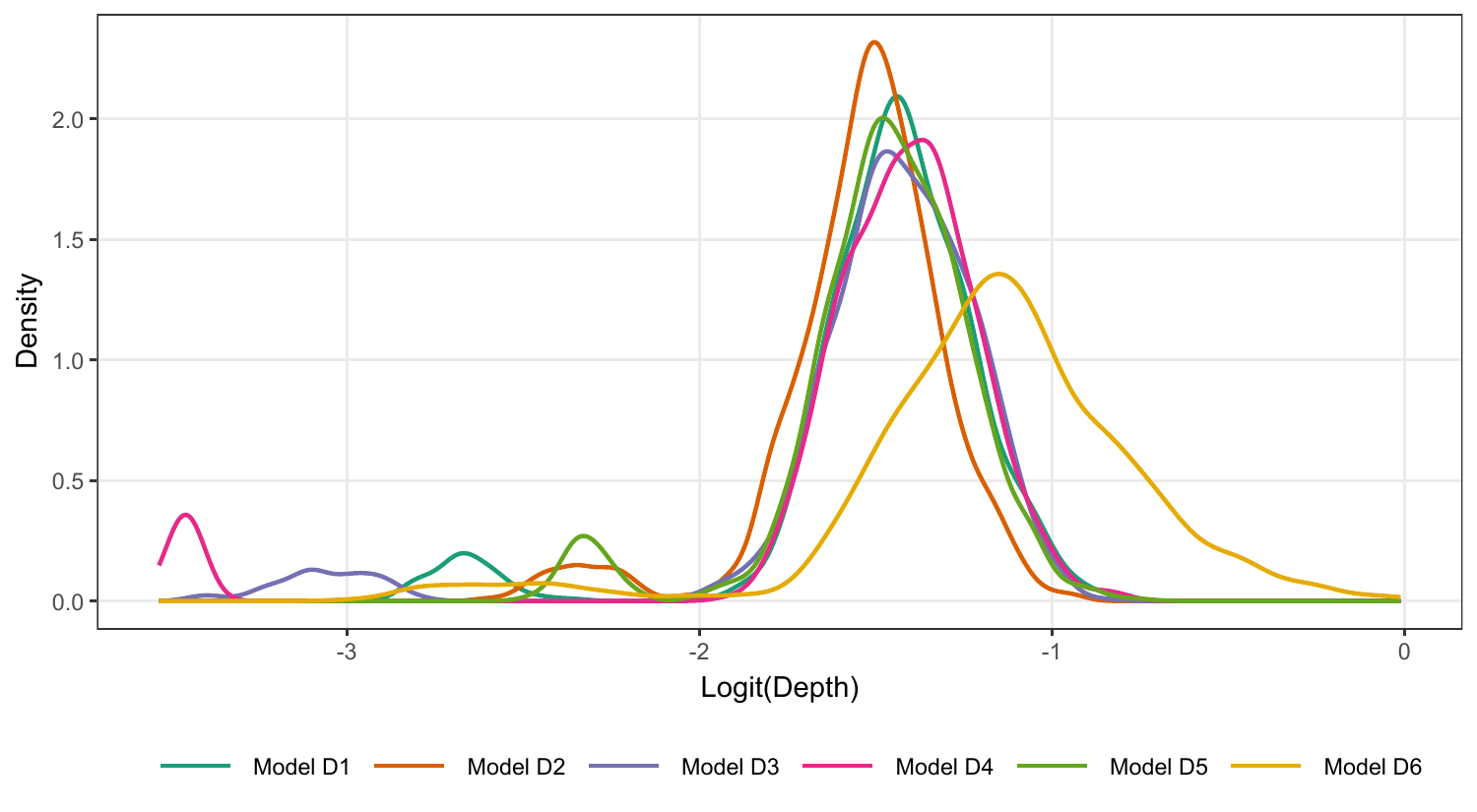}
    \caption{Outlier detection: Density estimates of the logit-transformed $\RPD{}$ (with $u=0.001$) scores across Models~\ref{ModelD1}--\ref{ModelD6}. For each model, the dataset consists of $n = 1\,000$ functional observations, out of which there are $m = 50$ outliers. The logit transformation is applied to stretch out the \RPD{}~values from $[0,1]$ to $\R$, highlighting the separation between the regular curves (main mass) and the outliers (lower depth values).}
    \label{fig: OutlierDetectionDensities}
\end{figure}

Motivated by these observations, we propose a threshold selection mechanism that models the logit-depth values of the sample as a finite Gaussian mixture.
We refer to this approach as the $\RPD{}$-outlier detection method (\RPDOD{}), which consists of three steps:

\begin{enumerate}
    \item We first apply a logit transformation to map the original depth values onto the real line for Gaussian modeling. 
    
    \item Using the \textsf{mclust} package in \textsf{R}~\citep{mclust_package}, we then perform clustering on these transformed values based on parameterized Gaussian mixture models estimated via the EM~algorithm. To determine the optimal number of components, the Bayesian Information Criterion (BIC)~\citep{Schwarz1978} is employed to strictly select either a single-component or a two-component Gaussian mixture model.

    \item The outlier detection rule is then defined as follows:
    \begin{itemize}
        \item \textbf{Unimodal density:} If the estimated density is unimodal, no outliers are detected in the sample.
        \item \textbf{Bimodal density:} If the selected model is a mixture of two Gaussian distributions and the resulting density is bimodal,\footnote{Note that the number of modes of a mixture of two univariate Gaussian distributions is at most two; see~\citet{Amendola_etal2020}.} the presence of outliers is flagged. The cutoff threshold is set at the antimode, defined as the stationary point of the estimated density located between the two modes. Furthermore, to minimize false discoveries, we impose a strict constraint: the proportion of flagged outliers cannot exceed $20$\% of the total sample size. If the proportion of observations falling below the calculated threshold exceeds this limit, we assume that this group represents a distinct sub-population rather than rare anomalies; consequently, no outliers are detected in the sample. This $20$\% upper bound can, of course, be adjusted.
    \end{itemize}
\end{enumerate}

To evaluate the effectiveness of the proposed~\RPDOD{}, we benchmark its performance against three established outlier detection techniques commonly utilized in functional data analysis. 
Specifically, we compare our approach with:

\begin{itemize}
    \item the functional boxplot~\citep{SunGenton2011}, which constructs a central region using the extreme rank length index~\citep{Dai_etal2020} and inflates this region by a factor of~$1.5$, analogous to the standard univariate boxplot, flagging any functions that fall outside this inflated region as outliers (available in the \textsf{fdaoutlier} \textsf{R}~package~\citep{fdaoutlier_package});

    \item the outliergram~\citep{Arribas_Romo2014}, designed primarily to detect shape anomalies by relating \MBD{} to the modified epigraph index (available in the \textsf{roahd} \textsf{R}~package~\citep{roahd_package}); and

    \item the Magnitude-Shape plot (MS-plot)~\citep{DaiGenton2018}, which separates magnitude and shape outlyingness (also available in the \textsf{fdaoutlier} package).
\end{itemize}

The performance of the four considered outlier detection methods is subsequently compared within the simulation models~\ref{ModelD1}--\ref{ModelD6}. 
In all scenarios we consider a fixed total sample size of $n = 500$, while the number of outliers $m$ takes on various values $\set{0,5,25,50}$.
Analogous results with $n=100$ and $n=1\,000$ are presented in Tables~\ref{tab: outlierDetectionResults_n100} and~\ref{tab: outlierDetectionResults_n1000} in Section~\ref{Supplement: Simulations} of the online supplementary material.
To compare the accuracy of the individual approaches, we utilize two evaluation metrics: the False Discovery Rate ($\mathrm{FDR}$) and the Matthews Correlation Coefficient ($\mathrm{MCC}$)~\citep[also known as the $\varphi$-coefficient]{Cramer1946}, both estimated using $1\,000$ Monte Carlo runs.

Formally, using the elements of the confusion matrix, the $\mathrm{FDR}$ is calculated as $\mathrm{FDR} = \mathrm{FP}/(\mathrm{TP} + \mathrm{FP})$,
where $\mathrm{FP}$ denotes the number of false positives (clean curves incorrectly identified as outliers) and $\mathrm{TP}$ stands for true positives (correctly identified outliers).
In cases where a method does not declare any observations as anomalous ($\mathrm{TP} + \mathrm{FP} = 0$), the $\mathrm{FDR}$ is conventionally set to $0$. 
A lower $\mathrm{FDR}$ value indicates reduced tendency to detect false anomalies.
On the other hand, the $\mathrm{MCC}$ is a comprehensive metric for evaluating the quality of binary classifications.
Unlike standard accuracy, it takes into account all four categories of the confusion matrix. 
Formally, the $\mathrm{MCC}$ is defined as
$$\mathrm{MCC} = \frac{\mathrm{TP} \times \mathrm{TN} - \mathrm{FP} \times \mathrm{FN}}{\sqrt{(\mathrm{TP} + \mathrm{FP})(\mathrm{TP} + \mathrm{FN})(\mathrm{TN} + \mathrm{FP})(\mathrm{TN} + \mathrm{FN})}},$$
where $\mathrm{TN}$ is the number of true negatives (correctly identified clean curves) and $\mathrm{FN}$ represents false negatives (undetected outliers). 
The primary advantage of the~$\mathrm{MCC}$ in the context of anomaly detection is that it provides a reliable and balanced evaluation even in cases with extremely imbalanced classes (since $m \ll n$). 
$\mathrm{MCC}$ values range from $-1$ to $1$, where a value of $1$ represents perfect outlier identification, $0$ corresponds to a random classification, and $-1$ indicates total disagreement between the prediction and the actual observation.

The results of the simulation study are summarized in Table~\ref{tab: outlierDetectionResults_n500}. 
They demonstrate that the proposed \RPDOD{} provides highly competitive performance across a variety of contamination scenarios.
As for the $\mathrm{FDR}$, the functional boxplot achieves the lowest values across all scenarios except Model~\ref{ModelD6}. 
However, this apparent advantage is heavily offset by its general inability to detect the actual outliers. 
The \RPDOD{} exhibits exceptional control over false positives, consistently yielding values close to zero across all models and contamination levels. 
This contrasts sharply with the outliergram and the MS-plot, which suffer from a significantly inflated false discovery rate, particularly at lower contamination amounts.
Notably, in the uncontaminated scenarios ($m = 0$), both the \RPDOD{} and the functional boxplot correctly refrain from flagging regular observations as anomalies.
In contrast, the outliergram and the MS-plot suffer from severe false-positive inflation, frequently reaching an $\mathrm{FDR}$ of $1.000$ when no outliers are present.
It should be noted that for the \RPDOD{}, we imposed a structural constraint that no more than $20\%$ of all observations can be flagged as outliers. This limit is primarily a conservative safeguard for very small sample sizes ($n \ll 100$), preventing false discoveries when random noise in uncontaminated data ($m=0$) accidentally produces a bimodal density estimate.
Since this threshold is practically never triggered for sufficiently large sample sizes like those in our study, any comparative advantage it seemingly provides over the outliergram and the MS-plot is negligible in practice.

In terms of overall classification quality, measured by the $\mathrm{MCC}$, the \RPDOD{} outperforms the competing approaches in the majority of the simulation settings. 
In Model~\ref{ModelD1} with pure magnitude outliers, both the \RPDOD{} and the functional boxplot achieve perfect or near-perfect identification ($\mathrm{MCC} \approx 1.000$).
However, as the outlier structures become more complex in Models~\ref{ModelD2}--\ref{ModelD5}, the performance of the functional boxplot deteriorates, dropping near zero in Models~\ref{ModelD2} and \ref{ModelD5}.
Meanwhile, the \RPDOD{} maintains excellent classification accuracy ($\mathrm{MCC} > 0.85$ for $m=5$ and approaching $1.000$ for higher contamination levels). 
Model~\ref{ModelD5} is particularly noteworthy, as the \RPDOD{} successfully identifies the anomalies while all three benchmark methods completely fail to detect the outliers ($\mathrm{MCC} \le 0.080$). 
An exception to this trend is observed in Model~\ref{ModelD6} at the lowest contamination level ($m = 5$). 
Here, the MS-plot outperforms the \RPDOD{} approach, which struggles to identify the anomalies ($\mathrm{MCC} = 0.060$).
This stems from our thresholding mechanism failing to detect a distinct outlier component at only $1\%$ contamination.
For $m = 25,50$, the bimodal structure becomes sufficiently pronounced for the \RPDOD{} to successfully identify anomalies, ultimately outperforming the MS-plot ($\mathrm{MCC} = 0.975$).

These results confirm that modeling the logit-transformed depth values via a Gaussian mixture provides an effective outlier detection rule. 
The \RPDOD{} strikes an optimal balance: it achieves state-of-the-art detection accuracy for various types of functional outliers while strictly safeguarding against false discoveries. 
We acknowledge, however, that our detection mechanism may fail when the absolute number of outliers is extremely small. 
For instance, in the presence of only a single outlying function, the EM algorithm will inevitably fail to capture a bimodal distribution. 
In such extreme cases, we highly recommend a direct visual inspection of the depth values.

\input{Tables_applications/OutlierDetection_n500}

\subsection{Classification}\label{sec: Classification}

As a next application, we consider a classical binary supervised classification task. For each of the three models considered below, we generate a training sample of size $n_{\textrm{train}} \in \set{100, 500, 1\,000}$ functions, drawn independently from two groups denoted by $X$ and $Y$.

\begin{enumerate}[label=\textbf{Model~(C$_\arabic*$)}, ref=(C$_\arabic*$), leftmargin=*]
    \item\label{ModelC1} \emph{Location difference model:} $X$ and $Y$ are Gaussian processes with mean functions $\mu_X(t) = 30\, t^{1.2}(1-t)$ and $\mu_Y(t) = 30\, t (1-t)^{1.2}$, respectively, both with covariance $\Sigma(s, t) = 0.2 \exp\left( - \abs{s - t}/0.3 \right)$, $s, t \in [0,1]$. This is the standard location-difference model from \citet{Cuevas_etal2007}; see also the left-hand panel of Figure~\ref{fig: ClassificationDatasets}. This model is a slight modification of~\ref{ModelD2}.

    \item\label{ModelC2} \emph{Central jump model:} For $\phi_1, \dots, \phi_7$ the first seven orthogonal polynomials on $[0,1]$ as in Model~\ref{ModelD6} (maximum degree $6$), we have $X = \mu + \sum_{j=1}^7 c_j \phi_j$, where $\mu(t) = 10 \, t (1-t)$ and $(c_1, \dots, c_7)\tr$ is centered Gaussian with covariance matrix $\Sigma \in \R^{7 \times 7}$ with unit diagonal and $0.9$ as off-diagonal terms, and $Y = \mu + \sum_{j=1}^5 d_j \phi_j + \sum_{j=6}^7 d_j \xi_j$, where $(d_1, \dots, d_7)\tr$ is Gaussian with mean $(0, \dots, 0, 1, 1)\tr \in \R^7$ and covariance matrix $\Sigma$ as for $X$. Here, $\xi_6, \xi_7$ are the $11$th and $12$th B-spline basis functions of degree $3$ with $21$ equispaced knots on $[0,1]$. The B-spline basis functions contribute to a significant jump in the center of the domain for $Y$, distinguishing these functions from $X$, see the middle panel of Figure~\ref{fig: ClassificationDatasets}.

    \item\label{ModelC3} \emph{Dimension difference model:} For $X$ as in Model~\ref{ModelC2}, $Y = \sum_{j=1}^5 d_j \phi_j$ with $(d_1, \dots, d_5)\tr$ Gaussian with mean the first eigenvector of $\Sigma_2 = \Sigma_{(5, 5)}^{-1}/10$, where $\Sigma_{(5, 5)} \in \R^{5 \times 5}$ is the first $5 \times 5$ sub-matrix of $\Sigma$ from Model~\ref{ModelC2}, and variance matrix $\Sigma_2$. In this model, the dimensions of $X$ and $Y$ differ, and the functions differ in shape similarly to those in Model~\ref{ModelD6}; see the right-hand panel of Figure~\ref{fig: ClassificationDatasets}.
\end{enumerate}

\begin{figure}
\centering
    \includegraphics[width=0.3\linewidth]{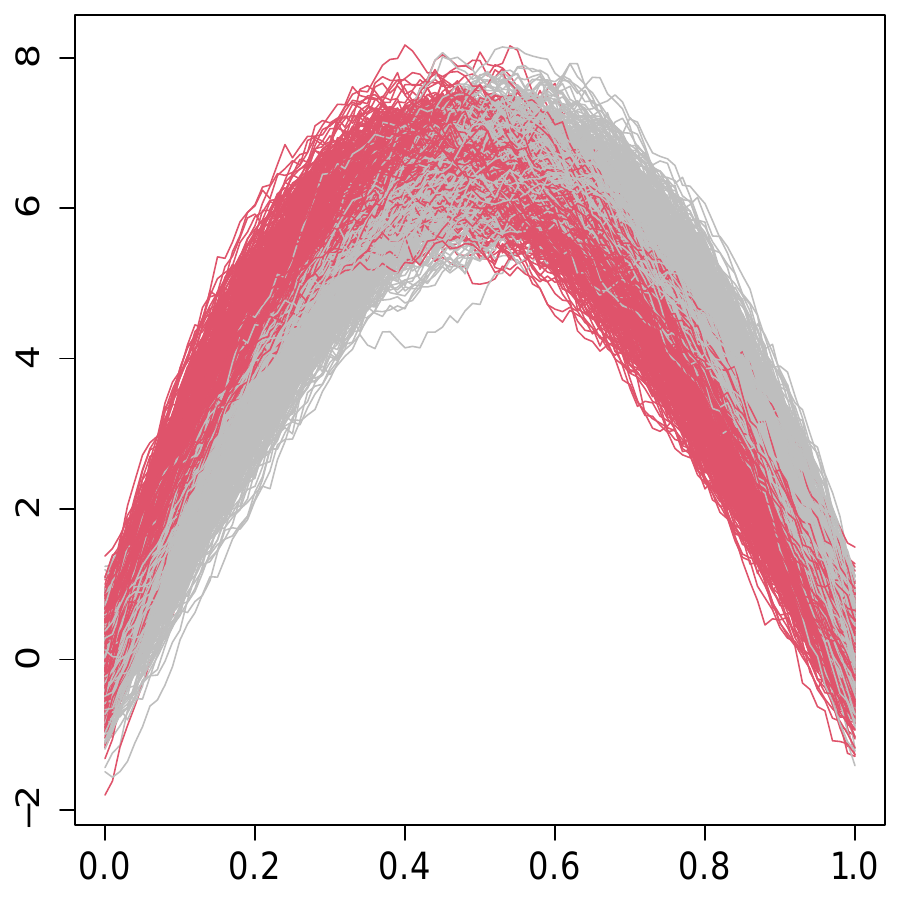}
    \includegraphics[width=0.3\linewidth]{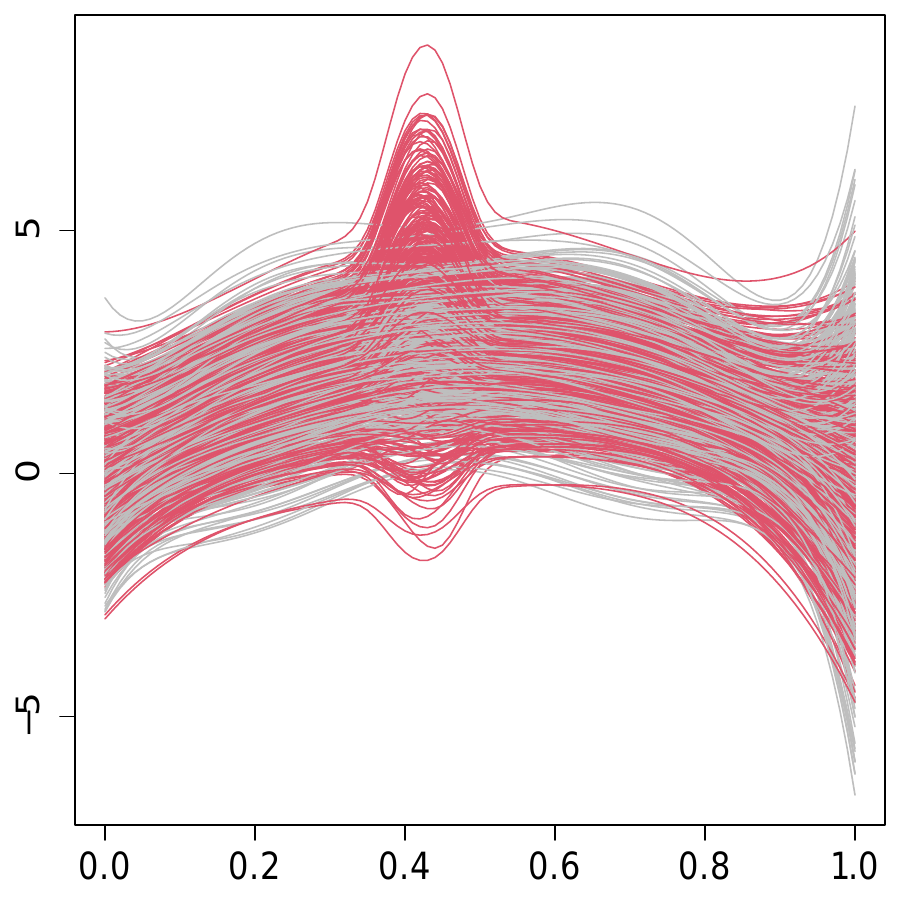}
    \includegraphics[width=0.3\linewidth]{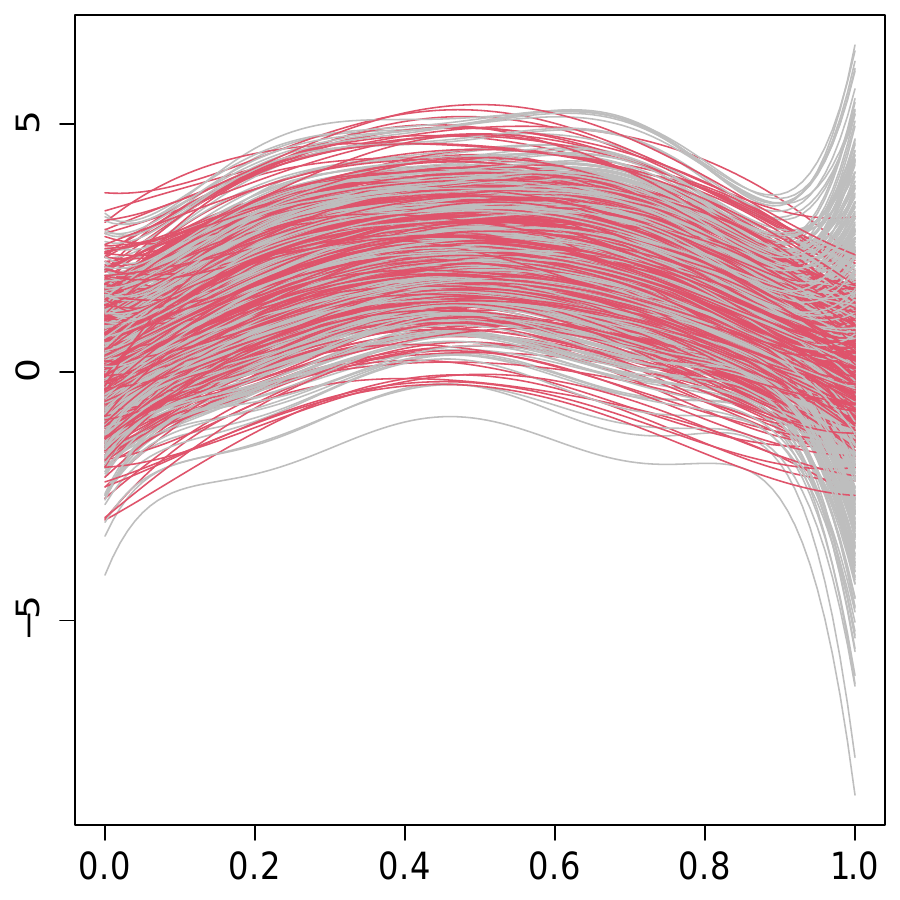}
    \caption{Supervised classification: Simulated datasets generated under Models~\ref{ModelC1}--\ref{ModelC3} detailed in Section~\ref{sec: Classification}. 
    Functions belonging to class~$X$ are depicted in gray, while those from class~$Y$ are red.}
    \label{fig: ClassificationDatasets}
\end{figure}

For a given function $z$ and a specific depth measure $D$, the depth-depth (DD) transform is defined as the bivariate vector $(D(z; X), D(z; Y))\tr \in [0,1]^2$. 
Here, $D(z; X)$ represents the sample $D$-depth of $z$ evaluated with respect to the sample from $X$, and analogously for $Y$. 
Following \citet{Li_etal2012}, we consider two DD-classifiers:
    \begin{itemize}
        \item The \textbf{max-depth} classifier (\textbf{max}), assigning a new observation $z$ to the group where it achieves a higher $D$-depth value.
        \item The \textbf{linear DD-classifier} (\textbf{DD}), establishing a decision boundary using a linear function (passing through the origin) that optimally separates the point clouds $\left\{ (D(X_i; X), D(X_i; Y))\tr \right\}_{i=1}^{n_{\textrm{train}}}$ and $\left\{ (D(Y_i; X), D(Y_i; Y))\tr \right\}_{i=1}^{n_{\textrm{train}}}$ in the $[0,1]^2$ space.
    \end{itemize} 

The predictive performance of each classifier and depth combination is assessed using the out-of-sample empirical misclassification rate.
This is computed on an independent test set comprising $n_{\textrm{test}} = 1\,000$ observations per group (ties are resolved by a fair coin toss, assigning a misclassification of $1/2$). 
Table~\ref{tab: Classification_n500} reports the means and standard deviations (in parentheses) of the misclassification rates across $1\,000$ independent simulation runs for $n_{\textrm{train}}=500$. 
The corresponding results for $n_{\textrm{train}}=100$ and $n_{\textrm{train}}=1\,000$ are provided in Tables~\ref{tab: Classification_n100} and~\ref{tab: Classification_n1000} within Section~\ref{Supplement: Simulations} of the online supplementary material.

\input{Tables_applications/Classification_n500}

The empirical findings are unambiguous: \RPD{} consistently outperforms all competing methods across all evaluated scenarios. 
This superiority is particularly striking in Models~\ref{ModelC2} and~\ref{ModelC3}, which are characterized by shape differences. 
In these complex settings, only \RPD{}---and to a markedly lesser extent, \ID{} and \RHD{}---demonstrate the capability to reliably distinguish between the two data groups. 
It is important to note that, for the DD-plot based on  \RPD{}, a simple linear classifier is sufficient to effectively classify the two groups of complex functional data, such as those in Model~\ref{ModelC3}.
This is explained by the analysis of the resulting DD-plots for Model~\ref{ModelC3} in Figure~\ref{fig: DD}, where the DD-transforms of the training data are displayed. Analogous diagnostic plots for Models~\ref{ModelC1} and~\ref{ModelC2} are in Section~\ref{Supplement: Simulations} of the online supplementary material.

From Table~\ref{tab: Classification_n500}, we observe that \RPD{} with regularization parameter $u = 0.1$ performs worse than \RPD{} with $u = 0.001$. Therefore, also for classification, we generally recommend using a very small regularization parameter, such as $u = 0.001$, consistently with the outlier detection task.

\begin{figure}
    \centering
    \includegraphics[width=0.9\linewidth]{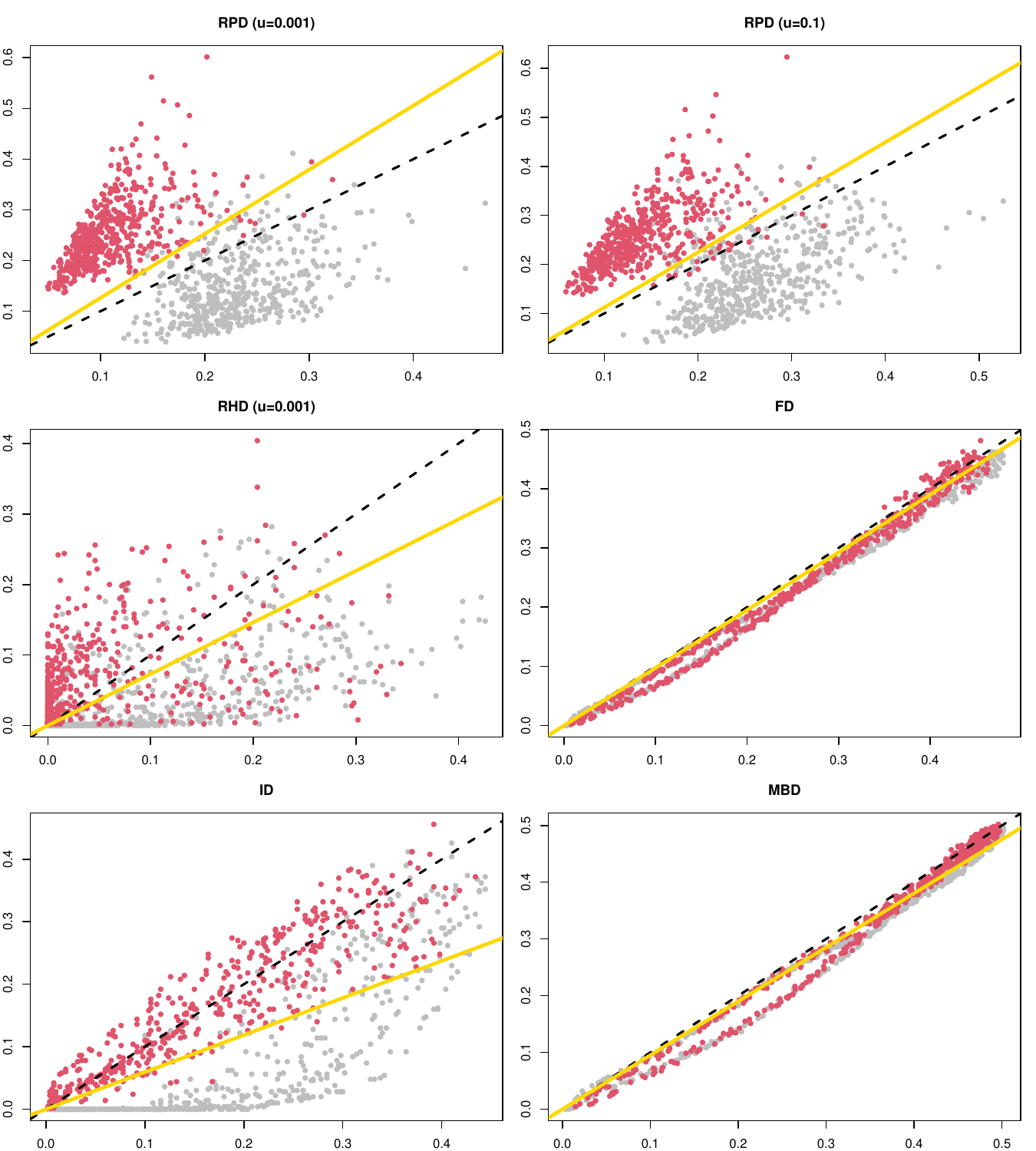}
    \caption{Supervised classification: DD-plots of the training data in a single run of Model~\ref{ModelC3} (gray points for $X_i$, red points for $Y_i$) with (a) \RPD{} (top left, $u = 0.001$), (b) \RPD{} (top right, $u = 0.1$), (c) \RHD{} (middle left, $u = 0.001$), (d) \FD{} (middle right), (e) \ID{} (bottom left), and (f) \MBD{} (bottom right). In the plots, the dashed line represents the max-depth classifier and the yellow line the linear DD-classifier best separating the two groups. The results are remarkable---while \RPD{} with DD-classifier achieves almost perfect separation of the clouds, \FD{} and \MBD{} are unable to cope with the shape difference between the clusters. \RHD{} and \ID{} perform slightly better, but the overlap of the two DD-clouds is substantial.}
\label{fig: DD}
\end{figure}

\subsection{Hypothesis testing}\label{sec: HypothesisTesting}

In our third application, we demonstrate the practical utility of the \RPD{} by employing it in rank-based tests to assess the equality of distributions.
This serves primarily as an empirical demonstration that \RPD{}-based rankings are statistically sensible.
Specifically, we construct the Kruskal–Wallis (KW) test statistic following the methodology of \cite{CS12}, utilizing rankings derived from the \RPD{}. 
When computing the KW statistic, standard tie corrections are applied, as implemented in the \textsf{kruskal.test} function in \textsf{R}.
For comparison, we perform identical KW test procedures, substituting the \RPD{} with other functional depths under consideration.\footnote{To prevent self-inclusion bias---which artificially inflates the depth of an observation when evaluated with respect to its own sample---we apply finite-sample corrections. For halfspace-based depths (\RHD{}, \FD{}, \ID{}) and \MBD{}, the adjusted depths are given by $(D - 1/n_k)\frac{n_k}{n_k-1}$ and $(D - 2/n_k)\frac{n_k}{n_k-2}$, respectively, where $n_k$ is the sample size and $D$ is the uncorrected in-sample depth. No continuity correction needs to be applied for \RPD{}.}
All tests are conducted at a 5\% significance level.

Since multi-sample extensions are straightforward, we restrict attention to two-sample observations $\set{X_{k, i}}_{i=1}^{n_k}$ for $k=1,2$, where $n_1=n_2=n/2$ and $n \in \set{50, 100, 500, 1\,000}$.
The functional observations are generated using the following expansion:
\begin{equation*}
    X_{k, i} = \mu_k + \sum_{j=1}^{J_{\mathrm{true}}} \gamma_{k, j}^{1/2} \xi_{k, i, j} \phi_{k, j},
\end{equation*}
for $i=1,\dots,n_k$ and $k=1,2$, setting $J_{\mathrm{true}} = 10$. 
For both groups, the scores $\xi_{k, i, j}$ are drawn independently from a standard normal distribution. 
The eigenvalues are determined by the eigengaps $\gamma_{k, j}-\gamma_{k,j+1} = 2j^{-a_k}$ with a decay rate $a_k \in (2,\infty)$, where the first eigenvalue is defined as $\gamma_{k, 1} = 2\sum_{j=1}^\infty j^{-a_k}$.

To investigate various types of group differences, we fix the parameters for the first group, $\{X_{1, i}\}_{i=1}^{n_1}$, and vary only those of the second group, $\{X_{2, i}\}_{i=1}^{n_2}$. For the first group, the mean function is set to zero, $\mu_1(t) = 0$ for $t \in [0,1]$, and the eigengap decay rate is $a_1 = 5$. The eigenfunctions correspond to the first $J_{\mathrm{true}}=10$ trigonometric functions, meaning $\phi_{1, j} = \phi_{\mathrm{tri},j}$ for $j=1,\dots,J_{\mathrm{true}}$. The trigonometric basis $\{\phi_{\mathrm{tri},j}\}_{j=1}^\infty$ is defined as
\begin{equation*}
    \phi_{\mathrm{tri},1}(t) = 1, \ \ 
    \phi_{\mathrm{tri},2m}(t) = \sqrt{2}\sin (2m\pi t), \ \
    \phi_{\mathrm{tri},2m+1}(t) = \sqrt{2}\cos (2m\pi t), \ \ t \in [0,1], \ \ m\in\N.
\end{equation*}

To construct scenarios with different alternatives, we define the following quantities:
(i) $\mu_{\mathrm{mag}}(t) = 3$ and
(ii) $\mu_{\sin}(t) = 1.2 \sin(2\pi t)$ for $t \in [0,1]$,
(iii) a decreasing sequence $\{\gamma_{\mathrm{rough},j}\}_{j=1}^\infty$ defined by  gap $\gamma_{\mathrm{rough},j}-\gamma_{\mathrm{rough},j+1} = 2j^{-a_{\mathrm{rough}}}$ with decay rate $a_{\mathrm{rough}}=2.5$ where $\gamma_{\mathrm{rough},1} = 2\sum_{j=1}^\infty j^{-a_{\mathrm{rough}}}$,
and (iv) an orthonormal subset with high frequency $\{\phi_{\mathrm{hfq},j}\}_{j=1}^{J_{\mathrm{true}}}$ where $\phi_{\mathrm{hfq},j} = \phi_{\mathrm{tri},j+J_{\mathrm{true}}}$ for $j=1,\dots,J_{\mathrm{true}}$.
Using a parameter $c \in \{0, 0.2, 0.4, 0.6, 0.8, 1\}$ to control the severity of the alternatives, we consider the following four models:

\begin{enumerate}[label=\textbf{Model~(T$_\arabic*$)}, ref=(T$_\arabic*$), leftmargin=*]
\item \label{ModelT1} \emph{Magnitude difference in means:}
$\mu_2 = c\mu_{\mathrm{mag}}$ but $\gamma_{2, j} = \gamma_{1, j}$ and $\phi_{2, j} = \phi_{1, j}$.
\item \label{ModelT2} \emph{Shape difference in means:}
$\mu_2 = c\mu_{\mathrm{sin}}$ but $\gamma_{2, j} = \gamma_{1, j}$ and $\phi_{2, j} = \phi_{1, j}$.
\item \label{ModelT3} \emph{Difference in eigenvalues:}
$a_2 = (1-c)a_1 + c \, a_{\mathrm{rough}}$ but $\mu_2=\mu_1$ and $\phi_{2, j} = \phi_{1, j}$.
\item \label{ModelT4} \emph{Difference in eigenfunctions:}
$\{\phi_{2, j}\}_{j=1}^{J_{\mathrm{true}}}$ is the Gram-Schmidt orthonormalized set of the functions $\{(1-c)\phi_{1, j} + c\, \phi_{\mathrm{hfq},j}\}_{j=1}^{J_{\mathrm{true}}}$
 but $\mu_2=\mu_1$ and $\gamma_{2, j} = \gamma_{1, j}$.
\end{enumerate}

Setting $c=0$ recovers the exact same null scenario across all Models~\ref{ModelT1}--\ref{ModelT4}.
Note that in Models~\ref{ModelT1} and \ref{ModelT2}, the groups differ only in their means, whereas in~\ref{ModelT3} and~\ref{ModelT4}, the difference lies solely in their second-order structure.
One realization of these curves (with $c=1$ for the second group) is depicted in Figure~\ref{fig_2test_real} in Section~\ref{Supplement: Simulations} of the online supplementary material.

For brevity, we present the results for $n=100$; comprehensive results for other sample sizes are summarized in Table~\ref{tab: KW_test_full} and Figures~\ref{fig: KW_test_n50}--\ref{fig: KW_test_n1000} in Section~\ref{Supplement: Simulations} of the online supplementary material. 

The empirical sizes (evaluated using $1\,000$ independent Monte Carlo replications under the shared null hypothesis $c=0$) are presented in Table~\ref{tab: KW_test}. 
All depth-based KW tests successfully control the Type I error. The empirical sizes cluster tightly below the nominal $5\%$ level (ranging from $0.037$ to $0.046$), indicating slightly conservative tests.

Figure~\ref{fig: KW_test_n100} illustrates the empirical power of the tests under Models~\ref{ModelT1}--\ref{ModelT4}.
In Model~\ref{ModelT1} (magnitude differences), all depths perform similarly, although \RPD{} shows a slightly slower increase in power with $c$. Performance differences are most evident in Models~\ref{ModelT2}--\ref{ModelT4}. First, \FD{} and \MBD{} behave similarly due to their related definitions. 
These two depths are the weakest among the compared methods, exhibiting lower power in Models~\ref{ModelT2}--\ref{ModelT4} and completely failing to capture group differences in Model~\ref{ModelT4}.
In Model~\ref{ModelT2}, \RHD{} (with both values of $u \in \set{0.001, 0.1}$) demonstrates the best performance, with \RPD{} closely behind and \ID{} performing slightly worse.
The limitations of traditional depths become apparent in Models~\ref{ModelT3} and~\ref{ModelT4}. 
In Model~\ref{ModelT3}, \FD{} and \MBD{} fail to detect the alternatives until the signal intensity is extremely high. 
In contrast, \RPD{} maintains excellent sensitivity to this group difference, while \RHD{} and \ID{} exhibit slightly weaker power.
In the final Model~\ref{ModelT4}, \RPD{} achieves a power very close to $1$ even for a small signal intensity of $c=0.2$. 
Regarding the other depths, \ID{} picks up the group differences only at $c=0.4$, and \RHD{} is even less sensitive, detecting the signal only for $c\geq 0.6$.
\FD{} and \MBD{} completely fail to detect any differences in this model.

To summarize, \RPD{} exhibits highly competitive results in the simple difference scenarios (Models~\ref{ModelT1} and~\ref{ModelT2}).
In Models~\ref{ModelT3} and~\ref{ModelT4}, \RPD{} strictly dominates all competitors.
Notably, \RPD{} tends to yield better results for shape differences when a weaker regularization parameter ($u=0.001$) is applied.

\input{Tables_applications/KW_test}

\begin{figure}
    \centering
\includegraphics[width=0.8\textwidth]{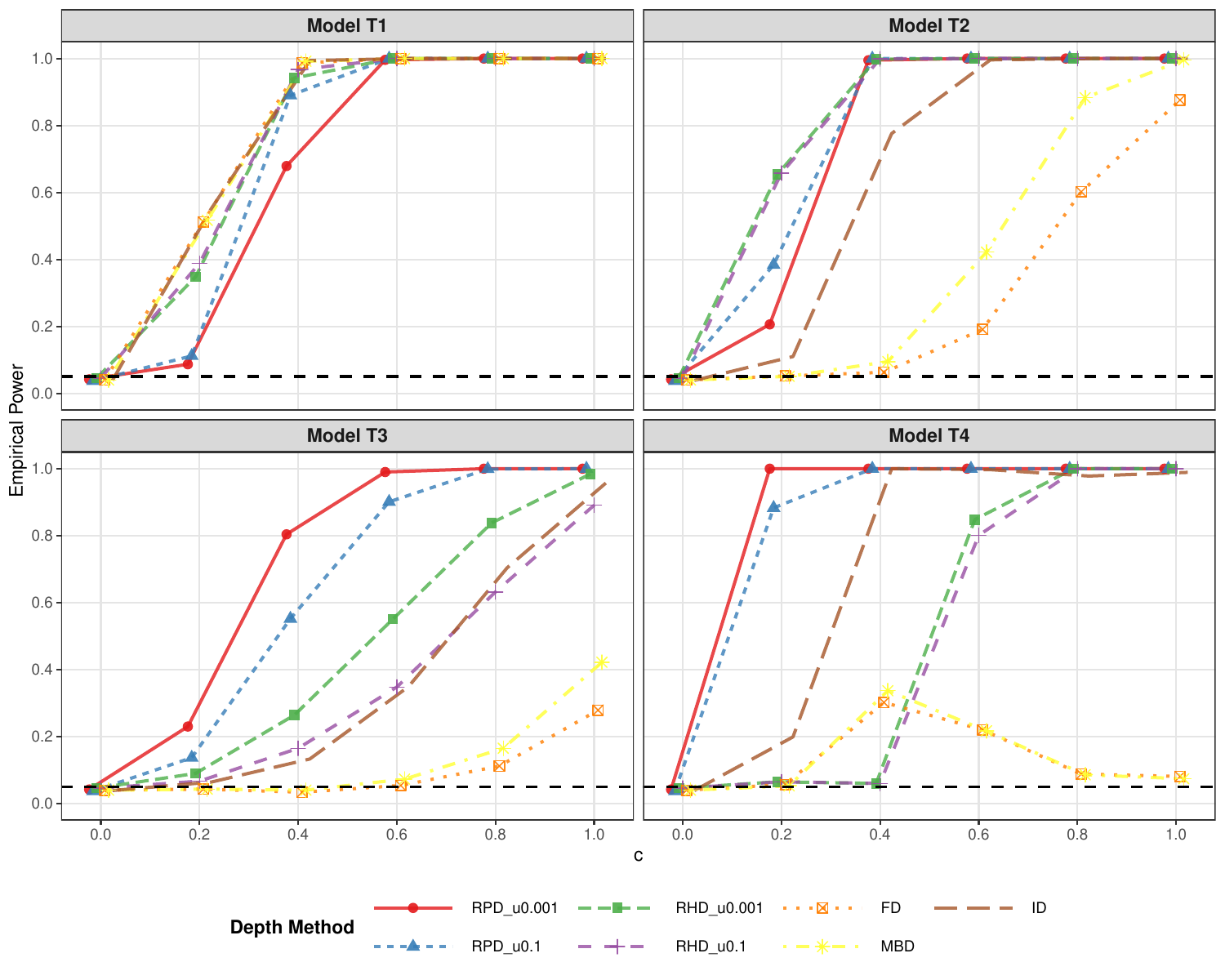}
    \caption{
        Hypothesis testing: Empirical power of the depth-based Kruskal--Wallis tests for sample size $n=100$ under Models~\ref{ModelT1}--\ref{ModelT4}. 
        The horizontal axis represents the intensity parameter $c$ controlling the magnitude of the distributional difference between two samples ($c=0$ for the null hypothesis), while the vertical axis shows the proportion of rejected null hypotheses over $1\,000$ Monte Carlo replications. 
        Higher curves (for $c>0$) indicate better ability to detect departures from the null hypothesis. The horizontal dashed line indicates the nominal type I error rate of $0.05$. 
    }
    \label{fig: KW_test_n100}
\end{figure}

\section{Real-world data analysis}\label{sec: RealData}
To conclude, we demonstrate the practical utility of the proposed \RPD{} in nonparametric hypothesis testing. 
Building upon the simulation results from Section~\ref{sec: HypothesisTesting}, we combine our method with the depth-based $k$-sample test of \cite{CS12} to capture structural differences among distributions in real-world data. Specifically, we analyze a dataset comprising two years of bike rental records (2011--2012) from the Capital Bikeshare system in Washington, D.C.~\citep{Fanaee_etal2014}.\footnote{Publicly available at \url{https://archive.ics.uci.edu/dataset/275/bike+sharing+dataset}.}
This dataset provides precise hourly rental aggregations, serving as an excellent proxy for analyzing temporal dynamics in human mobility.

In our analysis, each day is represented as a vector $X_{k,i} \in \mathbb{R}^{24}$ (treated as a functional random variable with $24$ equidistant discrete observations per curve) that contains the hourly aggregate number of rentals.
Here, the index $k \in \set{1, \dots, 7}$ denotes the day of the week (starting with Monday), and $i \in \set{1, \dots, n_k}$ corresponds to the specific occurrence of that day within the dataset.
We assume that $\set{X_{k,i}}_{i=1}^{n_k}$ is a random sample drawn from an underlying probability distribution $P_k$.
As illustrated in Figure~\ref{fig: DailyBikesProfiles}, these profiles exhibit strong daily seasonality.
However, identifying group differences remains challenging because variations often manifest in subtle shape patterns rather than obvious differences in magnitude, 
in which case a \RPD{}-based test proves particularly advantageous, see Section~\ref{sec: HypothesisTesting}.
Our primary objective is to test whether these distributions coincide across the different days of the week. 

\begin{figure}
    \centering
    \includegraphics[width=0.9\linewidth]{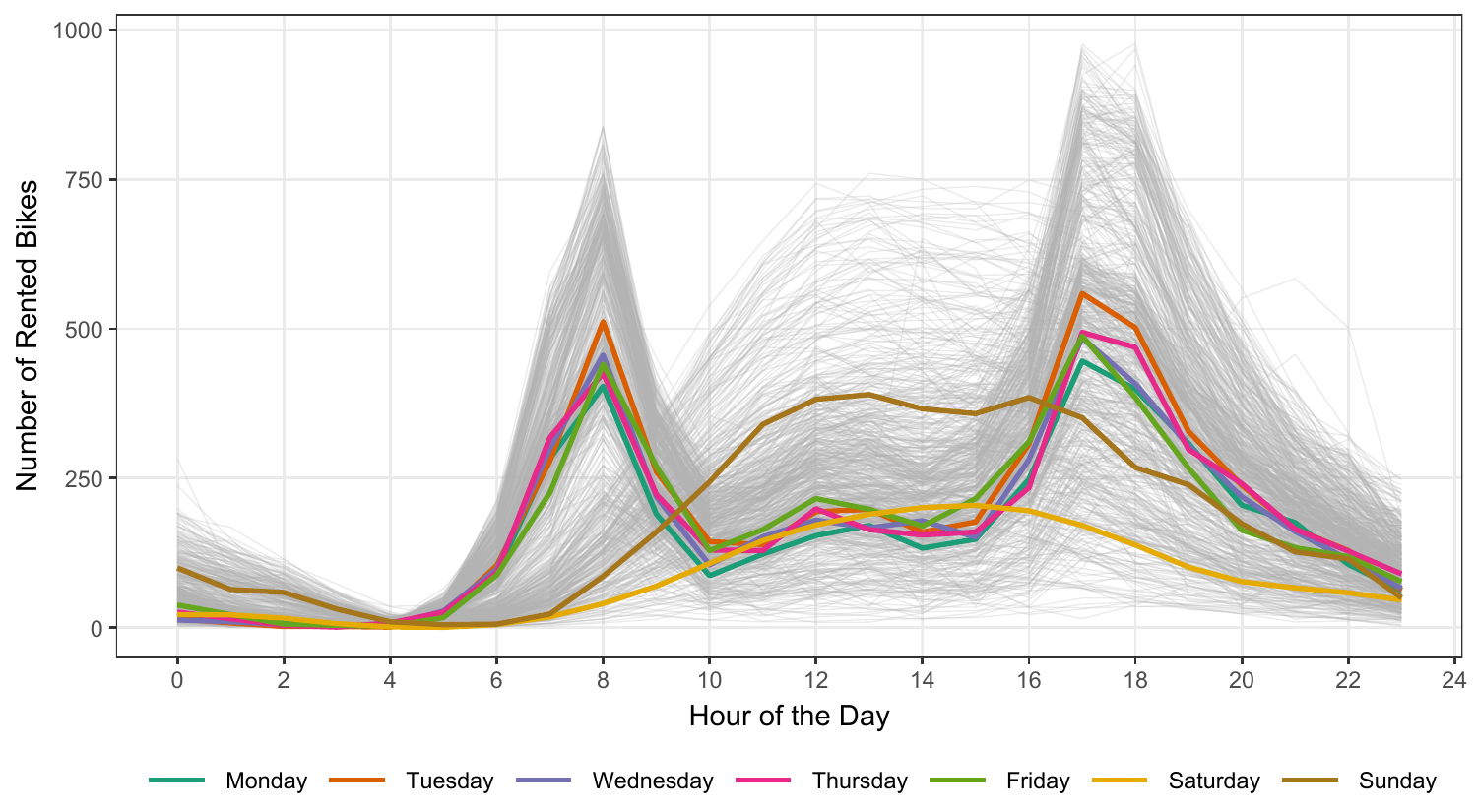}
    \caption{Real-world application: Gray curves represent all observed daily profiles, while colored curves correspond to the deepest functional profiles (with respect to $\RPD{}$ with $u=0.001$) for each day of the week. Workdays exhibit a distinct bimodal pattern with peaks occurring at 8 a.m. and 6 p.m., which corresponds to daily commuting to and from work. Notably, the volume of rentals during these peak hours shows a decreasing trend throughout the workweek, with Friday experiencing significantly fewer rentals than Monday. In contrast, weekends display a unimodal distribution with a single peak around 1 p.m., where the overall number of rented bikes is lower on Saturday compared to Sunday.}
    \label{fig: DailyBikesProfiles}
\end{figure}

Our fundamental question is whether the daily rental profiles differ significantly depending on the day of the week.
The null hypothesis of the global test is $$ H_0^{(1)}: P_1 = P_2 = \dots = P_7 $$ against the alternative that at least two distributions differ.
We calculated the KW test described in Section~\ref{sec: HypothesisTesting}, where the $\RPD{}$ with $u=0.001$ is used, and the resulting p-value $p < 0.001$ leads to an unambiguous rejection of $H_0^{(1)}$.\footnote{Given multiple tested hypotheses, we applied the Bonferroni procedure. All reported adjusted p-values reflect this correction.}
This is confirmed by a visual inspection Figure~\ref{fig: DailyBikesProfiles}, where the \RPD{}-deepest curves for the individual days indicate significant heterogeneity across the week. 

Given the decisive rejection of the global null hypothesis, we proceed with a hierarchical post-hoc analysis to pinpoint the specific sources of this heterogeneity.
Motivated by the clear visual distinction between the bimodal commuting patterns and the unimodal weekend profiles, we first contrast the working week and the weekend. 
We test the null hypothesis $$ H_0^{(2)}: P_{\text{workday}} = P_{\text{weekend}} $$ against the alternative that these two distributions differ. 
The resulting p-value of $p < 0.001$ confirms a highly significant structural difference between these two temporal regimes.
Focusing exclusively on the weekend, Figure~\ref{fig: DailyBikesProfiles} suggests that while the overall functional shape remains similar, the absolute rental volume is lower on Saturday compared to Sunday.
To formally assess this observation, we test the equality of the weekend distributions
$$ H_0^{(3)}: P_6 = P_7. $$
This test yields a p-value of $p = 0.001$, leading us to reject the null hypothesis and conclude that Saturdays and Sundays exhibit statistically distinct rental behaviors.
Next, we examine the intra-week dynamics to determine whether the rental patterns remain stable throughout the working week. We test the homogeneity of all five workday distributions
$$ H_0^{(4)}: P_1 = \ldots = P_5.$$
The corresponding p-value $p<0.001$ indicates significant variations among the workdays. As observed in the deepest profiles in Figure~\ref{fig: DailyBikesProfiles}, Friday appears to deviate from the rest of the week, showing a noticeable decrease in peak volumes. 
To verify if Friday is indeed the primary source of this intra-week variation, we exclude it and test the equivalence of the distributions from Monday through Thursday
$$ H_0^{(5)}: P_1 = \ldots = P_4.$$
For this subset, the test results in a p-value of $p=0.403$. 
Consequently, we fail to reject the null hypothesis, suggesting that the daily rental distributions are statistically indistinguishable from Monday to Thursday.
We conclude that the observed heterogeneity during the workweek is mainly driven by the distinct drop in bike rentals on Friday. Finally, we conducted the same hierarchical analysis using the alternative depth functions described in Section~\ref{sec: Simulations}. The corresponding p-values for all tested hypotheses $H_0^{(1)}, \dots, H_0^{(5)}$ are summarized in Table~\ref{tab: pvalues_realData} in Section~\ref{Supplement: Simulations} of the online supplementary material.

\section*{Online supplementary material}
\begin{itemize}
    \item A \textsf{PDF} document containing the proofs of all theoretical results (Section~\ref{Supplement: Proofs}), a simulation study regarding the robustness of the~\RPD{} median (Section~\ref{Supplement: LocationEstimation}), and supplements to the simulation studies~(Section~\ref{Supplement: Simulations}).
    \item A commented \textsf{R} script with the implementation of all our numerical experiments.
\end{itemize}

\printcredits

\putbib
\phantomsection\label{main:LastPage}
\end{bibunit}

\begin{bibunit}

\clearpage
\hypersetup{pageanchor=false}
\gdef\lastpage{\pageref*{supp:LastPage}}
\pagenumbering{arabic}
\setcounter{page}{1}
\setcounter{section}{0}
\setcounter{figure}{0}
\setcounter{table}{0}
\setcounter{equation}{0}
\renewcommand{\theHsection}{S.\arabic{section}}
\renewcommand{\theHfigure}{S.\arabic{figure}}
\renewcommand{\theHtable}{S.\arabic{table}}
\renewcommand{\theHequation}{S.\arabic{equation}}

\begin{center}
{\Large\bfseries Supplementary Material to ``Projection depth for functional data: Practical issues, computation and applications''\par}
\vspace{1em}
{Filip Bo\v{c}inec\textsuperscript{1}, Stanislav Nagy\textsuperscript{1}, Hyemin Yeon\textsuperscript{2}\par}
\vspace{0.5em}
{\small \textsuperscript{1}Faculty of Mathematics and Physics, Charles University, Prague, Czech Republic\par}
{\small \textsuperscript{2}Department of Mathematical Sciences, Kent State University, Kent, OH, USA\par}
\end{center}
\vspace{1em}

\renewcommand{\thesection}{S.\arabic{section}}
\renewcommand{\thefigure}{S.\arabic{figure}}
\renewcommand{\thetable}{S.\arabic{table}}
\renewcommand{\theequation}{S.\arabic{equation}}

\section{Proofs of theoretical results}\label{Supplement: Proofs}

\begin{proof}[\textbf{Proof of Lemma~\ref{lemma: betaPos}}]\mbox{}
    Note that it suffices to prove that $\nu\left(\set{ v \in \mathbb{S} : \MAD[\inner{X, v}] = 0 }\right)=0$. Assume, for a contradiction, that this is not true.  First note that for any $v \in \SS$ fixed, we have that $\MAD{[\inner{X,v}]}=0$ if and only if $\inner{X,v}=\med{[\inner{X,v}]}$ with probability at least $1/2$. 
    Consequently, if we denote 
    \begin{equation*}
        B=\set{v\in\SS\st P_X(\set{y \in \HH \st \inner{y,v}=\med{[\inner{X,v}]}})\geq 1/2},
    \end{equation*}
    then $\nu(B)>0$. Consider the set 
    \begin{equation*}
        E=\set{(y,v)\in\HH\times B\st \inner{y,v}=\med{[\inner{X,v}]}}.
    \end{equation*}
    By the Fubini-Tonelli theorem \citep[Theorem~4.4.5]{Dudley2002},
    \begin{align*}
        \int_\HH &\nu(\set{v\in B\st \inner{y,v}=\med{[\inner{X,v}]}})\dd P_X(y)
        =\int_\HH\int_B \indic_{E}(y,v)\dd\nu(v)\dd P_X(y)\\
        &=\int_B\int_\HH \indic_{E}(y,v)\dd P_X(y)\dd\nu(v)
        =\int_B P_X(\set{y\in \HH\st \inner{y,v}=\med[\inner{X,v}]}) \dd \nu(v)\\
        &\geq \frac{1}{2}\int_B \dd \nu(v)=\nu(B)/2>0,
    \end{align*}
    where $\indic_{E}(y,v) = 1$ if $(y, v) \in E$ and $0$ elsewhere, and the inequality above follows from the definition of $B$.
    As a consequence, there must exist $x\in\HH$ such that 
    \begin{equation}\label{eqLemma11a}
        \nu(\set{v\in B \st \inner{x,v}=\med[\inner{X,v}]})>0,
    \end{equation}
    that is, the set of normal vectors $v \in \SS$ of hyperplanes of probability at least $1/2$ containing this particular point $x$ has positive $\nu$-measure. We will show that $x$ is indeed the desired atom with probability at least $1/2$. 

    Let $X'\eqdis X-x$ and 
    \begin{align*}
        B'&=\set{v\in B\st P_{X}(\set{y\in\HH\st \inner{y,v}=\inner{x,v}})\geq 1/2}\\
        &=\set{v\in B\st P_{X'}(\set{y\in\HH\st \inner{y,v}=0})\geq 1/2}.
    \end{align*}
    By~\eqref{eqLemma11a} we know that $\nu(B')>0$. 
    We will use the Fubini-Tonelli theorem again for the set $G=\set{(y,v)\in\HH\times B'\st \inner{y,v}=0}$. We obtain
    \begin{align}
        \int_\HH& \nu(\set{v\in {B'}\st \inner{y,v}=0})\dd P_{X'}(y)=\int_\HH\int_{B'} \indic_{G}(y,v)\dd\nu(v)\dd P_{X'}(y)\nonumber\\
        &=\int_{B'}\int_\HH \indic_{G}(y,v)\dd P_{X'}(y)\dd\nu(v)    =\int_{B'} P_{X'}(\set{y\in \HH\st \inner{y,v}=0}) \dd \nu(v)\nonumber\\
        &\geq \frac{1}{2}\int_{B'} \dd \nu(v)=\nu(B')/2. \label{eqLemma11b}
    \end{align}
    Note that $\nu(\set{v\in {B'}\st \inner{y,v}=0})$ can be non-zero only for $y=0$,
    because $\nu$ is smooth. 
    As a consequence, 
    \begin{equation*}
        \int_\HH \nu(\set{v\in {B'}\st \inner{y,v}=0})\dd P_{X'}(y) = \nu(B')P_{X'}(\set{0}),
    \end{equation*}
    which together with~\eqref{eqLemma11b} implies $P_{X'}(\set{0})\geq 1/2$.
    This proves that $X'$ has an atom at $0\in\HH$ with probability at least $1/2$, which is a contradiction.
\end{proof}

\begin{proof}[\textbf{Proof of Theorem~\ref{thm: invariance}}]\mbox{}
    Using the shift-equivariance of the median and shift-invariance of the MAD, we obtain for all $v \in \SS$ and $x\in\HH$ that
    \begin{align*}
        O_v(\mathcal{T}x; P_{\mathcal{T}X})
        &= \frac{\abs{\inner{\mathcal{S}x + e, \mathcal{S}\mathcal{S}^*v} - \med[\inner{\mathcal{S}X + e, \mathcal{S}\mathcal{S}^*v}]}}{\MAD[\inner{\mathcal{S}X + e, \mathcal{S}\mathcal{S}^*v}]} 
        =\frac{\abs{\inner{\mathcal{S}x, \mathcal{S}\mathcal{S}^*v} - \med[\inner{\mathcal{S}X, \mathcal{S}\mathcal{S}^*v}]}}{\MAD[\inner{\mathcal{S}X, \mathcal{S}\mathcal{S}^*v}]} \\
        &= \frac{\abs{\inner{x, \mathcal{S}^*v} - \med[\inner{X, \mathcal{S}^*v}]}}{\MAD[\inner{X, \mathcal{S}^*v}]}
        = O_{\mathcal{S}^*v}(x; P_X).
    \end{align*}
    This directly implies that
    \begin{equation}\label{eq: proofInv}
        \frac{1}{1+O_v(\mathcal{T}x; P_{\mathcal{T}X})} = \frac{1}{1+ O_{\mathcal{S}^*v}(x; P_X)} \quad\text{for all } v\in\SS,x\in\HH.
    \end{equation}
    To conclude the proof, observe that the regularized set of directions transforms as 
    \begin{align*}
        &\set{v\in\SS\st \MAD[\inner{\mathcal{T}X,v}]\geq q_u[\MAD[\inner{\mathcal{T}X,V}]]}\\
        &\quad=\set{v\in\SS\st \MAD[\inner{\mathcal{S}X,v}]\geq q_u[\MAD[\inner{\mathcal{S}X,V}]]}\\
        &\quad=\set{v\in\SS\st \MAD[\inner{X,\mathcal{S}^*v}]\geq q_u[\MAD[\inner{X,\mathcal{S}^*V}]]}\\
        &\quad=\set{\mathcal{S}w\st w\in\SS, \MAD[\inner{X,w}]\geq q_u[\MAD[\inner{X,\mathcal{S}^*V}]]}.
    \end{align*}

    Taking the infimum of the left-hand side of~\eqref{eq: proofInv} over this set yields $D_{(u)}(\mathcal{T}x; P_{\mathcal{T}X})$. Similarly, minimizing the right-hand side of~\eqref{eq: proofInv} over the transformed set yields $D'_{(u)}(x; P_X)$, where $D'_{(u)}$ denotes the~\RPD{} with $\beta(u)$ determined using the distribution of $\mathcal{S}^*V$.
\end{proof}

\begin{proof}[\textbf{Proof of Lemma~\ref{lemma: betaConv}}]\mbox{}
    By~\citet[Lemma~12]{Bocinec_etal2026}, it holds that
    \begin{equation*}
        \sup_{v\in\SS}\abs{\widehat{\MAD}\left[\set{\inner{X_i, v}}_{i=1}^n\right]-\MAD[\inner{X,v}]} \xasto[n\to\infty] 0.
    \end{equation*}
    This means that $\widehat{\MAD}\left[\set{\inner{X_i, V(\omega)}}_{i=1}^n\right]$ converges to $\MAD[\inner{X,V(\omega)}]$ for every $\omega\in\Omega$. Consequently $\widehat{\MAD}\left[\set{\inner{X_i, V}}_{i=1}^n\right]$ converges to $\MAD[\inner{X,V}]$ in distribution. The claim follows from~\citet[Section~2.3]{Serfling1980} using the continuity assumption.
\end{proof}

\begin{proof}[\textbf{Proof of Theorem~\ref{thm: ConsistRegQuantile}}]\mbox{}
    Using the $1$-Lipschitz continuity of the function $t\mapsto 1/(1+t), t\geq 0$, we can write
    \begin{multline*}
        \sup_{x\in\mathcal{F}}\abs{D_{(u)}(x; \widehat{P}_n) - D_{(u)}(x; P_X)} 
        \leq \sup_{x\in\mathcal{F}} \abs{\sup_{v \in \widehat{\VV}_{n(u)}} O_v(x; \widehat{P}_n) - \sup_{v \in \VV_{(u)}} O_v(x; P_X)}\\
        \leq \sup_{x\in\mathcal{F}}\abs{\sup_{v \in \widehat{\VV}_{n(u)}} O_v(x;  \widehat{P}_n) - \sup_{v \in \VV_{\widehat{\beta}_n(u)}} O_v(x; P_X)}
        + \sup_{x\in\mathcal{F}}\abs{\sup_{v \in \VV_{\widehat{\beta}_n(u)}} O_v(x; P_X) - \sup_{v \in \VV_{\beta(u)}} O_v(x; P_X)}.
    \end{multline*}
    The first term converges to zero a.s.~\citep[Theorem~7]{Bocinec_etal2026}. Consequently, the only remaining step is to show that
    \begin{equation*} 
        \abs{\sup_{v \in \VV_{\widehat{\beta}_n(u)}} O_v(x; P_X) - \sup_{v \in \VV_{\beta(u)}} O_v(x; P_X)} \xasto[n\to\infty] 0.
    \end{equation*}
    This convergence follows immediately from Lemma~\ref{lemma: betaConv} together with assumption~\eqref{consAssump} from the main manuscript.
\end{proof}

\begin{proof}[\textbf{Proof of Theorem~\ref{thm: BP}}]\mbox{}
    Fix $\varepsilon \in (0,1/2)$. Because condition~\eqref{technAssBP} in the main manuscript holds, choose $\gamma > 0$ such that $\nu(A) \geq 1 - u$, where
    \begin{equation*}
        A = \set{v \in \SS \st 
        \omega_v(\gamma) \leq \frac{1/2 - \varepsilon}{1 - \varepsilon}}.
    \end{equation*}
    For any $Q \in \P{\HH}$, $X' \sim P_{(Q,\varepsilon)}$, $v \in A$ and $a\in\R$, we have
    \begin{multline*}
        \pr\!\left(\inner{X', v} \in (a - \gamma/2, a + \gamma/2)\right)
        \leq (1 - \varepsilon)\pr\!\left(\inner{X, v} \in (a - \gamma/2, a + \gamma/2)\right) + \varepsilon \\
        \leq (1-\varepsilon)\omega_v(\gamma) + \varepsilon 
        \leq \frac{1}{2} - \varepsilon + \varepsilon = \frac{1}{2}.
    \end{multline*}
    Hence, for $v \in A$ and any $Q \in \P{\HH}$, it follows that $\MAD[\inner{X', v}] \geq \gamma$.  
    Since $\nu(A) \geq 1 - u$, we obtain 
    \begin{equation}\label{eq: beta bound}
        \inf_Q\beta_{(Q,\varepsilon)}(u) \geq \gamma.
    \end{equation}
    Next, as in~\citet[Theorem~10]{Bocinec_etal2026}, for $X' \sim P_{(Q,\varepsilon)}$ we can bound
    \begin{equation}\label{eq: BPmed}
        \med[\inner{X', v}] \leq \med[\norm{X'}] \leq q,
    \end{equation}
    where $q > 0$ denotes the $(1/2 + \varepsilon)$-quantile of $\norm{X}$.
    The upper bound~\eqref{eq: BPmed} follows from the fact that the $\varepsilon$-contaminated distribution 
    $P_{(Q,\varepsilon)}$ differs from $P_X$ by at most an $\varepsilon$-fraction of its mass. 
    Therefore, the median of $\norm{X'}$ under $P_{(Q,\varepsilon)}$ cannot exceed the $(1/2 + \varepsilon)$-quantile of $\norm{X}$ under $P_X$.
    Similarly
    \begin{equation}\label{eq: BPMAD}
        \MAD[\inner{X', v}] \leq 2\,\med[\norm{X'}] \leq 2q.
    \end{equation}

    Since $\theta(P_{(Q,\varepsilon)})$ maximizes $D_{(u)}(\cdot; P_{(Q,\varepsilon)})$ 
    over $\Clo{\spn{\VV_{(u)}(Q,\varepsilon)}}$, its outlyingness cannot exceed that of 
    $0 \in \Clo{\spn{\VV_{(u)}(Q,\varepsilon)}}$. Therefore, by~\eqref{eq: BPmed} and~\eqref{eq: BPMAD}
    \begin{align}
        \sup_{v \in \VV_{(u)}(Q,\varepsilon)} O_v(0; P_{(Q,\varepsilon)})
        &\geq \sup_{v \in \VV_{(u)}(Q,\varepsilon)} O_v(\theta(P_{(Q,\varepsilon)}); P_{(Q,\varepsilon)}) \nonumber\\
        &\geq \frac{
            \sup_{v \in \VV_{(u)}(Q,\varepsilon)} \abs{\inner{\theta(P_{(Q,\varepsilon)}), v}} 
            - \sup_{v \in \SS} \abs{\med[\inner{X', v}]}
        }{\sup_{v \in \SS} \MAD[\inner{X', v}]} \nonumber\\
        &\geq \frac{
            \sup_{v \in \VV_{(u)}(Q,\varepsilon)} \abs{\inner{\theta(P_{(Q,\varepsilon)}), v}} - q
        }{2q}.\label{eq: BP1}
    \end{align}
    On the other hand, by~\eqref{eq: beta bound} and~\eqref{eq: BPmed}
    \begin{equation}
        \sup_{v \in \VV_{(u)}(Q,\varepsilon)} O_v(0; P_{(Q,\varepsilon)}) 
        \leq \frac{\sup_{v \in \SS} \abs{\med[\inner{X', v}]}}{\beta_{(Q,\varepsilon)}(u)} 
        \leq q / \gamma.\label{eq: BP2}
    \end{equation}
    Combining the inequalities~\eqref{eq: BP1} and~\eqref{eq: BP2} yields
    \begin{equation*}
        \sup_{v \in \VV_{(u)}(Q,\varepsilon)} 
        \abs{\inner{\theta(P_{(Q,\varepsilon)}), v}} 
        \leq \frac{2q^2}{\gamma} + q < \infty,
    \end{equation*}
    proving that the \RPD{} median cannot break down under the considered contamination.
\end{proof}

\begin{proof}[\textbf{Proof of Theorem~\ref{thm: consistency}}]\mbox{}
First, because $\beta(u)$ is uniquely defined, $\beta^{(L)}(u)$ consistently estimates $\beta(u)$ in the sense that $\beta^{(L)}(u)\to\beta(u)$ a.s. as $L\to\infty$~\citep[Section~2.3]{Serfling1980}. Next, note that
\begin{multline*}
    \abs{D_{(u)}^{(M)}(x; P_X)-D_{(u)}(x; P_X)}=\abs{D_{(u)}^{(M)}(x; P_X)-D_{\beta(u)}(x; P_X)}\\
    \leq \abs{D_{(u)}^{(M)}(x; P_X)-D_{\beta^{(L)}(u)}(x; P_X)}+\abs{D_{\beta^{(L)}(u)}(x; P_X)-D_{(u)}(x; P_X)}.
\end{multline*}
The second summand equals
\begin{equation*}
    \abs{\Gamma_x(\beta^{(L)}(u))-\Gamma_x(\beta(u))},
\end{equation*}
which converges to $0$ a.s. as $L\to\infty$ since $\Gamma_x(\cdot)$ is continuous at $\beta(u)$ and by the consistency of $\beta^{(L)}(u)$. The rest of the proof follows directly from the arguments of \citet[Proposition~11]{Dyckerhoff2004}. The only difference is that we are taking infimum over $\VV_{\beta^{(L) }(u)}$, but since $\nu$ does not vanish on $\SS$, it also does not vanish on its closed non-empty subset $\VV_{\beta^{(L)}(u)}$.
\end{proof}

\newpage

\section{Robust location estimation}\label{Supplement: LocationEstimation}

In this supplementary section, we present an additional simulation study concerning robust location estimation in a functional setting. 
Specifically, we demonstrate that the median induced by \RPD{} is robust---as theoretically established in Theorem \ref{thm: BP} of the main manuscript---and achieves competitive empirical performance when compared to medians induced by alternative depth functions.

Motivated by contamination schemes commonly used in the functional data literature (see, e.g., \citet{Sinova_etal2018} and the references therein), we independently generate $n$ random curves. 
Among these, $\lceil (1-\varepsilon) n \rceil$ curves are sampled from a non-contaminated (clean) reference distribution $X$, while the remaining $\lfloor \varepsilon n \rfloor$ curves are drawn from a contaminating distribution $Y$, according to the models described below. 
We consider $\varepsilon\in\set{0.01,0.05,0.1}$ and $n\in\set{100,500,1\,000}$.
To ensure a meaningful comparison of robustness properties across different depth notions, identifying a unique population median is essential. 
Because this quantity is unambiguously defined only under symmetry assumptions, we restrict our analysis to symmetric, non-contaminated reference distributions for $X$.
Throughout all models, $X$ is defined as a Gaussian process on $[0,1]$ with mean function $\mu_X(t) = \sin(2\pi t)$ and covariance function $\Sigma(s,t) = 0.5 \exp\!\left(-|s-t|/10\right)$ for $s,t \in [0,1]$.

\begin{enumerate}[label=\textbf{Model~(L$_\arabic*$)}, ref=(L$_\arabic*$), leftmargin=*]
    \item\label{ModelL1} \emph{Shift contamination model:} 
    $Y$ is a Gaussian process with mean function $\mu_Y(t) = \sin(2\pi t) + 1$ and covariance $\Sigma$.

    \item\label{ModelL2} \emph{Amplitude contamination model:} 
    $Y \eqdis 2X$.

    \item\label{ModelL3} \emph{Phase contamination model:} 
    $Y$ is a Gaussian process with mean function $\mu_Y(t) = \sin\!\big(2\pi (t + 1/8)\big)$ and covariance $\Sigma$.

    \item\label{ModelL4} \emph{Shape contamination model:} 
    $Y$ is a Gaussian process with mean function 
    $\mu_Y(t) = \sin(2\pi t) + \sin(8\pi t)/6$ and covariance $\Sigma$.

    \item\label{ModelL5} \emph{Peak contamination model:} 
    $Y$ is a Gaussian process with mean function 
    $\mu_Y(t) = \sin(2\pi t) + 2\exp\!\big(-1\,000(t-0.8)^2\big)$ and covariance $\Sigma$.
\end{enumerate}
Illustrations of simulated datasets generated under these models are provided in Figure~\ref{fig: LocationEstimationDatasets} in Section~\ref{Supplement: Simulations}.

Our objective is to assess the robustness of various depth-induced medians under different contamination scenarios. 
Recall that by Theorem \ref{thm: BP} and the accompanying discussion, stronger regularization (a larger $u$) yields more robust location estimators. 
Therefore, in contrast to Sections~\ref{sec: OutlierDetection}, \ref{sec: Classification}, and \ref{sec: HypothesisTesting}, where $u\in\set{0.001, 0.1}$ was used, here we consider $u\in\set{0.1,0.5}$.
For each depth measure, we define the location estimator $\widehat{\mu}$ as the deepest curve in the generated sample. 
Its performance is evaluated using the mean integrated squared error ($\mathrm{MISE}$), which can be decomposed into an integrated variance term ($\mathrm{IVAR}$) and an integrated squared bias term ($\mathrm{ISB}$) as follows

\begin{equation}\label{eq: MISEiSB}
\begin{aligned}
    \mathrm{MISE}(\widehat{\mu}) & =\E\norm{\widehat{\mu}-\mu_X}^2 = \E \int_0^1\left(\widehat{\mu}(t)-\mu_X(t)\right)^2\dd t=\int_0^1\E\left(\widehat{\mu}(t)-\mu_X(t)\right)^2\dd t\\
    & =\int_0^1\var\left(\widehat{\mu}(t)-\mu_X(t)\right)\dd t+\int_0^1\left(\E\,\widehat{\mu}(t)-\mu_X(t)\right)^2\dd t=\mathrm{IVAR}(\widehat{\mu})+\mathrm{ISB}(\widehat{\mu}).
    \end{aligned}
\end{equation}

For Models \ref{ModelL1}--\ref{ModelL5}, these quantities are estimated using Monte Carlo simulations with $1\,000$ replications. 
A comprehensive summary of the results across all models and depth measures is provided in Table~\ref{tab: Location_n500} for $n=500$, and in Tables~\ref{tab: Location_n100} and \ref{tab: Location_n1000} in Section~\ref{Supplement: Simulations} for $n=100$ and $n=1\,000$, respectively.

\input{Tables_applications/Location_n500}

Here, we focus our discussion on the shape contamination scenario (Model \ref{ModelL4}) for $n=500$ to highlight the behavior of the \RPD{}‑induced median.
Contrary to classical depth measures like \MBD{} and \FD{}, which primarily achieve lower $\mathrm{MISE}$ by minimizing the $\mathrm{ISB}$, the \RPD{}‑median is characterized by an exceptionally low $\mathrm{IVAR}$.
In fact, across most contamination levels, $\RPD_{0.5}$ achieves the lowest variance among all compared methods.
While its squared bias is generally higher---leading to an overall $\mathrm{MISE}$ that is slightly larger than the best-performing \MBD{}---the \RPD{}‑median remains highly competitive.
Most importantly, the results clearly illustrate the critical role of the tuning parameter $u$.
Increasing the regularization from $u=0.1$ to $u=0.5$ substantially mitigates the bias and notably improves the overall $\mathrm{MISE}$.
This empirical behavior perfectly aligns with the theoretical findings presented in Section \ref{subsec: robustness} of the main manuscript (see Theorem \ref{thm: BP}), confirming that a larger $u$ yields a more robust location estimator.

Taken together with our previous simulation results in Section~\ref{sec: Simulations} of the main manuscript, we assert that classical depth functions often overlook shape-related features within functional data.
In contrast, the \RPD{}‑median continues to yield robust estimates by simultaneously accounting for centrality and various types of outlyingness.
Consequently, we recommend selecting the quantile level $u$ based on the underlying data structure and the primary analytical goal.
When detecting severe outlyingness is the main objective, smaller values (e.g., $u=0.001$) are preferable, as demonstrated in Sections \ref{sec: OutlierDetection}--\ref{sec: HypothesisTesting}.
Conversely, for simpler functional models with lower intrinsic complexity, larger quantile levels (such as $u=0.1$) may be more appropriate, as evidenced by the present location estimation study.

\clearpage
\newpage

\section{Simulation study: Supplementary results}\label{Supplement: Simulations}
\suppressfloats[t]
\input{Tables_applications/OutlierDetectionRanks_n100}
\input{Tables_applications/OutlierDetectionRanks_n1000}

\clearpage

\begin{figure}[htpb]
    \centering
    \includegraphics[width=\textwidth]{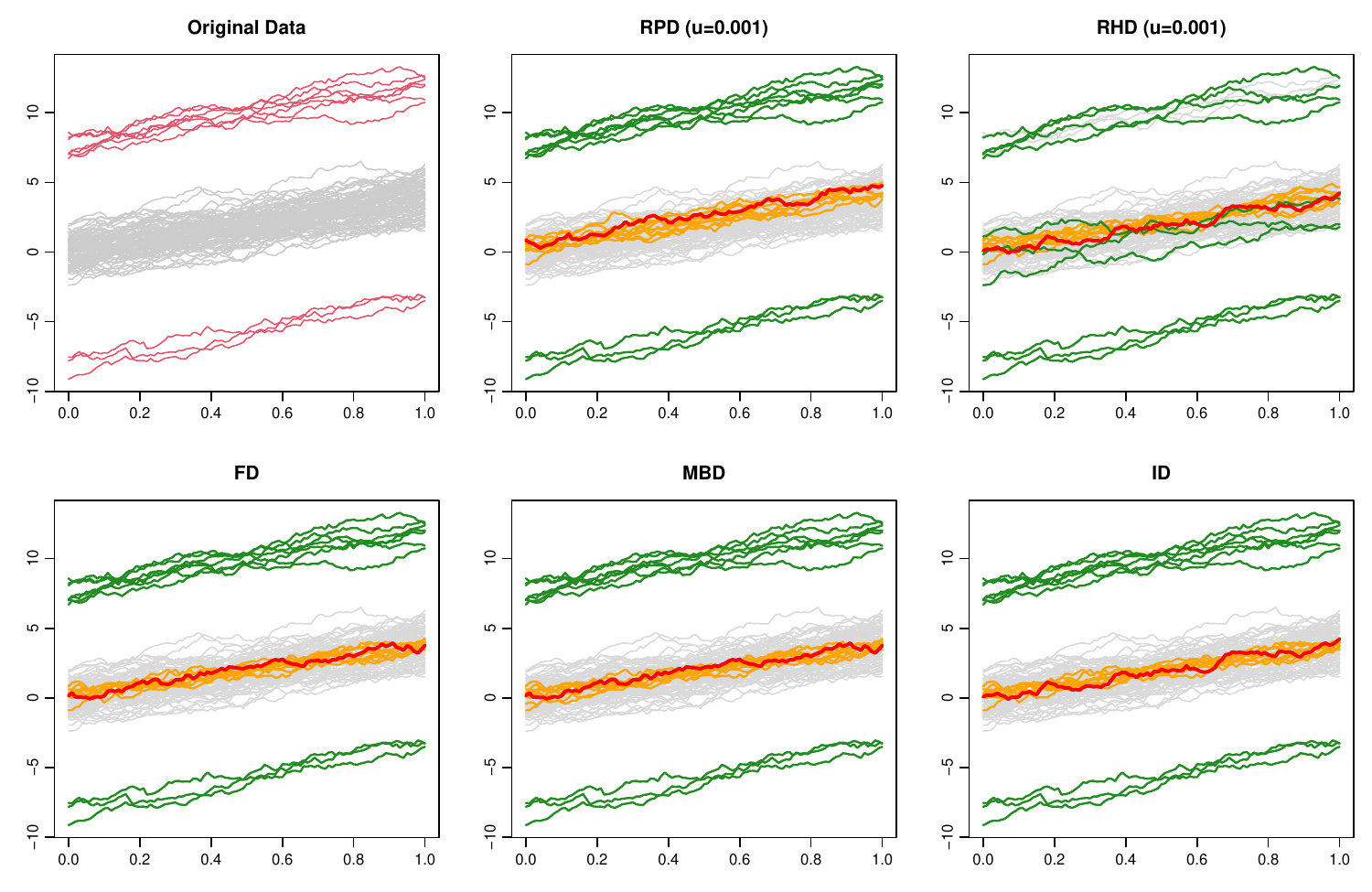}
    \caption{Outlier detection (Section~\ref{sec: OutlierDetection}): A single random sample of functional observations (gray) generated from Model~\ref{ModelD1} with sample size $n = 100$, containing $m = 10$ outlying curves (red) (top left). 
    The same sample is shown with the sample median curve highlighted in red, the $10\%$ deepest observations in orange, and the $10\%$ least deep observations in green, based on \RPD{} and \RHD{} with $u = 0.001$, \FD{}, \MBD{} and \ID{}.
    }
    \label{fig: OutlierDetDepths_D1}
\end{figure}

\begin{figure}[htpb]
    \centering
    \includegraphics[width=\textwidth]{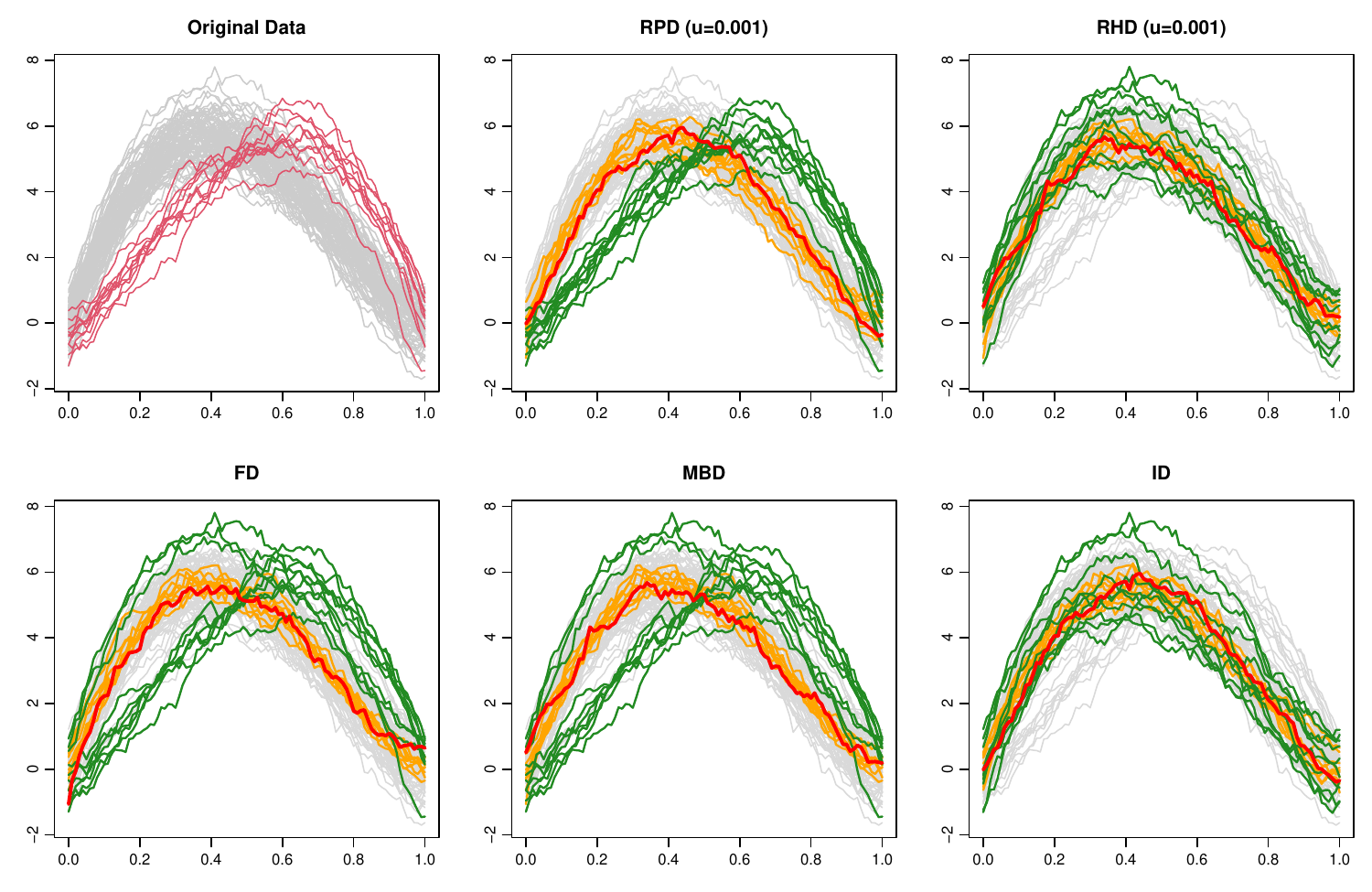}
    \caption{Outlier detection (Section~\ref{sec: OutlierDetection}): A single random sample of functional observations (gray) generated from Model~\ref{ModelD2} with sample size $n = 100$, containing $m = 10$ outlying curves (red) (top left). 
    The same sample is shown with the sample median curve highlighted in red, the $10\%$ deepest observations in orange, and the $10\%$ least deep observations in green, based on \RPD{} and \RHD{} with $u = 0.001$, \FD{}, \MBD{} and \ID{}.
    }
    \label{fig: OutlierDetDepths_D2}
\end{figure}

\begin{figure}[htpb]
    \centering
    \includegraphics[width=\textwidth]{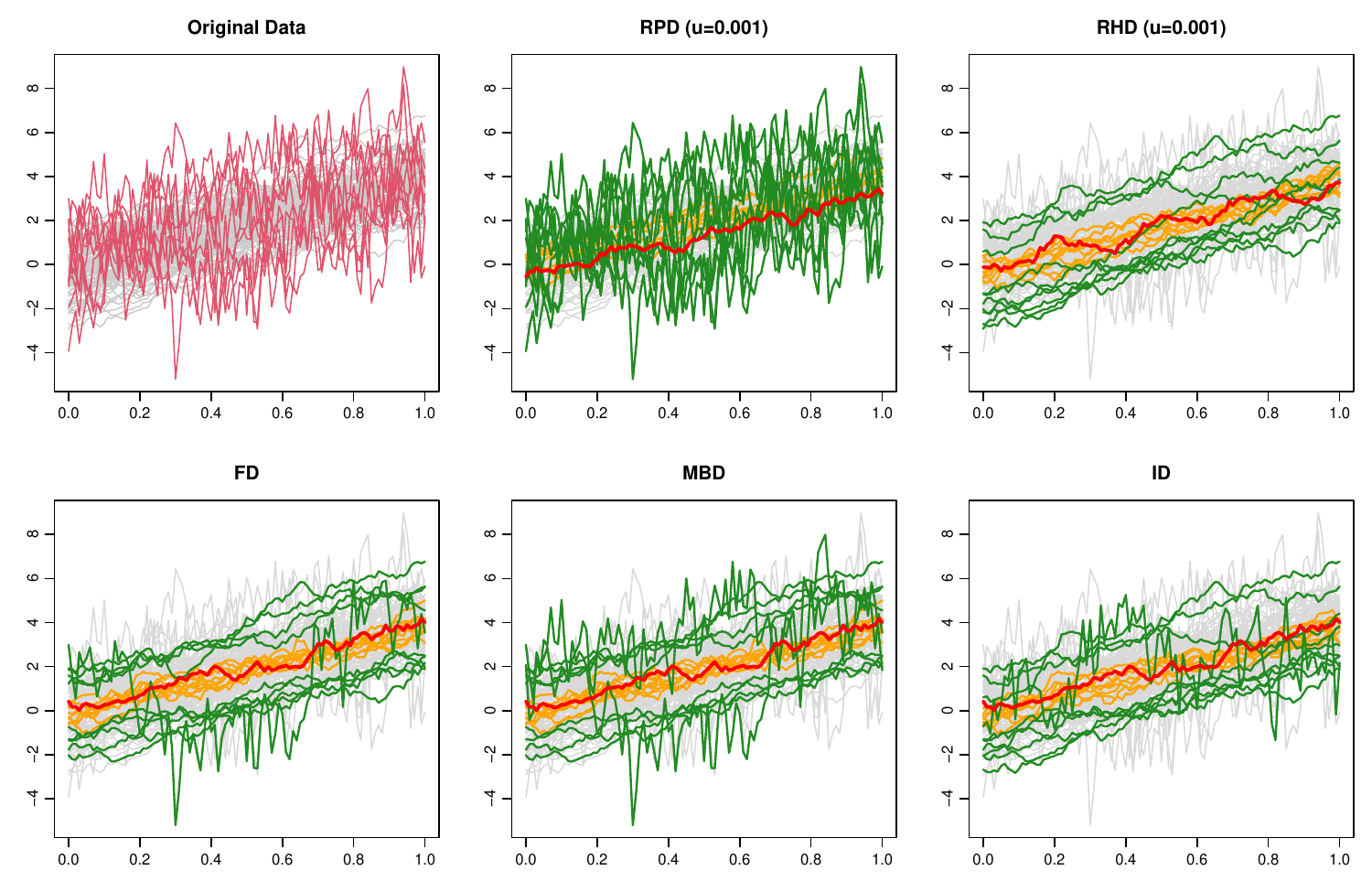}
    \caption{Outlier detection (Section~\ref{sec: OutlierDetection}): A single random sample of functional observations (gray) generated from Model~\ref{ModelD3} with sample size $n = 100$, containing $m = 10$ outlying curves (red) (top left). 
    The same sample is shown with the sample median curve highlighted in red, the $10\%$ deepest observations in orange, and the $10\%$ least deep observations in green, based on \RPD{} and \RHD{} with $u = 0.001$, \FD{}, \MBD{} and \ID{}.
    }
    \label{fig: OutlierDetDepths_D3}
\end{figure}

\begin{figure}[htpb]
    \centering
    \includegraphics[width=\textwidth]{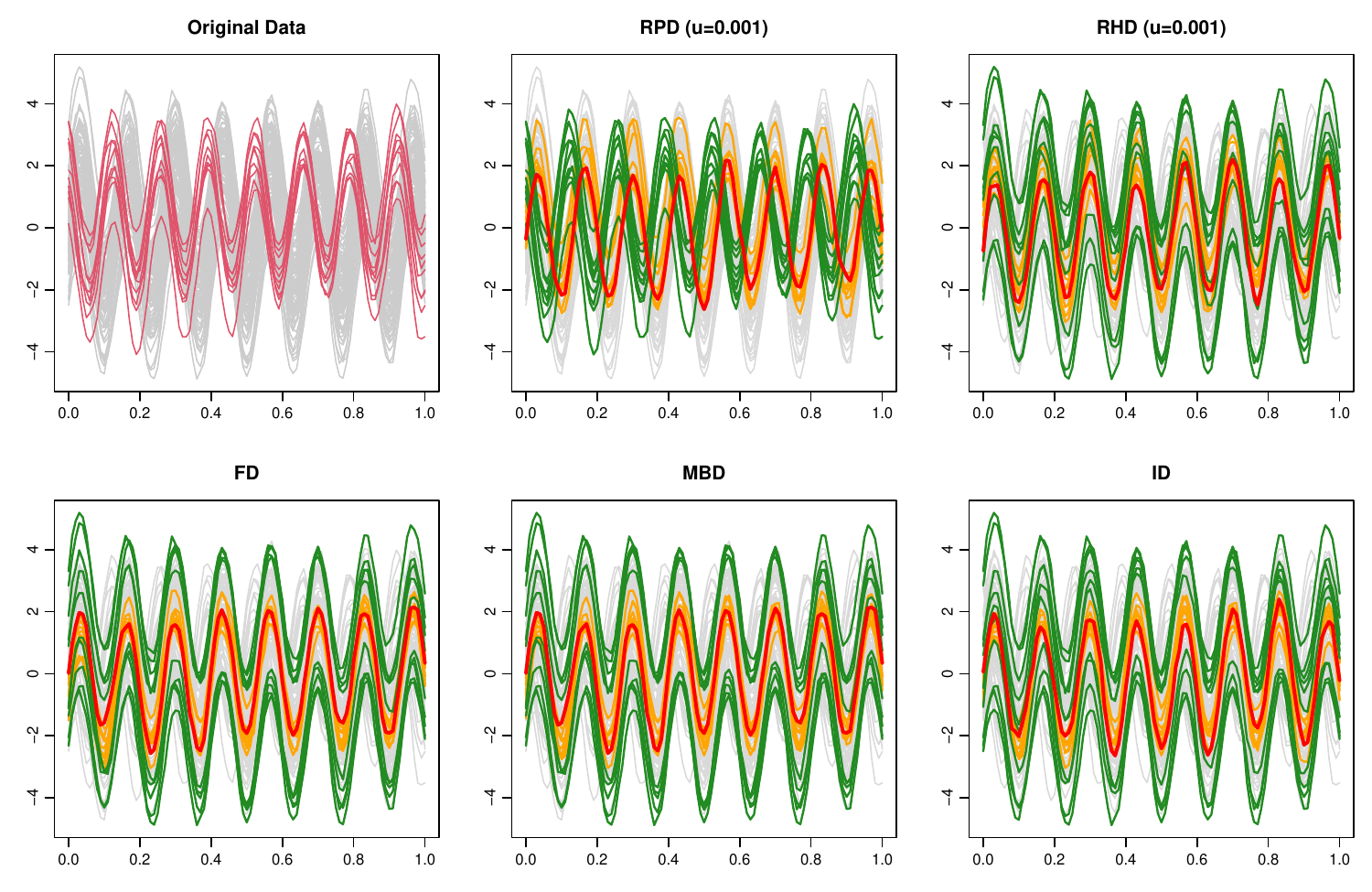}
    \caption{Outlier detection (Section~\ref{sec: OutlierDetection}): A single random sample of functional observations (gray) generated from Model~\ref{ModelD4} with sample size $n = 100$, containing $m = 10$ outlying curves (red) (top left). 
    The same sample is shown with the sample median curve highlighted in red, the $10\%$ deepest observations in orange, and the $10\%$ least deep observations in green, based on \RPD{} and \RHD{} with $u = 0.001$, \FD{}, \MBD{} and \ID{}.
    }
    \label{fig: OutlierDetDepths_D4}
\end{figure}

\begin{figure}[htpb]
    \centering
    \includegraphics[width=\textwidth]{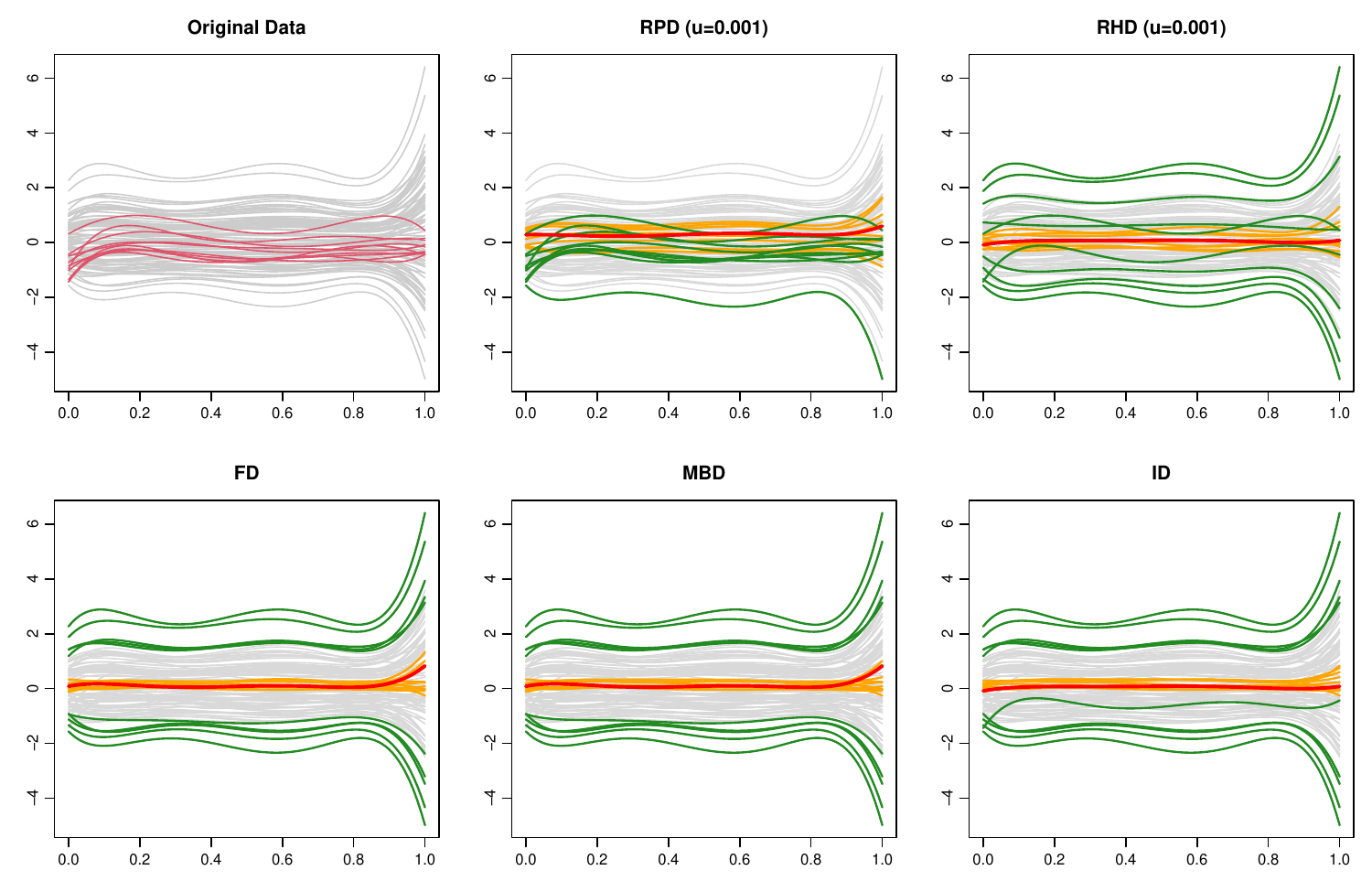}
    \caption{Outlier detection (Section~\ref{sec: OutlierDetection}): A single random sample of functional observations (gray) generated from Model~\ref{ModelD6} with sample size $n = 100$, containing $m = 10$ outlying curves (red) (top left). 
    The same sample is shown with the sample median curve highlighted in red, the $10\%$ deepest observations in orange, and the $10\%$ least deep observations in green, based on \RPD{} and \RHD{} with $u = 0.001$, \FD{}, \MBD{} and \ID{}.
    }
    \label{fig: OutlierDetDepths_D6}
\end{figure}

\clearpage

\input{Tables_applications/OutlierDetection_n100}
\input{Tables_applications/OutlierDetection_n1000}

\input{Tables_applications/Classification_n100}
\input{Tables_applications/Classification_n1000}

\clearpage

\begin{figure}
    \centering
    \includegraphics[width=0.9\linewidth]{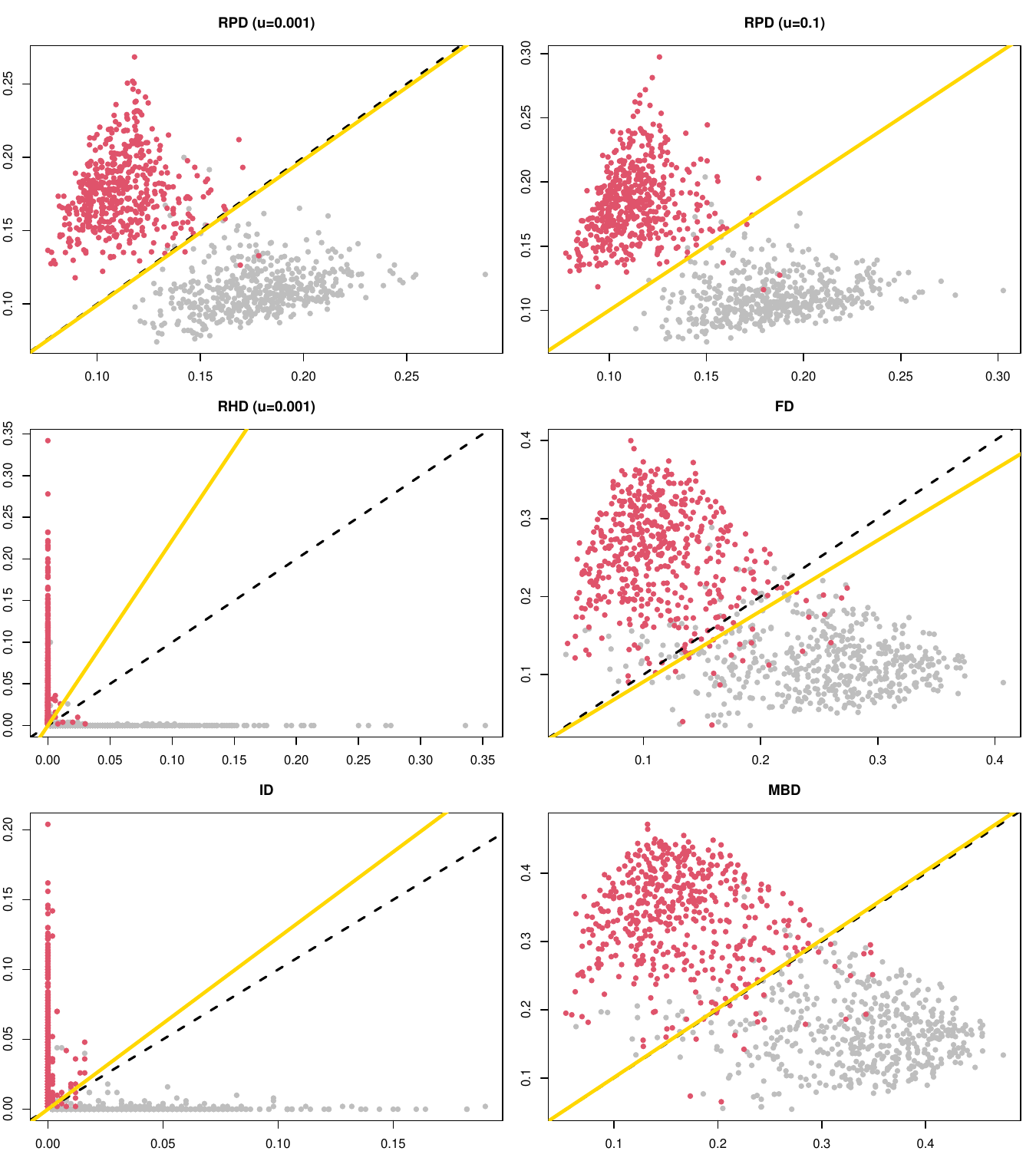}
    \caption{Supervised classification (Section~\ref{sec: Classification}): DD-plots of the training data in a single run of Model~\ref{ModelC1} (gray points for $X_i$, red points for $Y_i$) with (a) \RPD{} (top left, $u = 0.001$), (b) \RPD{} (top right, $u = 0.1$), (c) \RHD{} (middle left, $u = 0.001$), (d) \FD{} (middle right), (e) \ID{} (bottom left), and (f) \MBD{} (bottom right). In the plots, the dashed line represents the max-depth classifier and the yellow line the linear DD-classifier best separating the two groups.}
\end{figure}

\begin{figure}
    \centering
    \includegraphics[width=0.9\linewidth]{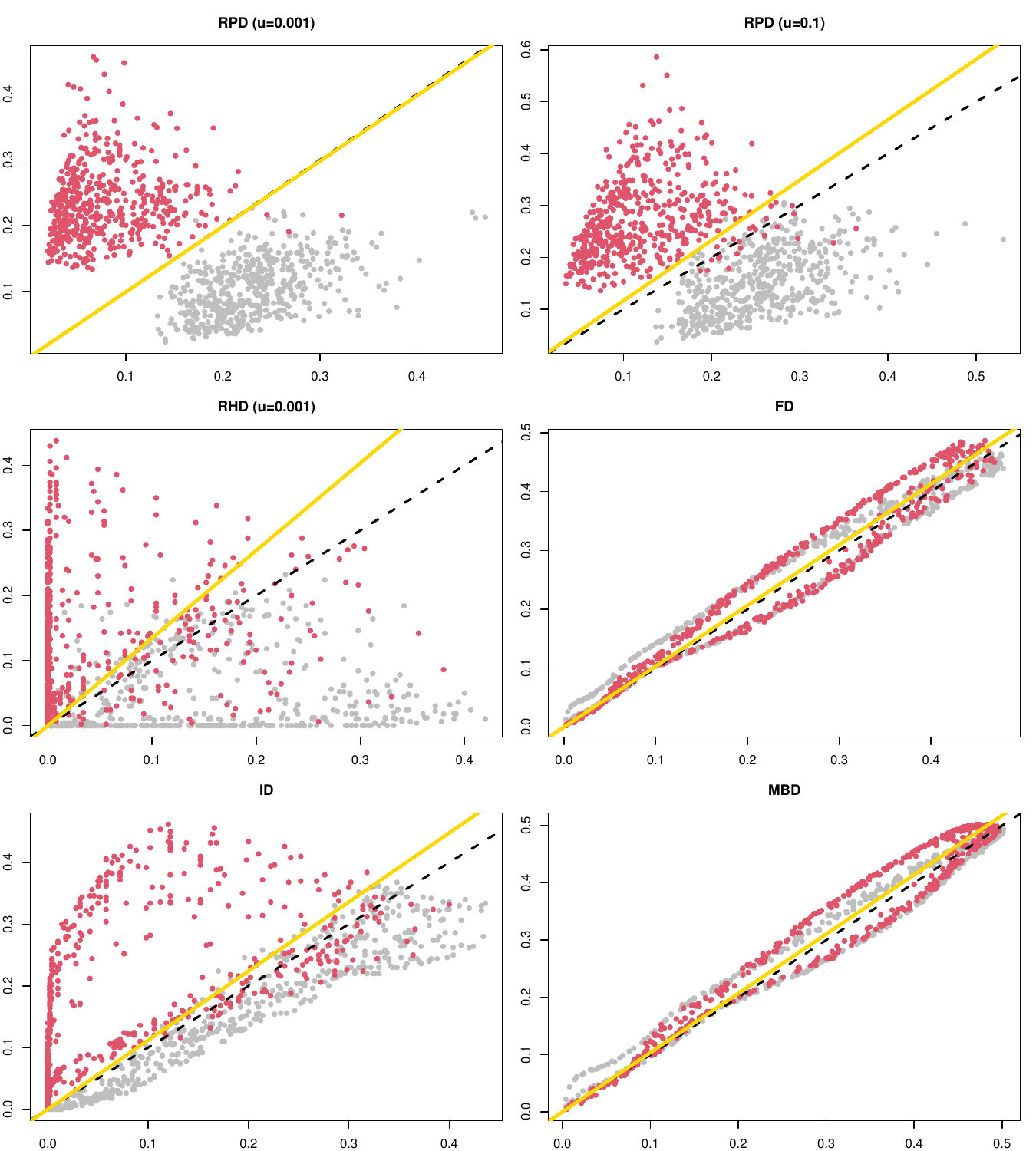}
    \caption{Supervised classification (Section~\ref{sec: Classification}): DD-plots of the training data in a single run of Model~\ref{ModelC2} (gray points for $X_i$, red points for $Y_i$) with (a) \RPD{} (top left, $u = 0.001$), (b) \RPD{} (top right, $u = 0.1$), (c) \RHD{} (middle left, $u = 0.001$), (d) \FD{} (middle right), (e) \ID{} (bottom left), and (f) \MBD{} (bottom right). In the plots, the dashed line represents the max-depth classifier and the yellow line the linear DD-classifier best separating the two groups.}
\end{figure}

\clearpage

\begin{figure}
    \centering
\includegraphics[width=0.99\textwidth]{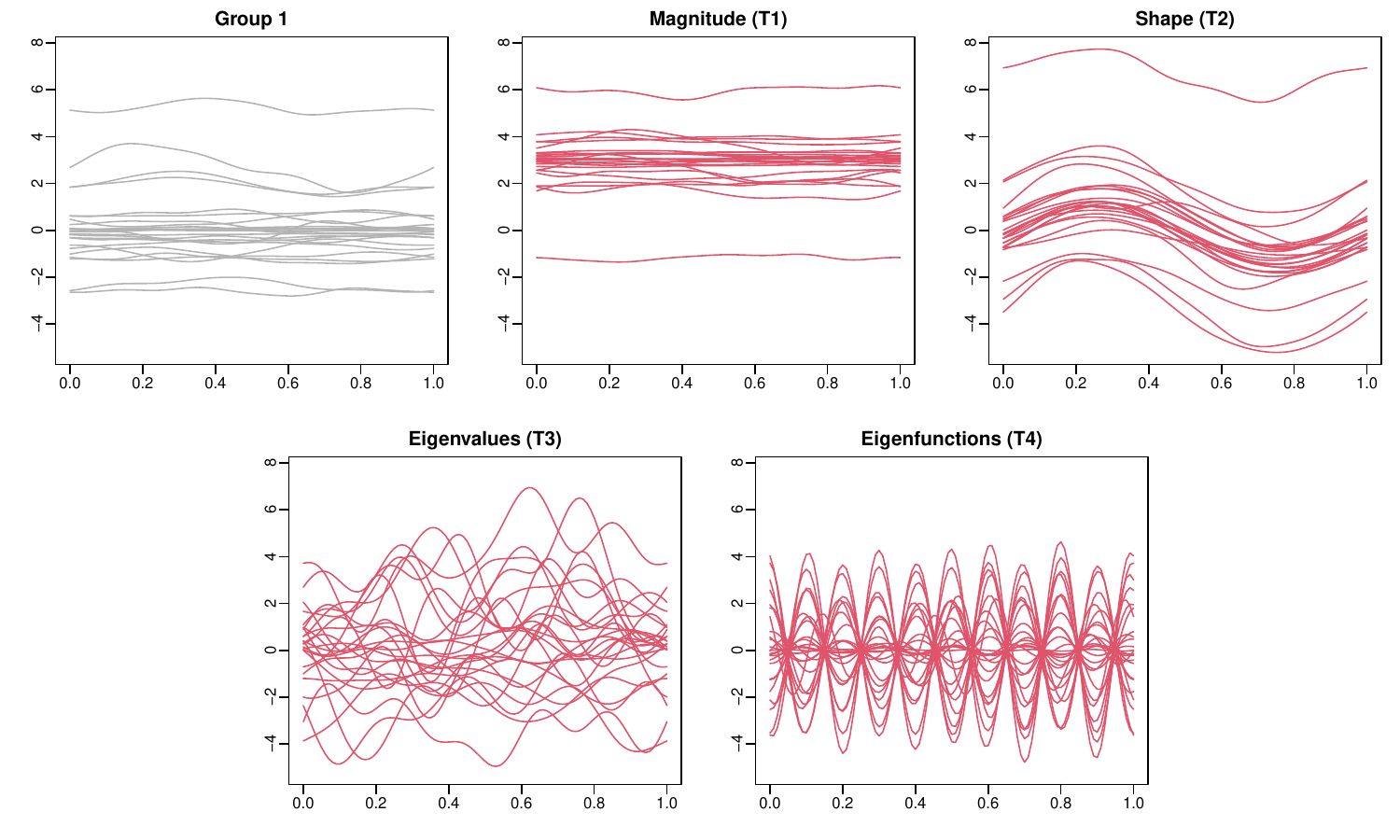}
    \caption{
        Hypothesis testing (Section~\ref{sec: HypothesisTesting}): Simulated datasets generated under Models~\ref{ModelT1}–\ref{ModelT4}, as described in Section~\ref{sec: HypothesisTesting} of the main manuscript. Each panel contains $50$ functional observations. The top left panel represents the reference group, while the remaining panels display the second group of functions differing in magnitude, shape, eigenvalues, and eigenfunctions, respectively.
    }
    \label{fig_2test_real}
\end{figure}

\begin{figure}
    \centering
\includegraphics[width=0.8\textwidth]{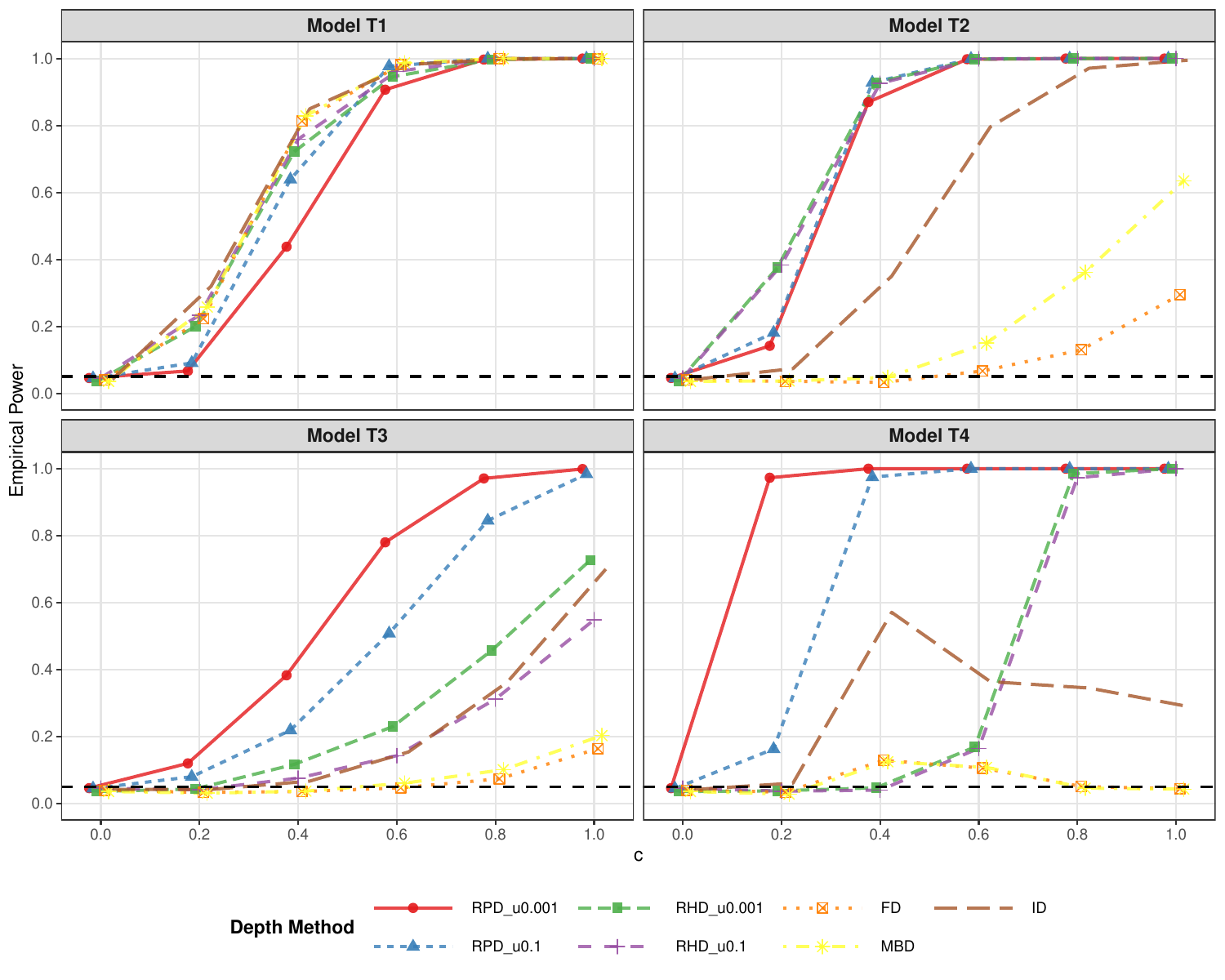}
    \caption{
        Hypothesis testing (Section~\ref{sec: HypothesisTesting}): Empirical power of the depth-based Kruskal--Wallis tests for sample size $n=50$ under Models~\ref{ModelT1}--\ref{ModelT4}. 
        The horizontal axis represents the intensity parameter $c$ controlling the magnitude of the distributional difference between two samples, while the vertical axis shows the proportion of rejected null hypotheses over $1\,000$ Monte Carlo replications. 
        Higher curves indicate greater ability to detect departures from the null hypothesis. The horizontal dashed line indicates the nominal type I error rate of $0.05$. 
    }
    \label{fig: KW_test_n50}
\end{figure}

\begin{figure}
    \centering
\includegraphics[width=0.8\textwidth]{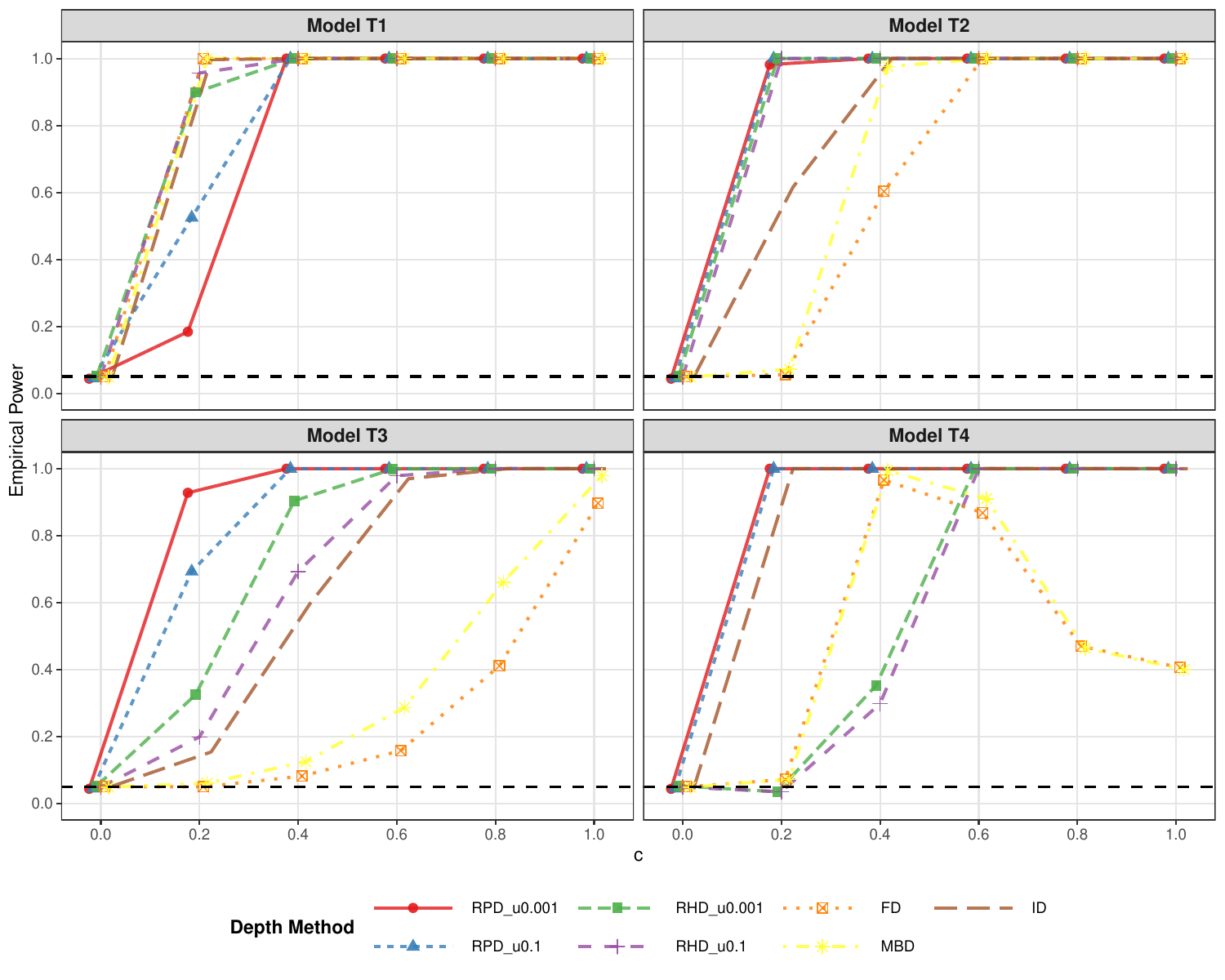}
    \caption{
        Hypothesis testing (Section~\ref{sec: HypothesisTesting}): Empirical power of the depth-based Kruskal--Wallis tests for sample size $n=500$ under Models~\ref{ModelT1}--\ref{ModelT4}. 
        The horizontal axis represents the intensity parameter $c$ controlling the magnitude of the distributional difference between two samples, while the vertical axis shows the proportion of rejected null hypotheses over $1\,000$ Monte Carlo replications. 
        Higher curves indicate greater ability to detect departures from the null hypothesis. The horizontal dashed line indicates the nominal type I error rate of $0.05$. 
    }
    \label{fig: KW_test_n500}
\end{figure}

\begin{figure}
    \centering
    \includegraphics[width=0.8\textwidth]{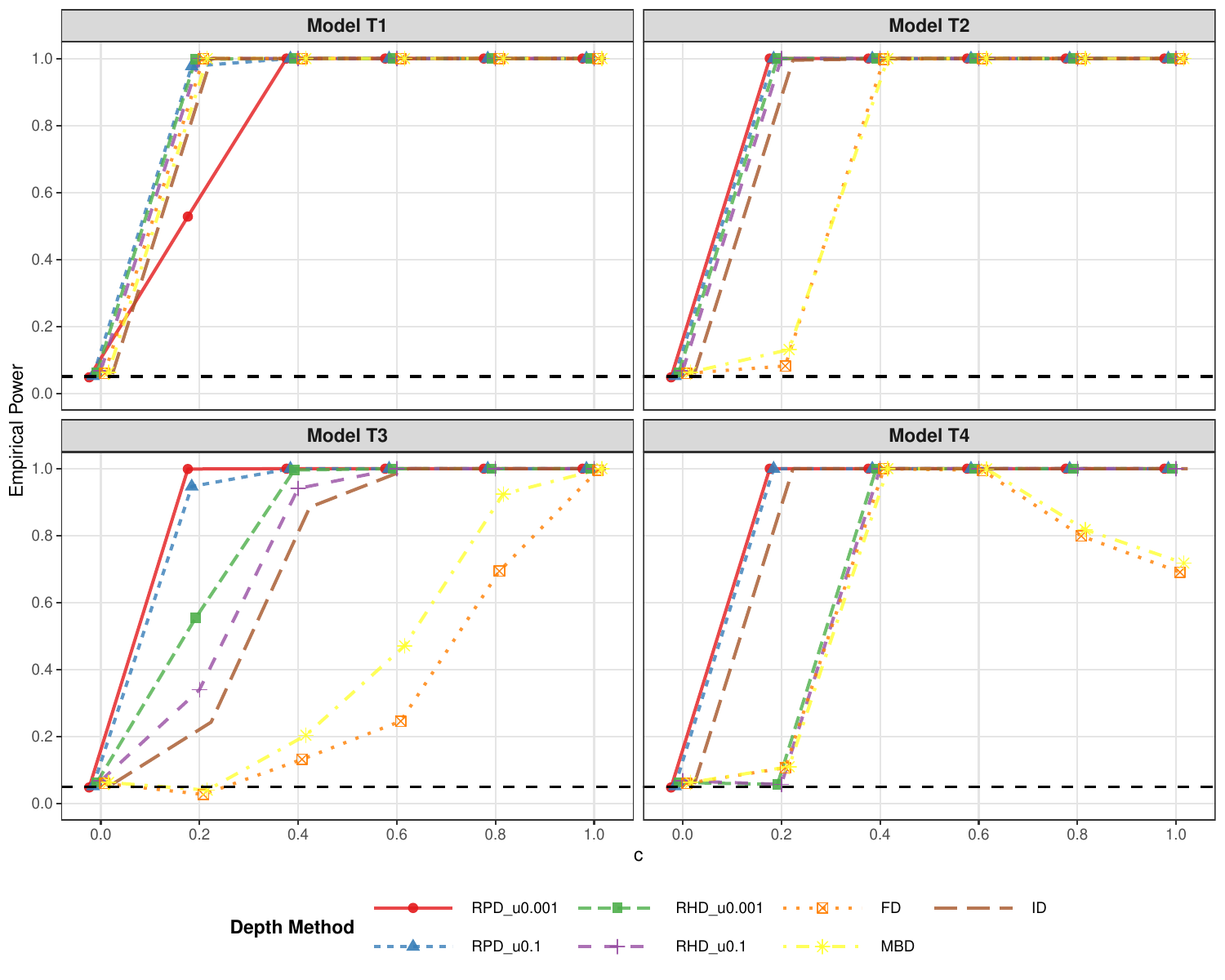}
    \caption{
        Hypothesis testing (Section~\ref{sec: HypothesisTesting}): Empirical power of the depth-based Kruskal--Wallis tests for sample size $n=1\,000$ under Models~\ref{ModelT1}--\ref{ModelT4}. 
        The horizontal axis represents the intensity parameter $c$ controlling the magnitude of the distributional difference between two samples, while the vertical axis shows the proportion of rejected null hypotheses over $1\,000$ Monte Carlo replications. 
        Higher curves indicate greater ability to detect departures from the null hypothesis. The horizontal dashed line indicates the nominal type I error rate of $0.05$. 
    }
    \label{fig: KW_test_n1000}
\end{figure}

\clearpage

\input{Tables_applications/KW_test_full}

\input{Tables_applications/Real_World_Application}

\clearpage

\input{Tables_applications/Location_n100}
\input{Tables_applications/Location_n1000}

\begin{figure}[htbp]
    \centering
    \includegraphics[width=0.98\textwidth]{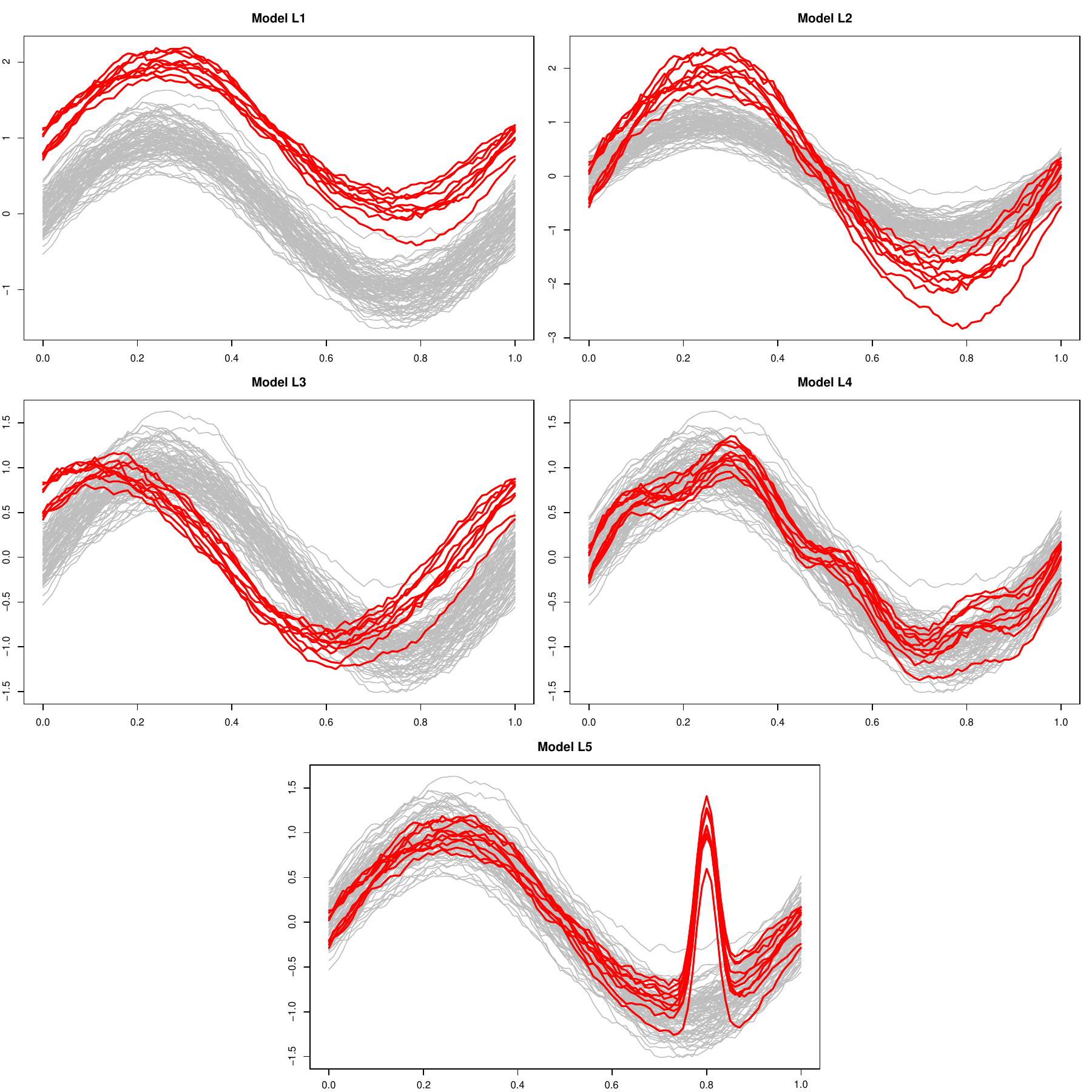}
    \caption{Robust location estimation (Section~\ref{Supplement: LocationEstimation}): Simulated datasets generated under the contamination Models~\ref{ModelL1}--\ref{ModelL5} described in Section~\ref{Supplement: LocationEstimation}. Each panel contains $100$ functions, of which $10$ are contaminated. The non-contaminated functions are shown in gray, while the contaminating functions are highlighted in red.
    \label{fig: LocationEstimationDatasets}}
\end{figure}

\clearpage
\putbib
\phantomsection\label{supp:LastPage}
\end{bibunit}
\end{document}

%% file: Tables_applications/OutlierDetectionRanks_n500.tex
\begin{table}
\centering
\resizebox{0.9\textwidth}{!}{%
\begin{tabular}{cccccccc}
\toprule
Model & $\RPD_{0.001}$ & $\RPD_{0.1}$ & $\RHD_{0.001}$ & $\RHD_{0.1}$ & \FD{} & \MBD{} & \ID{} \\
\midrule

\multicolumn{8}{c}{\textbf{Contamination level $\varepsilon = 0.01$ ($m = 5$, $n = 500$)}} \\
\midrule
\ref{ModelD1} & \textbf{0.006} \,(0.000) & \textbf{0.006} \,(0.000) & 0.019 \,(0.006) & 0.015 \,(0.005) & \textbf{0.006} \,(0.000) & \textbf{0.006} \,(0.000) & 0.009 \,(0.006) \\ 
\ref{ModelD2} & \textbf{0.006} \,(0.001) & \textbf{0.006} \,(0.000) & 0.060 \,(0.014) & 0.051 \,(0.013) & 0.013 \,(0.005) & 0.011 \,(0.004) & 0.050 \,(0.013) \\ 
\ref{ModelD3} & \textbf{0.006} \,(0.000) & \textbf{0.006} \,(0.000) & 0.047 \,(0.032) & 0.075 \,(0.054) & 0.163 \,(0.034) & 0.134 \,(0.030) & 0.023 \,(0.004) \\ 
\ref{ModelD4} & \textbf{0.006} \,(0.000) & \textbf{0.006} \,(0.000) & 0.036 \,(0.006) & 0.030 \,(0.006) & 0.079 \,(0.010) & 0.060 \,(0.009) & 0.028 \,(0.006) \\ 
\ref{ModelD5} & \textbf{0.006} \,(0.000) & 0.007 \,(0.002) & 0.986 \,(0.011) & 0.980 \,(0.015) & 0.510 \,(0.024) & 0.517 \,(0.029) & 0.571 \,(0.061) \\ 
\ref{ModelD6} & \textbf{0.006} \,(0.003) & 0.014 \,(0.014) & 0.221 \,(0.179) & 0.720 \,(0.083) & 0.716 \,(0.068) & 0.683 \,(0.065) & 0.349 \,(0.082) \\
\midrule

\multicolumn{8}{c}{\textbf{Contamination level $\varepsilon = 0.05$ ($m = 25$, $n = 500$)}} \\
\midrule
\ref{ModelD1} & \textbf{0.026} \,(0.000) & \textbf{0.026} \,(0.000) & 0.044 \,(0.005) & 0.039 \,(0.004) & \textbf{0.026} \,(0.000) & \textbf{0.026} \,(0.000) & 0.027 \,(0.003) \\ 
\ref{ModelD2} & \textbf{0.026} \,(0.000) & \textbf{0.026} \,(0.000) & 0.135 \,(0.015) & 0.124 \,(0.013) & 0.035 \,(0.004) & 0.033 \,(0.003) & 0.127 \,(0.016) \\ 
\ref{ModelD3} & \textbf{0.026} \,(0.000) & \textbf{0.026} \,(0.000) & 0.062 \,(0.016) & 0.089 \,(0.025) & 0.173 \,(0.017) & 0.144 \,(0.015) & 0.034 \,(0.003) \\ 
\ref{ModelD4} & \textbf{0.026} \,(0.000) & \textbf{0.026} \,(0.000) & 0.096 \,(0.010) & 0.090 \,(0.009) & 0.101 \,(0.009) & 0.084 \,(0.008) & 0.084 \,(0.009) \\ 
\ref{ModelD5} & \textbf{0.026} \,(0.000) & 0.028 \,(0.002) & 0.374 \,(0.107) & 0.820 \,(0.039) & 0.510 \,(0.021) & 0.517 \,(0.022) & 0.568 \,(0.031) \\ 
\ref{ModelD6} & \textbf{0.026} \,(0.002) & 0.035 \,(0.008) & 0.183 \,(0.104) & 0.689 \,(0.040) & 0.709 \,(0.032) & 0.681 \,(0.031) & 0.378 \,(0.036) \\
\midrule

\multicolumn{8}{c}{\textbf{Contamination level $\varepsilon = 0.10$ ($m = 50$, $n = 500$)}} \\
\midrule
\ref{ModelD1} & \textbf{0.051} \,(0.000) & \textbf{0.051} \,(0.000) & 0.077 \,(0.007) & 0.067 \,(0.007) & \textbf{0.051} \,(0.000) & \textbf{0.051} \,(0.000) & 0.052 \,(0.002) \\ 
\ref{ModelD2} & \textbf{0.051} \,(0.000) & \textbf{0.051} \,(0.000) & 0.208 \,(0.016) & 0.194 \,(0.014) & 0.064 \,(0.004) & 0.063 \,(0.003) & 0.204 \,(0.017) \\ 
\ref{ModelD3} & \textbf{0.051} \,(0.000) & \textbf{0.051} \,(0.000) & 0.086 \,(0.012) & 0.111 \,(0.018) & 0.188 \,(0.014) & 0.160 \,(0.012) & 0.055 \,(0.002) \\ 
\ref{ModelD4} & \textbf{0.051} \,(0.000) & \textbf{0.051} \,(0.000) & 0.167 \,(0.012) & 0.161 \,(0.011) & 0.129 \,(0.009) & 0.115 \,(0.008) & 0.148 \,(0.011) \\ 
\ref{ModelD5} & \textbf{0.051} \,(0.000) & 0.055 \,(0.003) & 0.180 \,(0.023) & 0.532 \,(0.055) & 0.510 \,(0.021) & 0.516 \,(0.021) & 0.565 \,(0.024) \\ 
\ref{ModelD6} & \textbf{0.051} \,(0.001) & 0.063 \,(0.008) & 0.226 \,(0.068) & 0.657 \,(0.042) & 0.701 \,(0.024) & 0.678 \,(0.023) & 0.415 \,(0.026) \\
\bottomrule
\end{tabular}}
\caption{Outlier detection: Means and standard deviations (in parentheses) of the average rank of the outlying curves for a sample size of $n=500$ across three contamination levels. Lower values indicate better performance. The smallest average ranks achieved are highlighted in bold. In the ideal case, the $m$ curves with the smallest depth values coincide exactly with the true outliers, yielding a minimum possible average outlier rank of $0.006$ for $m=5$, $0.026$ for $m=25$, and $0.051$ for $m=50$.}
\label{tab: OutlierDetection_n500}
\end{table}

%% file: Tables_applications/OutlierDetection_n500.tex
\begin{table}
    \centering
    \resizebox{0.6\textheight}{!}{%
    \begin{subtable}{\textwidth}
        \centering
        \begin{tabular}{cc rrrr}
            \toprule
            \textbf{Model} & $\boldsymbol{m}$ & \RPDOD{} & Functional boxplot & Outliergram & MS-plot \\
            \midrule
            \multirow{4}{*}{\ref{ModelD1}} 
                & 0  & \textbf{0.000 (0.000)} & \textbf{0.000 (0.000)} & 1.000 (0.000) & 1.000 (0.000) \\
                & 5  & \textbf{0.000 (0.005)} & \textbf{0.000 (0.005)} & 1.000 (0.000) & 0.786 (0.055) \\
                & 25 & \textbf{0.000 (0.002)} & \textbf{0.000 (0.000)} & 1.000 (0.000) & 0.401 (0.074) \\
                & 50 & \textbf{0.000 (0.000)} & \textbf{0.000 (0.000)} & 1.000 (0.000) & 0.220 (0.055) \\
            \midrule
            \multirow{4}{*}{\ref{ModelD2}} 
                & 0  & 0.001 (0.032) & \textbf{0.000 (0.000)} & 1.000 (0.000) & 0.999 (0.032) \\
                & 5  & 0.007 (0.036) & \textbf{0.000 (0.000)} & 0.652 (0.089) & 0.627 (0.107) \\
                & 25 & 0.010 (0.021) & \textbf{0.000 (0.000)} & 0.209 (0.072) & 0.237 (0.081) \\
                & 50 & 0.015 (0.020) & \textbf{0.000 (0.000)} & 0.078 (0.036) & 0.119 (0.049) \\
            \midrule
            \multirow{4}{*}{\ref{ModelD3}} 
                & 0  & \textbf{0.000 (0.000)} & \textbf{0.000 (0.000)} & 1.000 (0.000) & 1.000 (0.000) \\
                & 5  & \textbf{0.000 (0.005)} & \textbf{0.000 (0.000)} & 0.831 (0.046) & 0.790 (0.057) \\
                & 25 & \textbf{0.000 (0.000)} & \textbf{0.000 (0.000)} & 0.420 (0.070) & 0.405 (0.079) \\
                & 50 & \textbf{0.000 (0.000)} & \textbf{0.000 (0.000)} & 0.179 (0.053) & 0.223 (0.059) \\
            \midrule
            \multirow{4}{*}{\ref{ModelD4}} 
                & 0  & \textbf{0.002 (0.045)} & 0.003 (0.055) & 1.000 (0.000) & 1.000 (0.000) \\
                & 5  & \textbf{0.000 (0.007)} & 0.002 (0.039) & 0.803 (0.040) & 0.786 (0.060) \\
                & 25 & \textbf{0.000 (0.000)} & 0.001 (0.019) & 0.364 (0.065) & 0.398 (0.078) \\
                & 50 & \textbf{0.000 (0.000)} & \textbf{0.000 (0.000)} & 0.133 (0.044) & 0.219 (0.058) \\
            \midrule
            \multirow{4}{*}{\ref{ModelD5}} 
                & 0  & \textbf{0.000 (0.000)} & 0.001 (0.032) & 1.000 (0.000) & 1.000 (0.000) \\
                & 5  & 0.004 (0.027) & \textbf{0.001 (0.032)} & 0.956 (0.054) & 1.000 (0.000) \\
                & 25 & 0.005 (0.015) & \textbf{0.001 (0.032)} & 0.975 (0.069) & 1.000 (0.000) \\
                & 50 & 0.010 (0.017) & \textbf{0.002 (0.045)} & 0.997 (0.055) & 1.000 (0.000) \\
            \midrule
            \multirow{4}{*}{\ref{ModelD6}} 
                & 0  & \textbf{0.012 (0.109)} & 0.858 (0.349) & 1.000 (0.000) & 0.973 (0.162) \\
                & 5  & \textbf{0.000 (0.000)} & 0.846 (0.360) & 0.791 (0.047) & 0.429 (0.156) \\
                & 25 & \textbf{0.000 (0.004)} & 0.796 (0.402) & 0.351 (0.073) & 0.129 (0.074) \\
                & 50 & \textbf{0.003 (0.008)} & 0.782 (0.413) & 0.115 (0.045) & 0.070 (0.051) \\
            \bottomrule
        \end{tabular}
        \caption{False Discovery Rate (FDR) --- Lower is better}
    \end{subtable}%
    }
    \resizebox{0.6\textheight}{!}{%
    \begin{subtable}{\textwidth}
        \centering
        \begin{tabular}{cc rrrr}
            \toprule
            \textbf{Model} & $\boldsymbol{m}$ & \RPDOD{} & Functional boxplot & Outliergram & MS-plot \\
            \midrule
            \multirow{3}{*}{\ref{ModelD1}} 
                & 5  & 0.969 (0.173) & \textbf{1.000 (0.003)} & -0.022 (0.003) & 0.450 (0.058) \\
                & 25 & \textbf{1.000 (0.001)} & \textbf{1.000 (0.001)} & -0.049 (0.006) & 0.758 (0.051) \\
                & 50 & \textbf{1.000 (0.001)} & \textbf{1.000 (0.001)} & -0.070 (0.009) & 0.869 (0.036) \\
            \midrule
            \multirow{3}{*}{\ref{ModelD2}} 
                & 5  & \textbf{0.859 (0.327)} & 0.018 (0.088) & 0.578 (0.075) & 0.599 (0.087) \\
                & 25 & \textbf{0.985 (0.019)} & 0.002 (0.020) & 0.879 (0.043) & 0.864 (0.050) \\
                & 50 & \textbf{0.979 (0.016)} & 0.000 (0.004) & 0.935 (0.024) & 0.931 (0.030) \\
            \midrule
            \multirow{3}{*}{\ref{ModelD3}} 
                & 5  & \textbf{0.988 (0.109)} & 0.543 (0.238) & 0.363 (0.081) & 0.445 (0.061) \\
                & 25 & \textbf{1.000 (0.000)} & 0.552 (0.090) & 0.668 (0.064) & 0.756 (0.055) \\
                & 50 & \textbf{1.000 (0.000)} & 0.504 (0.068) & 0.778 (0.046) & 0.867 (0.038) \\
            \midrule
            \multirow{3}{*}{\ref{ModelD4}} 
                & 5  & \textbf{1.000 (0.004)} & 0.242 (0.259) & 0.433 (0.045) & 0.450 (0.063) \\
                & 25 & \textbf{1.000 (0.000)} & 0.270 (0.131) & 0.784 (0.044) & 0.760 (0.054) \\
                & 50 & \textbf{1.000 (0.000)} & 0.231 (0.092) & 0.921 (0.027) & 0.869 (0.037) \\
            \midrule
            \multirow{3}{*}{\ref{ModelD5}} 
                & 5  & \textbf{0.852 (0.351)} & 0.000 (0.000) & 0.080 (0.127) & -0.020 (0.003) \\
                & 25 & \textbf{0.997 (0.009)} & 0.000 (0.000) & -0.015 (0.070) & -0.042 (0.007) \\
                & 50 & \textbf{0.994 (0.010)} & 0.000 (0.001) & -0.041 (0.010) & -0.053 (0.012) \\
            \midrule
            \multirow{3}{*}{\ref{ModelD6}} 
                & 5  & 0.060 (0.234) & -0.005 (0.013) & 0.445 (0.051) & \textbf{0.744 (0.104)} \\
                & 25 & 0.899 (0.268) & -0.011 (0.009) & 0.792 (0.049) & \textbf{0.922 (0.044)} \\
                & 50 & \textbf{0.975 (0.071)} & -0.016 (0.011) & 0.922 (0.029) & 0.913 (0.041) \\
            \bottomrule
        \end{tabular}
        \caption{Matthews Correlation Coefficient (MCC) --- Higher is better}
        \vspace{1em}
    \end{subtable}
    }
    \caption{Outlier detection: Means and standard deviations (in parentheses) of~$\mathrm{MCC}$ and $\mathrm{FDR}$ of outlier detection methods described in Section~\ref{subsec: threshold}. Sample size $n=500$ across four contamination levels $m\in\set{0,5,25,50}$ in models~\ref{ModelD1}--\ref{ModelD6} is considered. Best results are highlighted in bold.}
    \label{tab: outlierDetectionResults_n500}
\end{table}

%% file: Tables_applications/Classification_n500.tex
\begin{table}
\centering
\resizebox{0.9\textwidth}{!}{%
\begin{tabular}{c c c c c c c c c}
\toprule
Model & Classifier & $\RPD_{0.001}$ & $\RPD_{0.1}$ & $\RHD_{0.001}$ & $\RHD_{0.1}$ & \FD{} & \MBD{} & \ID{} \\
\midrule
\multirow{2}{*}{\ref{ModelC1}} 
& max & \textbf{0.013} \,(0.003) & 0.014 \,(0.003) & 0.069 \,(0.007) & 0.061 \,(0.006) & 0.069 \,(0.006) & 0.058 \,(0.005) & 0.067 \,(0.007) \\ 
& DD & \textbf{0.014} \,(0.003) & 0.015 \,(0.003) & 0.070 \,(0.007) & 0.061 \,(0.006) & 0.070 \,(0.006) & 0.059 \,(0.006) & 0.068 \,(0.007) \\ 
\hline 
\multirow{2}{*}{\ref{ModelC2}} 
& max & \textbf{0.005} \,(0.002) & 0.021 \,(0.005) & 0.196 \,(0.035) & 0.312 \,(0.057) & 0.456 \,(0.044) & 0.401 \,(0.051) & 0.122 \,(0.043) \\ 
& DD & \textbf{0.005} \,(0.002) & 0.020 \,(0.005) & 0.203 \,(0.039) & 0.327 \,(0.080) & 0.433 \,(0.037) & 0.381 \,(0.045) & 0.102 \,(0.034) \\ 
\hline 
\multirow{2}{*}{\ref{ModelC3}} 
& max & \textbf{0.077} \,(0.012) & 0.087 \,(0.012) & 0.242 \,(0.038) & 0.423 \,(0.057) & 0.496 \,(0.027) & 0.441 \,(0.040) & 0.211 \,(0.040) \\ 
& DD & \textbf{0.028} \,(0.005) & 0.060 \,(0.008) & 0.245 \,(0.044) & 0.427 \,(0.059) & 0.483 \,(0.024) & 0.428 \,(0.036) & 0.153 \,(0.017) \\ 
\bottomrule
\end{tabular}}
\caption{Supervised classification: Mean empirical misclassification rates, with standard deviations in parentheses, for the max-depth and linear DD-classifiers under Models~\ref{ModelC1}--\ref{ModelC3}. Evaluated on a test set of size $n_{\textrm{test}} = 500$. The lowest misclassification rates are highlighted in bold.}
\label{tab: Classification_n500}
\end{table}

%% file: Tables_applications/KW_test.tex
\begin{table}
\centering
\resizebox{0.6\textwidth}{!}{%
\begin{tabular}{ccccccc}
\toprule
$\RPD_{0.001}$ & $\RPD_{0.1}$ & $\RHD_{0.001}$ & $\RHD_{0.1}$ & \FD{} & \MBD{} & \ID{} \\
\midrule
 0.042 & 0.037 & 0.046 & 0.046 & 0.040 & 0.040 & 0.039 \\
\bottomrule
\end{tabular}}
\caption{Hypothesis testing: Empirical sizes of the KW tests using depth-based rankings with $n=100$ (all Models~\ref{ModelT1}--\ref{ModelT4} correspond to the same null hypothesis). All methods successfully achieve the nominal significance level $0.05$.
\label{tab: KW_test}}
\end{table}

%% file: Tables_applications/Location_n500.tex
\begin{table}
\centering
\resizebox{\textwidth}{!}{%
\begin{tabular}{cccccccccc}
\toprule
Metric & $\varepsilon$ & Model & $\RPD_{0.1}$ & $\RPD_{0.5}$ & $\RHD_{0.1}$ & $\RHD_{0.5}$ & \FD{} & \MBD{} & \ID{} \\
\midrule

\multirow{15}{*}{\rotatebox{90}{\textbf{$\mathrm{MISE}\cdot 10^3$}}} 
  & \multirow{5}{*}{0.01} 
    & \ref{ModelL1} & 7.645 \,(6.159) & 5.136 \,(2.686) & 5.189 \,(2.355) & 5.581 \,(2.308) & 4.485 \,(1.641) & \textbf{4.413} \,(1.540) & 5.081 \,(2.204) \\
  & & \ref{ModelL2} & 7.308 \,(5.604) & 5.012 \,(2.670) & 5.154 \,(2.461) & 5.617 \,(2.451) & 4.418 \,(1.691) & \textbf{4.314} \,(1.564) & 4.973 \,(2.196) \\
  & & \ref{ModelL3} & 7.383 \,(5.866) & 5.204 \,(3.116) & 5.331 \,(2.492) & 5.428 \,(2.234) & 4.486 \,(1.738) & \textbf{4.400} \,(1.654) & 4.862 \,(1.990) \\
  & & \ref{ModelL4} & 6.976 \,(5.603) & 5.020 \,(2.784) & 5.167 \,(2.513) & 5.508 \,(2.452) & 4.438 \,(1.755) & \textbf{4.346} \,(1.671) & 4.879 \,(2.078) \\
  & & \ref{ModelL5} & 7.260 \,(5.344) & 4.999 \,(2.674) & 5.139 \,(2.480) & 5.546 \,(2.281) & 4.430 \,(1.696) & \textbf{4.322} \,(1.599) & 4.930 \,(2.096) \\
  \cmidrule{2-10}
  & \multirow{5}{*}{0.05} 
    & \ref{ModelL1} & 7.796 \,(6.118) & 5.539 \,(3.241) & 5.606 \,(2.709) & 6.094 \,(2.594) & 4.798 \,(2.058) & \textbf{4.686} \,(1.949) & 5.458 \,(2.761) \\
  & & \ref{ModelL2} & 7.267 \,(5.702) & 5.314 \,(3.125) & 5.256 \,(2.336) & 5.715 \,(2.489) & 4.617 \,(1.842) & \textbf{4.488} \,(1.742) & 5.274 \,(2.407) \\
  & & \ref{ModelL3} & 7.491 \,(5.746) & 5.133 \,(2.736) & 5.469 \,(2.727) & 5.677 \,(2.524) & 4.576 \,(1.805) & \textbf{4.507} \,(1.672) & 5.129 \,(2.156) \\
  & & \ref{ModelL4} & 7.370 \,(5.442) & 5.333 \,(3.252) & 5.339 \,(2.674) & 5.790 \,(2.860) & 4.545 \,(1.743) & \textbf{4.447} \,(1.680) & 5.023 \,(2.147) \\
  & & \ref{ModelL5} & 7.404 \,(6.130) & 5.246 \,(2.836) & 5.913 \,(3.071) & 5.857 \,(2.530) & 4.707 \,(4.856) & \textbf{4.416} \,(1.532) & 5.037 \,(2.049) \\
  \cmidrule{2-10}
  & \multirow{5}{*}{0.10} 
    & \ref{ModelL1} & 9.382 \,(8.015) & 6.461 \,(4.345) & 6.802 \,(3.843) & 7.039 \,(3.079) & 5.643 \,(2.679) & \textbf{5.615} \,(2.581) & 6.679 \,(3.665) \\
  & & \ref{ModelL2} & 7.245 \,(5.129) & 5.522 \,(3.036) & 5.661 \,(2.567) & 6.055 \,(2.659) & 5.184 \,(2.191) & \textbf{5.056} \,(2.070) & 5.654 \,(2.588) \\
  & & \ref{ModelL3} & 7.293 \,(5.170) & 5.375 \,(2.930) & 5.658 \,(2.787) & 5.816 \,(2.607) & 5.203 \,(2.079) & \textbf{5.161} \,(2.116) & 6.025 \,(2.773) \\
  & & \ref{ModelL4} & 7.238 \,(5.755) & 5.309 \,(3.004) & 5.457 \,(2.708) & 5.866 \,(2.707) & 4.640 \,(1.818) & \textbf{4.525} \,(1.643) & 5.150 \,(2.192) \\
  & & \ref{ModelL5} & 7.491 \,(6.107) & 5.487 \,(3.057) & 6.085 \,(3.369) & 5.957 \,(2.613) & 4.763 \,(4.707) & \textbf{4.521} \,(1.753) & 5.236 \,(2.400) \\
\midrule

\multirow{15}{*}{\rotatebox{90}{\textbf{$\mathrm{IVAR}\cdot 10^3$}}} 
  & \multirow{5}{*}{0.01} 
    & \ref{ModelL1} & 3.225 \,(1.269) & \textbf{3.196} \,(1.022) & 3.694 \,(1.368) & 4.868 \,(2.118) & 3.582 \,(1.182) & 3.508 \,(1.083) & 3.863 \,(1.482) \\
  & & \ref{ModelL2} & 3.284 \,(1.352) & \textbf{3.171} \,(1.020) & 3.618 \,(1.281) & 4.718 \,(2.161) & 3.535 \,(1.196) & 3.461 \,(1.110) & 3.818 \,(1.423) \\
  & & \ref{ModelL3} & 3.169 \,(1.284) & \textbf{3.105} \,(0.975) & 3.608 \,(1.296) & 4.525 \,(1.888) & 3.565 \,(1.217) & 3.473 \,(1.128) & 3.816 \,(1.464) \\
  & & \ref{ModelL4} & 3.148 \,(1.213) & \textbf{3.100} \,(0.996) & 3.593 \,(1.437) & 4.790 \,(2.243) & 3.516 \,(1.207) & 3.428 \,(1.129) & 3.832 \,(1.492) \\
  & & \ref{ModelL5} & 3.227 \,(1.305) & \textbf{3.138} \,(0.980) & 3.631 \,(1.328) & 4.796 \,(2.103) & 3.559 \,(1.197) & 3.440 \,(1.092) & 3.836 \,(1.399) \\
  \cmidrule{2-10}
  & \multirow{5}{*}{0.05} 
    & \ref{ModelL1} & \textbf{3.142} \,(1.220) & 3.149 \,(1.024) & 3.851 \,(1.409) & 5.227 \,(2.337) & 3.566 \,(1.218) & 3.487 \,(1.138) & 3.800 \,(1.501) \\
  & & \ref{ModelL2} & 3.326 \,(1.441) & \textbf{3.231} \,(1.083) & 3.815 \,(1.383) & 4.709 \,(2.048) & 3.708 \,(1.360) & 3.647 \,(1.290) & 4.057 \,(1.691) \\
  & & \ref{ModelL3} & 3.201 \,(1.224) & \textbf{3.182} \,(1.050) & 3.672 \,(1.260) & 4.376 \,(1.890) & 3.684 \,(1.367) & 3.596 \,(1.230) & 4.035 \,(1.653) \\
  & & \ref{ModelL4} & 3.209 \,(1.330) & \textbf{3.097} \,(0.973) & 3.607 \,(1.262) & 4.961 \,(2.651) & 3.605 \,(1.247) & 3.496 \,(1.110) & 3.800 \,(1.454) \\
  & & \ref{ModelL5} & 3.300 \,(1.331) & \textbf{3.224} \,(1.070) & 3.765 \,(1.429) & 4.946 \,(2.292) & 3.752 \,(4.551) & 3.531 \,(1.136) & 3.908 \,(1.478) \\
  \cmidrule{2-10}
  & \multirow{5}{*}{0.10} 
    & \ref{ModelL1} & \textbf{3.113} \,(1.164) & 3.177 \,(1.025) & 4.025 \,(1.621) & 5.662 \,(2.753) & 3.663 \,(1.313) & 3.577 \,(1.197) & 3.956 \,(1.562) \\
  & & \ref{ModelL2} & 3.511 \,(1.658) & \textbf{3.430} \,(1.276) & 4.146 \,(1.619) & 4.967 \,(2.223) & 4.259 \,(1.809) & 4.183 \,(1.760) & 4.510 \,(2.108) \\
  & & \ref{ModelL3} & 3.430 \,(1.432) & \textbf{3.410} \,(1.193) & 3.966 \,(1.479) & 4.508 \,(1.841) & 4.295 \,(1.687) & 4.243 \,(1.716) & 4.735 \,(2.195) \\
  & & \ref{ModelL4} & 3.260 \,(1.235) & \textbf{3.196} \,(1.019) & 3.731 \,(1.354) & 5.128 \,(2.581) & 3.676 \,(1.353) & 3.580 \,(1.202) & 3.971 \,(1.565) \\
  & & \ref{ModelL5} & \textbf{3.273} \,(1.328) & 3.305 \,(1.149) & 3.837 \,(1.460) & 4.866 \,(2.170) & 3.776 \,(4.340) & 3.574 \,(1.228) & 3.999 \,(1.645) \\
\midrule

\multirow{15}{*}{\rotatebox{90}{\textbf{$\mathrm{ISB}\cdot 10^3$}}} 
  & \multirow{5}{*}{0.01} 
    & \ref{ModelL1} & 4.420 \,(6.194) & 1.940 \,(2.591) & 1.495 \,(2.026) & \textbf{0.713} \,(1.032) & 0.903 \,(1.253) & 0.904 \,(1.253) & 1.218 \,(1.770) \\
  & & \ref{ModelL2} & 4.024 \,(5.533) & 1.840 \,(2.580) & 1.536 \,(2.213) & 0.899 \,(1.236) & 0.883 \,(1.207) & \textbf{0.853} \,(1.152) & 1.155 \,(1.705) \\
  & & \ref{ModelL3} & 4.213 \,(5.869) & 2.098 \,(3.103) & 1.723 \,(2.224) & \textbf{0.903} \,(1.224) & 0.921 \,(1.295) & 0.927 \,(1.285) & 1.046 \,(1.523) \\
  & & \ref{ModelL4} & 3.828 \,(5.505) & 1.920 \,(2.669) & 1.573 \,(2.146) & \textbf{0.717} \,(1.045) & 0.921 \,(1.264) & 0.918 \,(1.302) & 1.047 \,(1.548) \\
  & & \ref{ModelL5} & 4.032 \,(5.328) & 1.861 \,(2.585) & 1.508 \,(2.178) & \textbf{0.750} \,(0.992) & 0.870 \,(1.237) & 0.882 \,(1.213) & 1.093 \,(1.664) \\
  \cmidrule{2-10}
  & \multirow{5}{*}{0.05} 
    & \ref{ModelL1} & 4.654 \,(6.097) & 2.389 \,(3.186) & 1.754 \,(2.377) & \textbf{0.866} \,(1.185) & 1.231 \,(1.711) & 1.199 \,(1.647) & 1.657 \,(2.411) \\
  & & \ref{ModelL2} & 3.941 \,(5.572) & 2.083 \,(3.034) & 1.440 \,(1.940) & 1.006 \,(1.424) & 0.909 \,(1.249) & \textbf{0.841} \,(1.159) & 1.217 \,(1.766) \\
  & & \ref{ModelL3} & 4.290 \,(5.713) & 1.950 \,(2.687) & 1.797 \,(2.435) & 1.300 \,(1.795) & \textbf{0.892} \,(1.227) & 0.911 \,(1.295) & 1.093 \,(1.625) \\
  & & \ref{ModelL4} & 4.160 \,(5.435) & 2.236 \,(3.250) & 1.732 \,(2.370) & \textbf{0.829} \,(1.150) & 0.940 \,(1.321) & 0.951 \,(1.397) & 1.223 \,(1.723) \\
  & & \ref{ModelL5} & 4.104 \,(6.096) & 2.022 \,(2.719) & 2.148 \,(2.822) & 0.910 \,(1.213) & 0.954 \,(1.299) & \textbf{0.885} \,(1.156) & 1.128 \,(1.553) \\
  \cmidrule{2-10}
  & \multirow{5}{*}{0.10} 
    & \ref{ModelL1} & 6.268 \,(7.990) & 3.283 \,(4.275) & 2.776 \,(3.429) & \textbf{1.377} \,(1.594) & 1.980 \,(2.337) & 2.038 \,(2.304) & 2.722 \,(3.253) \\
  & & \ref{ModelL2} & 3.734 \,(5.001) & 2.092 \,(2.838) & 1.515 \,(2.095) & 1.087 \,(1.545) & 0.924 \,(1.298) & \textbf{0.873} \,(1.218) & 1.144 \,(1.584) \\
  & & \ref{ModelL3} & 3.863 \,(5.084) & 1.964 \,(2.786) & 1.692 \,(2.364) & 1.308 \,(1.949) & \textbf{0.907} \,(1.273) & 0.918 \,(1.337) & 1.289 \,(1.895) \\
  & & \ref{ModelL4} & 3.977 \,(5.675) & 2.112 \,(2.922) & 1.726 \,(2.396) & \textbf{0.738} \,(1.016) & 0.964 \,(1.292) & 0.945 \,(1.238) & 1.178 \,(1.630) \\
  & & \ref{ModelL5} & 4.218 \,(6.058) & 2.181 \,(2.931) & 2.247 \,(3.048) & 1.090 \,(1.646) & 0.987 \,(1.325) & \textbf{0.946} \,(1.292) & 1.236 \,(1.774) \\
\bottomrule
\end{tabular}}
\caption{Robust location estimation: Comparison of performance measures for Models~\ref{ModelL1}--\ref{ModelL5} with $n=500$. The table reports $\mathrm{MISE}\cdot 10^3$, $\mathrm{IVAR}\cdot 10^3$, and $\mathrm{ISB}\cdot 10^3$, as defined in~\eqref{eq: MISEiSB}, for various contamination levels $\varepsilon\in\set{0.01,0.05,0.1}$. For each experimental setting, the deepest function is identified and the corresponding quantities are evaluated. Reported values are Monte Carlo averages with standard deviations in parentheses, based on $1\,000$ simulations.}
\label{tab: Location_n500}
\end{table}

%% file: Tables_applications/OutlierDetectionRanks_n100.tex
\begin{table}[!h]
\centering
\resizebox{\textwidth}{!}{%
\begin{tabular}{cccccccc}
\toprule
Model & $\RPD_{0.001}$ & $\RPD_{0.1}$ & $\RHD_{0.001}$ & $\RHD_{0.1}$ & \FD{} & \MBD{} & \ID{} \\
\midrule

\multicolumn{8}{c}{\textbf{Contamination level $\varepsilon = 0.01$ ($m = 1$, $n = 100$)}} \\
\midrule
\ref{ModelD1} & \textbf{0.010} \,(0.000) & \textbf{0.010} \,(0.000) & 0.085 \,(0.015) & 0.061 \,(0.012) & \textbf{0.010} \,(0.000) & \textbf{0.010} \,(0.000) & 0.049 \,(0.013) \\ 
\ref{ModelD2} & \textbf{0.010} \,(0.001) & \textbf{0.010} \,(0.001) & 0.155 \,(0.020) & 0.127 \,(0.020) & 0.017 \,(0.011) & 0.015 \,(0.009) & 0.140 \,(0.022) \\ 
\ref{ModelD3} & \textbf{0.010} \,(0.000) & \textbf{0.010} \,(0.000) & 0.113 \,(0.017) & 0.093 \,(0.038) & 0.168 \,(0.078) & 0.140 \,(0.069) & 0.083 \,(0.019) \\ 
\ref{ModelD4} & \textbf{0.010} \,(0.000) & \textbf{0.010} \,(0.000) & 0.112 \,(0.016) & 0.088 \,(0.015) & 0.080 \,(0.023) & 0.062 \,(0.020) & 0.081 \,(0.017) \\ 
\ref{ModelD5} & \textbf{0.010} \,(0.001) & 0.015 \,(0.010) & 0.960 \,(0.077) & 0.971 \,(0.044) & 0.517 \,(0.057) & 0.522 \,(0.067) & 0.542 \,(0.147) \\ 
\ref{ModelD6} & \textbf{0.010} \,(0.004) & 0.023 \,(0.038) & 0.225 \,(0.251) & 0.715 \,(0.185) & 0.715 \,(0.149) & 0.684 \,(0.143) & 0.364 \,(0.192) \\
\midrule

\multicolumn{8}{c}{\textbf{Contamination level $\varepsilon = 0.05$ ($m = 5$, $n = 100$)}} \\
\midrule
\ref{ModelD1} & \textbf{0.030} \,(0.000) & \textbf{0.030} \,(0.000) & 0.068 \,(0.020) & 0.055 \,(0.017) & 0.031 \,(0.002) & 0.031 \,(0.003) & 0.042 \,(0.021) \\ 
\ref{ModelD2} & \textbf{0.030} \,(0.001) & \textbf{0.030} \,(0.001) & 0.174 \,(0.034) & 0.152 \,(0.033) & 0.039 \,(0.009) & 0.037 \,(0.007) & 0.168 \,(0.036) \\ 
\ref{ModelD3} & \textbf{0.030} \,(0.000) & \textbf{0.030} \,(0.000) & 0.099 \,(0.019) & 0.097 \,(0.027) & 0.181 \,(0.041) & 0.153 \,(0.036) & 0.088 \,(0.016) \\ 
\ref{ModelD4} & \textbf{0.030} \,(0.000) & \textbf{0.030} \,(0.000) & 0.134 \,(0.020) & 0.119 \,(0.020) & 0.103 \,(0.022) & 0.086 \,(0.019) & 0.108 \,(0.023) \\ 
\ref{ModelD5} & \textbf{0.030} \,(0.001) & 0.037 \,(0.008) & 0.336 \,(0.103) & 0.754 \,(0.092) & 0.516 \,(0.048) & 0.521 \,(0.049) & 0.539 \,(0.073) \\ 
\ref{ModelD6} & \textbf{0.030} \,(0.002) & 0.043 \,(0.018) & 0.191 \,(0.116) & 0.682 \,(0.091) & 0.707 \,(0.072) & 0.680 \,(0.070) & 0.384 \,(0.085) \\
\midrule

\multicolumn{8}{c}{\textbf{Contamination level $\varepsilon = 0.10$ ($m = 10$, $n = 100$)}} \\
\midrule
\ref{ModelD1} & \textbf{0.055} \,(0.000) & \textbf{0.055} \,(0.000) & 0.090 \,(0.016) & 0.077 \,(0.015) & 0.056 \,(0.002) & 0.056 \,(0.003) & 0.062 \,(0.013) \\ 
\ref{ModelD2} & \textbf{0.055} \,(0.002) & \textbf{0.055} \,(0.001) & 0.228 \,(0.036) & 0.209 \,(0.033) & 0.068 \,(0.009) & 0.067 \,(0.008) & 0.225 \,(0.039) \\ 
\ref{ModelD3} & \textbf{0.055} \,(0.000) & \textbf{0.055} \,(0.000) & 0.108 \,(0.020) & 0.115 \,(0.025) & 0.195 \,(0.032) & 0.168 \,(0.028) & 0.097 \,(0.013) \\ 
\ref{ModelD4} & \textbf{0.055} \,(0.000) & \textbf{0.055} \,(0.000) & 0.185 \,(0.026) & 0.175 \,(0.025) & 0.131 \,(0.022) & 0.117 \,(0.019) & 0.163 \,(0.028) \\ 
\ref{ModelD5} & \textbf{0.055} \,(0.001) & 0.068 \,(0.011) & 0.208 \,(0.038) & 0.481 \,(0.101) & 0.515 \,(0.046) & 0.520 \,(0.046) & 0.536 \,(0.058) \\ 
\ref{ModelD6} & 0.056 \,(0.003) & 0.072 \,(0.017) & 0.227 \,(0.075) & 0.643 \,(0.096) & 0.700 \,(0.054) & 0.678 \,(0.054) & 0.423 \,(0.060) \\
\bottomrule
\end{tabular}}
\caption{Outlier detection (Section~\ref{sec: OutlierDetection}): Means and standard deviations (in parentheses) of the average rank of the outlying curves for a sample size of $n=100$ across three contamination levels. Lower values indicate better performance. The smallest average ranks achieved are highlighted in bold. In the ideal case, the $m$ curves with the smallest depth values coincide exactly with the true outliers, yielding a minimum possible average outlier rank of $0.010$ for $m=1$, $0.030$ for $m=5$, and $0.055$ for $m=10$.}
\label{tab: OutlierDetection_n100}
\end{table}

%% file: Tables_applications/OutlierDetectionRanks_n1000.tex
\begin{table}
\centering
\resizebox{\textwidth}{!}{%
\begin{tabular}{cccccccc}
\toprule
Model & $\RPD_{0.001}$ & $\RPD_{0.1}$ & $\RHD_{0.001}$ & $\RHD_{0.1}$ & \FD{} & \MBD{} & \ID{} \\
\midrule

\multicolumn{8}{c}{\textbf{Contamination level $\varepsilon = 0.01$ ($m = 10$, $n = 1000$)}} \\
\midrule
\ref{ModelD1} & \textbf{0.006} \,(0.000) & \textbf{0.006} \,(0.000) & 0.014 \,(0.003) & 0.011 \,(0.002) & \textbf{0.006} \,(0.000) & \textbf{0.006} \,(0.000) & \textbf{0.006} \,(0.002) \\
\ref{ModelD2} & \textbf{0.006} \,(0.000) & \textbf{0.006} \,(0.000) & 0.050 \,(0.010) & 0.044 \,(0.009) & 0.013 \,(0.003) & 0.010 \,(0.002) & 0.041 \,(0.009) \\
\ref{ModelD3} & \textbf{0.006} \,(0.000) & \textbf{0.006} \,(0.000) & 0.045 \,(0.027) & 0.081 \,(0.044) & 0.163 \,(0.024) & 0.133 \,(0.021) & 0.014 \,(0.002) \\
\ref{ModelD4} & \textbf{0.006} \,(0.000) & \textbf{0.006} \,(0.000) & 0.028 \,(0.005) & 0.025 \,(0.004) & 0.079 \,(0.007) & 0.060 \,(0.006) & 0.023 \,(0.004) \\
\ref{ModelD5} & \textbf{0.006} \,(0.000) & \textbf{0.006} \,(0.001) & 0.987 \,(0.006) & 0.981 \,(0.010) & 0.510 \,(0.017) & 0.517 \,(0.020) & 0.575 \,(0.042) \\
\ref{ModelD6} & \textbf{0.006} \,(0.002) & 0.013 \,(0.010) & 0.224 \,(0.171) & 0.722 \,(0.057) & 0.717 \,(0.046) & 0.685 \,(0.044) & 0.348 \,(0.058) \\
\midrule

\multicolumn{8}{c}{\textbf{Contamination level $\varepsilon = 0.05$ ($m = 50$, $n = 1000$)}} \\
\midrule
\ref{ModelD1} & \textbf{0.026} \,(0.000) & \textbf{0.026} \,(0.000) & 0.043 \,(0.004) & 0.038 \,(0.003) & \textbf{0.026} \,(0.000) & \textbf{0.026} \,(0.000) & \textbf{0.026} \,(0.001) \\
\ref{ModelD2} & \textbf{0.026} \,(0.000) & \textbf{0.026} \,(0.000) & 0.133 \,(0.012) & 0.122 \,(0.010) & 0.035 \,(0.003) & 0.033 \,(0.002) & 0.125 \,(0.011) \\
\ref{ModelD3} & \textbf{0.026} \,(0.000) & \textbf{0.026} \,(0.000) & 0.062 \,(0.013) & 0.093 \,(0.020) & 0.173 \,(0.013) & 0.145 \,(0.011) & 0.029 \,(0.001) \\
\ref{ModelD4} & \textbf{0.026} \,(0.000) & \textbf{0.026} \,(0.000) & 0.094 \,(0.007) & 0.089 \,(0.007) & 0.101 \,(0.007) & 0.084 \,(0.006) & 0.082 \,(0.006) \\
\ref{ModelD5} & \textbf{0.026} \,(0.000) & 0.027 \,(0.001) & 0.359 \,(0.119) & 0.824 \,(0.035) & 0.510 \,(0.015) & 0.517 \,(0.016) & 0.572 \,(0.022) \\
\ref{ModelD6} & \textbf{0.026} \,(0.001) & 0.034 \,(0.006) & 0.183 \,(0.109) & 0.690 \,(0.030) & 0.710 \,(0.022) & 0.682 \,(0.022) & 0.377 \,(0.027) \\
\midrule

\multicolumn{8}{c}{\textbf{Contamination level $\varepsilon = 0.10$ ($m = 100$, $n = 1000$)}} \\
\midrule
\ref{ModelD1} & \textbf{0.051} \,(0.000) & \textbf{0.051} \,(0.000) & 0.076 \,(0.006) & 0.066 \,(0.007) & \textbf{0.051} \,(0.000) & \textbf{0.051} \,(0.000) & \textbf{0.051} \,(0.001) \\
\ref{ModelD2} & \textbf{0.051} \,(0.000) & \textbf{0.051} \,(0.000) & 0.207 \,(0.012) & 0.193 \,(0.010) & 0.063 \,(0.003) & 0.063 \,(0.003) & 0.203 \,(0.011) \\
\ref{ModelD3} & \textbf{0.051} \,(0.000) & \textbf{0.051} \,(0.000) & 0.085 \,(0.009) & 0.112 \,(0.013) & 0.187 \,(0.010) & 0.159 \,(0.009) & 0.052 \,(0.001) \\
\ref{ModelD4} & \textbf{0.051} \,(0.000) & \textbf{0.051} \,(0.000) & 0.166 \,(0.009) & 0.161 \,(0.009) & 0.129 \,(0.007) & 0.115 \,(0.006) & 0.147 \,(0.007) \\
\ref{ModelD5} & \textbf{0.051} \,(0.000) & 0.054 \,(0.002) & 0.175 \,(0.022) & 0.539 \,(0.046) & 0.509 \,(0.015) & 0.516 \,(0.015) & 0.568 \,(0.017) \\
\ref{ModelD6} & \textbf{0.051} \,(0.001) & 0.061 \,(0.006) & 0.219 \,(0.066) & 0.658 \,(0.031) & 0.701 \,(0.017) & 0.678 \,(0.017) & 0.414 \,(0.019) \\
\bottomrule
\end{tabular}}
\caption{Outlier detection (Section~\ref{sec: OutlierDetection}): Means and standard deviations (in parentheses) of the average rank of the outlying curves for a sample size of $n=1\,000$ across three contamination levels. Lower values indicate better performance. The smallest average ranks achieved are highlighted in bold. In the ideal case, the $m$ curves with the smallest depth values coincide exactly with the true outliers, yielding a minimum possible average outlier rank of $0.006$ for $m=10$, $0.026$ for $m=50$, and $0.051$ for $m=100$.}
\label{tab: OutlierDetection_n1000}
\end{table}

%% file: Tables_applications/OutlierDetection_n100.tex
\begin{table}
    \centering
    \resizebox{0.6\textheight}{!}{%
    \begin{subtable}{\textwidth}
        \centering
        \begin{tabular}{cc rrrr}
            \toprule
            \textbf{Model} & $m$ & \RPDOD{} & Funct. boxplot & Outliergram & MS-plot \\
            \midrule
            \multirow{4}{*}{\ref{ModelD1}} 
                & 0  & \textbf{0.018} (0.133) & 0.033 (0.179) & 0.989 (0.104) & 0.894 (0.308) \\
                & 1  & \textbf{0.008} (0.062) & 0.014 (0.083) & 0.990 (0.100) & 0.621 (0.260) \\
                & 5  & \textbf{0.001} (0.016) & 0.003 (0.023) & 0.987 (0.113) & 0.277 (0.174) \\
                & 10 & \textbf{0.001} (0.012) & \textbf{0.001} (0.010) & 0.985 (0.122) & 0.148 (0.112) \\
            \midrule
            \multirow{4}{*}{\ref{ModelD2}} 
                & 0  & 0.023 (0.150) & \textbf{0.001} (0.032) & 0.882 (0.323) & 0.670 (0.470) \\
                & 1  & 0.019 (0.105) & \textbf{0.002} (0.035) & 0.578 (0.263) & 0.400 (0.306) \\
                & 5  & 0.027 (0.072) & \textbf{0.000} (0.011) & 0.206 (0.150) & 0.152 (0.143) \\
                & 10 & 0.029 (0.056) & \textbf{0.000} (0.006) & 0.080 (0.081) & 0.074 (0.082) \\
            \midrule
            \multirow{4}{*}{\ref{ModelD3}} 
                & 0  & \textbf{0.019} (0.137) & 0.032 (0.176) & 0.987 (0.113) & 0.902 (0.297) \\
                & 1  & \textbf{0.004} (0.051) & 0.024 (0.136) & 0.807 (0.158) & 0.633 (0.248) \\
                & 5  & \textbf{0.000} (0.000) & 0.008 (0.057) & 0.403 (0.170) & 0.287 (0.170) \\
                & 10 & \textbf{0.000} (0.006) & 0.002 (0.019) & 0.177 (0.121) & 0.153 (0.110) \\
            \midrule
            \multirow{4}{*}{\ref{ModelD4}} 
                & 0  & \textbf{0.014} (0.118) & 0.038 (0.191) & 0.986 (0.118) & 0.890 (0.313) \\
                & 1  & \textbf{0.000} (0.000) & 0.022 (0.117) & 0.770 (0.142) & 0.619 (0.261) \\
                & 5  & \textbf{0.000} (0.000) & 0.005 (0.038) & 0.348 (0.144) & 0.283 (0.170) \\
                & 10 & \textbf{0.000} (0.000) & 0.002 (0.019) & 0.131 (0.092) & 0.146 (0.109) \\
            \midrule
            \multirow{4}{*}{\ref{ModelD5}} 
                & 0  & \textbf{0.009} (0.094) & 0.038 (0.191) & 0.982 (0.133) & 0.884 (0.320) \\
                & 1  & \textbf{0.016} (0.097) & 0.044 (0.205) & 0.918 (0.175) & 0.876 (0.330) \\
                & 5  & 0.035 (0.084) & \textbf{0.034} (0.181) & 0.834 (0.291) & 0.801 (0.399) \\
                & 10 & 0.051 (0.082) & \textbf{0.031} (0.173) & 0.761 (0.407) & 0.679 (0.467) \\
            \midrule
            \multirow{4}{*}{\ref{ModelD6}} 
                & 0  & \textbf{0.088} (0.283) & 0.590 (0.492) & 0.985 (0.122) & 0.590 (0.492) \\
                & 1  & \textbf{0.005} (0.064) & 0.569 (0.495) & 0.767 (0.158) & 0.369 (0.324) \\
                & 5  & \textbf{0.003} (0.023) & 0.518 (0.499) & 0.345 (0.160) & 0.175 (0.174) \\
                & 10 & \textbf{0.010} (0.036) & 0.517 (0.499) & 0.123 (0.098) & 0.133 (0.140) \\
            \bottomrule
        \end{tabular}
        \caption{False Discovery Rate (FDR) --- Lower is better}
    \end{subtable}%
    }
    \resizebox{0.6\textheight}{!}{%
    \begin{subtable}{\textwidth}
        \centering
        \begin{tabular}{cc rrrr}
            \toprule
            \textbf{Model} & $m$ & \RPDOD{} & Funct. boxplot & Outliergram & MS-plot \\
            \midrule
            \multirow{3}{*}{\ref{ModelD1}} 
                & 1  & 0.639 (0.477) & \textbf{0.992} (0.050) & -0.021 (0.006) & 0.577 (0.197) \\
                & 5  & \textbf{0.999} (0.009) & 0.998 (0.013) & -0.048 (0.014) & 0.831 (0.112) \\
                & 10 & \textbf{0.999} (0.008) & \textbf{0.999} (0.006) & -0.068 (0.020) & 0.907 (0.070) \\
            \midrule
            \multirow{3}{*}{\ref{ModelD2}} 
                & 1  & 0.378 (0.478) & 0.194 (0.395) & 0.618 (0.188) & \textbf{0.745} (0.203) \\
                & 5  & \textbf{0.944} (0.172) & 0.163 (0.237) & 0.878 (0.092) & 0.912 (0.086) \\
                & 10 & 0.951 (0.132) & 0.101 (0.167) & 0.927 (0.059) & \textbf{0.956} (0.049) \\
            \midrule
            \multirow{3}{*}{\ref{ModelD3}} 
                & 1  & 0.503 (0.499) & \textbf{0.582} (0.489) & 0.376 (0.205) & 0.571 (0.185) \\
                & 5  & \textbf{1.000} (0.000) & 0.715 (0.199) & 0.672 (0.145) & 0.828 (0.108) \\
                & 10 & \textbf{1.000} (0.003) & 0.683 (0.140) & 0.775 (0.109) & 0.908 (0.069) \\
            \midrule
            \multirow{3}{*}{\ref{ModelD4}} 
                & 1  & 0.044 (0.205) & 0.534 (0.492) & 0.455 (0.124) & \textbf{0.581} (0.194) \\
                & 5  & \textbf{1.000} (0.000) & 0.649 (0.250) & 0.790 (0.095) & 0.831 (0.108) \\
                & 10 & \textbf{1.000} (0.000) & 0.612 (0.195) & 0.920 (0.057) & 0.913 (0.068) \\
            \midrule
            \multirow{3}{*}{\ref{ModelD5}} 
                & 1  & \textbf{0.205} (0.396) & 0.000 (0.002) & 0.118 (0.212) & -0.014 (0.013) \\
                & 5  & \textbf{0.903} (0.261) & -0.001 (0.004) & 0.078 (0.215) & -0.029 (0.020) \\
                & 10 & \textbf{0.932} (0.187) & -0.001 (0.006) & -0.008 (0.112) & -0.034 (0.027) \\
            \midrule
            \multirow{3}{*}{\ref{ModelD6}} 
                & 1  & 0.030 (0.168) & -0.007 (0.033) & 0.453 (0.142) & \textbf{0.751} (0.234) \\
                & 5  & 0.761 (0.405) & -0.013 (0.034) & 0.787 (0.106) & \textbf{0.871} (0.120) \\
                & 10 & 0.878 (0.285) & -0.021 (0.028) & \textbf{0.895} (0.074) & 0.792 (0.161) \\
            \bottomrule
        \end{tabular}
        \caption{Matthews Correlation Coefficient (MCC) --- Higher is better}
        \vspace{1em}
    \end{subtable}
    }
    \caption{Outlier detection (Section~\ref{sec: OutlierDetection}): Means and standard deviations (in parentheses) of~$\mathrm{MCC}$ and $\mathrm{FDR}$ of outlier detection methods described in Section~\ref{subsec: threshold}. Sample size $n=100$ across four contamination levels $m\in\{0,1,5,10\}$ in models~\ref{ModelD1}--\ref{ModelD6} is considered. Best results are highlighted in bold.}
    \label{tab: outlierDetectionResults_n100}
\end{table}

%% file: Tables_applications/OutlierDetection_n1000.tex
\begin{table}
    \centering
    \resizebox{0.6\textheight}{!}{%
    \begin{subtable}{\textwidth}
        \centering
        \begin{tabular}{cc rrrr}
            \toprule
            \textbf{Model} & $m$ & \RPDOD{} & Funct. boxplot & Outliergram & MS-plot \\
            \midrule
            \multirow{4}{*}{\ref{ModelD1}} 
                & 0   & \textbf{0.000} (0.000) & \textbf{0.000} (0.000) & 1.000 (0.000) & 1.000 (0.000) \\
                & 10  & 0.000 (0.004) & \textbf{0.000} (0.000) & 1.000 (0.000) & 0.892 (0.016) \\
                & 50  & \textbf{0.000} (0.000) & \textbf{0.000} (0.000) & 1.000 (0.000) & 0.597 (0.039) \\
                & 100 & \textbf{0.000} (0.000) & \textbf{0.000} (0.000) & 1.000 (0.000) & 0.389 (0.039) \\
            \midrule
            \multirow{4}{*}{\ref{ModelD2}} 
                & 0   & \textbf{0.000} (0.000) & \textbf{0.000} (0.000) & 1.000 (0.000) & 1.000 (0.000) \\
                & 10  & 0.005 (0.022) & \textbf{0.000} (0.000) & 0.659 (0.061) & 0.833 (0.031) \\
                & 50  & 0.006 (0.013) & \textbf{0.000} (0.000) & 0.211 (0.050) & 0.475 (0.054) \\
                & 100 & 0.010 (0.011) & \textbf{0.000} (0.000) & 0.077 (0.025) & 0.283 (0.044) \\
            \midrule
            \multirow{4}{*}{\ref{ModelD3}} 
                & 0   & \textbf{0.000} (0.000) & 0.001 (0.032) & 1.000 (0.000) & 1.000 (0.000) \\
                & 10  & \textbf{0.000} (0.000) & 0.000 (0.008) & 0.834 (0.031) & 0.894 (0.015) \\
                & 50  & \textbf{0.000} (0.000) & \textbf{0.000} (0.000) & 0.422 (0.050) & 0.602 (0.038) \\
                & 100 & \textbf{0.000} (0.000) & \textbf{0.000} (0.000) & 0.181 (0.037) & 0.395 (0.039) \\
            \midrule
            \multirow{4}{*}{\ref{ModelD4}} 
                & 0   & \textbf{0.000} (0.000) & 0.002 (0.045) & 1.000 (0.000) & 1.000 (0.000) \\
                & 10  & \textbf{0.000} (0.000) & 0.002 (0.035) & 0.804 (0.027) & 0.891 (0.016) \\
                & 50  & \textbf{0.000} (0.000) & 0.000 (0.016) & 0.365 (0.048) & 0.595 (0.040) \\
                & 100 & \textbf{0.000} (0.000) & 0.000 (0.004) & 0.134 (0.031) & 0.389 (0.041) \\
            \midrule
            \multirow{4}{*}{\ref{ModelD5}} 
                & 0   & \textbf{0.000} (0.000) & 0.001 (0.032) & 1.000 (0.000) & 1.000 (0.000) \\
                & 10  & 0.004 (0.018) & \textbf{0.001} (0.032) & 0.963 (0.041) & 0.995 (0.012) \\
                & 50  & 0.003 (0.008) & \textbf{0.001} (0.032) & 0.991 (0.034) & 0.990 (0.027) \\
                & 100 & 0.007 (0.010) & \textbf{0.000} (0.000) & 1.000 (0.000) & 0.993 (0.024) \\
            \midrule
            \multirow{4}{*}{\ref{ModelD6}} 
                & 0   & \textbf{0.000} (0.000) & 0.950 (0.218) & 1.000 (0.000) & 1.000 (0.000) \\
                & 10  & \textbf{0.000} (0.000) & 0.941 (0.235) & 0.790 (0.034) & 0.747 (0.052) \\
                & 50  & \textbf{0.000} (0.003) & 0.924 (0.265) & 0.349 (0.051) & 0.350 (0.064) \\
                & 100 & \textbf{0.002} (0.005) & 0.905 (0.292) & 0.115 (0.032) & 0.196 (0.054) \\
            \bottomrule
        \end{tabular}
        \caption{False Discovery Rate (FDR) --- Lower is better}
    \end{subtable}%
    }
    \resizebox{0.6\textheight}{!}{%
    \begin{subtable}{\textwidth}
        \centering
        \begin{tabular}{cc rrrr}
            \toprule
            \textbf{Model} & $m$ & \RPDOD{} & Funct. boxplot & Outliergram & MS-plot \\
            \midrule
            \multirow{3}{*}{\ref{ModelD1}} 
                & 10  & 0.976 (0.153) & \textbf{1.000} (0.002) & -0.022 (0.002) & 0.314 (0.025) \\
                & 50  & \textbf{1.000} (0.001) & \textbf{1.000} (0.001) & -0.049 (0.004) & 0.609 (0.033) \\
                & 100 & \textbf{1.000} (0.001) & \textbf{1.000} (0.001) & -0.070 (0.006) & 0.753 (0.029) \\
            \midrule
            \multirow{3}{*}{\ref{ModelD2}} 
                & 10  & \textbf{0.875} (0.305) & 0.007 (0.045) & 0.576 (0.052) & 0.396 (0.039) \\
                & 50  & \textbf{0.986} (0.012) & 0.000 (0.006) & 0.880 (0.030) & 0.706 (0.040) \\
                & 100 & \textbf{0.981} (0.011) & 0.000 (0.000) & 0.937 (0.017) & 0.828 (0.030) \\
            \midrule
            \multirow{3}{*}{\ref{ModelD3}} 
                & 10  & \textbf{0.996} (0.063) & 0.522 (0.155) & 0.361 (0.056) & 0.311 (0.024) \\
                & 50  & \textbf{1.000} (0.000) & 0.502 (0.065) & 0.667 (0.045) & 0.604 (0.033) \\
                & 100 & \textbf{1.000} (0.000) & 0.459 (0.049) & 0.778 (0.034) & 0.749 (0.029) \\
            \midrule
            \multirow{3}{*}{\ref{ModelD4}} 
                & 10  & \textbf{1.000} (0.000) & 0.177 (0.187) & 0.432 (0.031) & 0.315 (0.025) \\
                & 50  & \textbf{1.000} (0.000) & 0.186 (0.092) & 0.784 (0.032) & 0.610 (0.034) \\
                & 100 & \textbf{1.000} (0.000) & 0.156 (0.063) & 0.921 (0.019) & 0.753 (0.030) \\
            \midrule
            \multirow{3}{*}{\ref{ModelD5}} 
                & 10  & \textbf{0.898} (0.298) & 0.000 (0.000) & 0.061 (0.092) & -0.015 (0.041) \\
                & 50  & \textbf{0.998} (0.005) & 0.000 (0.000) & -0.031 (0.032) & -0.051 (0.041) \\
                & 100 & \textbf{0.996} (0.006) & 0.000 (0.000) & -0.042 (0.006) & -0.079 (0.023) \\
            \midrule
            \multirow{3}{*}{\ref{ModelD6}} 
                & 10  & 0.005 (0.069) & -0.005 (0.010) & 0.448 (0.037) & \textbf{0.492} (0.052) \\
                & 50  & \textbf{0.939} (0.197) & -0.011 (0.005) & 0.794 (0.034) & 0.792 (0.042) \\
                & 100 & \textbf{0.982} (0.011) & -0.015 (0.007) & 0.924 (0.021) & 0.866 (0.034) \\
            \bottomrule
        \end{tabular}
        \caption{Matthews Correlation Coefficient (MCC) --- Higher is better}
        \vspace{1em}
    \end{subtable}
    }
    \caption{Outlier detection (Section~\ref{sec: OutlierDetection}): Means and standard deviations (in parentheses) of~$\mathrm{MCC}$ and $\mathrm{FDR}$ of outlier detection methods described in Section~\ref{subsec: threshold}. Sample size $n=1000$ across four contamination levels $m\in\{0,10,50,100\}$ in models~\ref{ModelD1}--\ref{ModelD6} is considered. Best results are highlighted in bold.}
    \label{tab: outlierDetectionResults_n1000}
\end{table}

%% file: Tables_applications/Classification_n100.tex
\begin{table}
\centering
\resizebox{\textwidth}{!}{%
\begin{tabular}{c c c c c c c c c}
\toprule
Model & Classifier & $\RPD_{0.001}$ & $\RPD_{0.1}$ & $\RHD_{0.001}$ & $\RHD_{0.1}$ & \FD{} & \MBD{} & \ID{} \\
\midrule
\multirow{2}{*}{\ref{ModelC1}} 
& max & \textbf{0.018} \,(0.004) & \textbf{0.018} \,(0.004) & 0.167 \,(0.017) & 0.141 \,(0.017) & 0.071 \,(0.006) & 0.060 \,(0.006) & 0.168 \,(0.018) \\ 
& DD & \textbf{0.020} \,(0.005) & \textbf{0.020} \,(0.005) & 0.167 \,(0.017) & 0.142 \,(0.016) & 0.074 \,(0.009) & 0.062 \,(0.008) & 0.168 \,(0.018) \\ 
\hline 
\multirow{2}{*}{\ref{ModelC2}} 
& max & \textbf{0.007} \,(0.002) & 0.028 \,(0.010) & 0.226 \,(0.039) & 0.357 \,(0.062) & 0.479 \,(0.041) & 0.426 \,(0.049) & 0.197 \,(0.063) \\ 
& DD & \textbf{0.008} \,(0.003) & 0.026 \,(0.008) & 0.231 \,(0.043) & 0.350 \,(0.070) & 0.463 \,(0.042) & 0.421 \,(0.046) & 0.155 \,(0.046) \\ 
\hline 
\multirow{2}{*}{\ref{ModelC3}} 
& max & \textbf{0.085} \,(0.018) & 0.090 \,(0.018) & 0.277 \,(0.042) & 0.458 \,(0.050) & 0.496 \,(0.023) & 0.467 \,(0.033) & 0.264 \,(0.065) \\ 
& DD & \textbf{0.038} \,(0.007) & 0.071 \,(0.012) & 0.277 \,(0.046) & 0.457 \,(0.056) & 0.492 \,(0.025) & 0.461 \,(0.031) & 0.196 \,(0.032) \\ 
\bottomrule
\end{tabular}}
\caption{Supervised classification (Section~\ref{sec: Classification}): Mean empirical misclassification rates, with standard deviations in parentheses, for the max-depth and linear DD-classifiers under Models~\ref{ModelC1}--\ref{ModelC3}. Evaluated on a test set of size $n_{\textrm{test}} = 100$. The lowest misclassification rates are highlighted in bold.}
\label{tab: Classification_n100}
\end{table}

%% file: Tables_applications/Classification_n1000.tex
\begin{table}
\centering
\resizebox{\textwidth}{!}{%
\begin{tabular}{c c c c c c c c c}
\toprule
Model & Classifier & $\RPD_{0.001}$ & $\RPD_{0.1}$ & $\RHD_{0.001}$ & $\RHD_{0.1}$ & \FD{} & \MBD{} & \ID{} \\
\midrule
\multirow{2}{*}{\ref{ModelC1}} 
& max & \textbf{0.013} \,(0.002) & \textbf{0.013} \,(0.003) & 0.049 \,(0.005) & 0.045 \,(0.005) & 0.069 \,(0.006) & 0.058 \,(0.005) & 0.051 \,(0.005) \\ 
& DD & \textbf{0.013} \,(0.003) & 0.014 \,(0.003) & 0.050 \,(0.005) & 0.046 \,(0.005) & 0.070 \,(0.006) & 0.059 \,(0.005) & 0.051 \,(0.005) \\ 
\hline 
\multirow{2}{*}{\ref{ModelC2}} 
& max & \textbf{0.004} \,(0.002) & 0.020 \,(0.004) & 0.194 \,(0.035) & 0.304 \,(0.054) & 0.446 \,(0.040) & 0.398 \,(0.048) & 0.104 \,(0.033) \\ 
& DD & \textbf{0.005} \,(0.002) & 0.020 \,(0.004) & 0.201 \,(0.040) & 0.327 \,(0.080) & 0.425 \,(0.036) & 0.377 \,(0.044) & 0.089 \,(0.027) \\ 
\hline 
\multirow{2}{*}{\ref{ModelC3}} 
& max & \textbf{0.076} \,(0.011) & 0.086 \,(0.011) & 0.236 \,(0.038) & 0.419 \,(0.057) & 0.495 \,(0.031) & 0.427 \,(0.042) & 0.198 \,(0.030) \\ 
& DD & \textbf{0.027} \,(0.004) & 0.058 \,(0.007) & 0.236 \,(0.043) & 0.425 \,(0.059) & 0.480 \,(0.025) & 0.414 \,(0.038) & 0.146 \,(0.014) \\ 
\bottomrule
\end{tabular}}
\caption{Supervised classification (Section~\ref{sec: Classification}): Mean empirical misclassification rates, with standard deviations in parentheses, for the max-depth and linear DD-classifiers under Models~\ref{ModelC1}--\ref{ModelC3}. Evaluated on a test set of size $n_{\textrm{test}} = 1000$. The lowest misclassification rates are highlighted in bold.}
\label{tab: Classification_n1000}
\end{table}

%% file: Tables_applications/KW_test_full.tex
\begin{table}
\centering
\resizebox{0.8\textwidth}{!}{%
\begin{tabular}{cccccccc}
\toprule
& $\RPD_{0.001}$ & $\RPD_{0.1}$ & $\RHD_{0.001}$ & $\RHD_{0.1}$ & \FD{} & \MBD{} & \ID{} \\
\midrule
$n=50$   & 0.046 & 0.045 & 0.036 & 0.046 & 0.040 & 0.036 & 0.042 \\
$n=100$  & 0.042 & 0.037 & 0.046 & 0.046 & 0.040 & 0.040 & 0.039 \\
$n=500$  & 0.044 & 0.045 & 0.053 & 0.049 & 0.050 & 0.049 & 0.052 \\
$n=1000$ & 0.048 & 0.052 & 0.063 & 0.068 & 0.060 & 0.063 & 0.058 \\
\bottomrule
\end{tabular}}
\caption{Hypothesis testing (Section~\ref{sec: HypothesisTesting}): Empirical sizes of the Kruskal-Wallis tests using depth-based rankings for all sample sizes $n \in \{50,100,500,1000\}$. The empirical sizes are generally close to the nominal significance level $0.05$, with only minor deviations observed for larger sample sizes.
\label{tab: KW_test_full}}
\end{table}

%% file: Tables_applications/Real_World_Application.tex
\begin{table}
    \centering
    \begin{tabular}{c ccccccc}
        \toprule
        \multirow{2}{*}{\textbf{Hypothesis}} & \multicolumn{7}{c}{\textbf{Depth}} \\
        \cmidrule(lr){2-8}
        & $\RPD{}_{0.001}$ & $\RPD{}_{0.1}$ & $\RHD{}_{0.001}$ & $\RHD{}_{0.1}$ & \FD{} & \MBD{} & \ID{} \\
        \midrule
        $H_0^{(1)}$ & $\mathbf{<0.001}$ & $\mathbf{<0.001}$ & $\mathbf{<0.001}$ & $\mathbf{<0.001}$ & $\mathbf{<0.001}$ & $\mathbf{<0.001}$ & $\mathbf{<0.001}$ \\
        $H_0^{(2)}$ & $\mathbf{<0.001}$ & $\mathbf{<0.001}$ & $\mathbf{<0.001}$ & $\mathbf{<0.001}$ & $\mathbf{<0.001}$ & $\mathbf{<0.001}$ & $\mathbf{<0.001}$ \\
        $H_0^{(3)}$ & \textbf{0.001} & \textbf{0.014} & 1.000 & 0.714 & 0.406 & 0.248 & $\mathbf{0.002}$ \\
        $H_0^{(4)}$ & $\mathbf{<0.001}$ & $\mathbf{<0.001}$ & $\mathbf{<0.001}$ & $\mathbf{<0.001}$ & 0.511 & 0.194 & $\mathbf{<0.001}$ \\
        $H_0^{(5)}$ & 0.403 & 0.745 & 0.699 & 1.000 & 1.000 & 1.000 & 0.223 \\
        \bottomrule
    \end{tabular}
    \caption{Real world application (Section~\ref{sec: RealData}): Results of the multisample Kruskal-Wallis test (described in Section~\ref{sec: HypothesisTesting}) for the null hypotheses $H_0^{(1)}, \dots, H_0^{(5)}$ defined in Section~\ref{sec: RealData}. The table reports p-values, with statistically significant results ($p < 0.05$) highlighted in bold. The results obtained using \RPD{} are broadly consistent with those produced by the competing depth notions. All considered depths strongly reject $H_0^{(1)}$ and $H_0^{(2)}$, indicating clear differences among days. The main discrepancies arise for $H_0^{(3)}$ and $H_0^{(4)}$. For $H_0^{(3)}$, only \RPD{} and \ID{} detect significant differences, whereas \RHD{}, \FD{}, and \MBD{} do not. For $H_0^{(4)}$, both versions of \RPD{}, both versions of \RHD{}, and \ID{} reject the null hypothesis, while \FD{} and \MBD{} fail to do so. Finally, none of the considered depths provide evidence against $H_0^{(5)}$.
    }
    \label{tab: pvalues_realData}
\end{table}

%% file: Tables_applications/Location_n100.tex
\begin{table}
\centering
\resizebox{\textwidth}{!}{%
\begin{tabular}{cccccccccc}
\toprule
Metric & $\varepsilon$ & Model & $\RPD_{0.1}$ & $\RPD_{0.5}$ & $\RHD_{0.1}$ & $\RHD_{0.5}$ & \FD{} & \MBD{} & \ID{} \\
\midrule

\multirow{15}{*}{\rotatebox{90}{\textbf{$\mathrm{MISE}\cdot 10^3$}}} 
  & \multirow{5}{*}{0.01} 
    & \ref{ModelL1} & 7.645 \,(6.159) & 5.136 \,(2.686) & 5.189 \,(2.355) & 5.581 \,(2.308) & 4.485 \,(1.641) & \textbf{4.413} \,(1.540) & 5.081 \,(2.204) \\
  & & \ref{ModelL2} & 7.308 \,(5.604) & 5.012 \,(2.670) & 5.154 \,(2.461) & 5.617 \,(2.451) & 4.418 \,(1.691) & \textbf{4.314} \,(1.564) & 4.973 \,(2.196) \\
  & & \ref{ModelL3} & 7.383 \,(5.866) & 5.204 \,(3.116) & 5.331 \,(2.492) & 5.428 \,(2.234) & 4.486 \,(1.738) & \textbf{4.400} \,(1.654) & 4.862 \,(1.990) \\
  & & \ref{ModelL4} & 6.976 \,(5.603) & 5.020 \,(2.784) & 5.167 \,(2.513) & 5.508 \,(2.452) & 4.438 \,(1.755) & \textbf{4.346} \,(1.671) & 4.879 \,(2.078) \\
  & & \ref{ModelL5} & 7.260 \,(5.344) & 4.999 \,(2.674) & 5.139 \,(2.480) & 5.546 \,(2.281) & 4.430 \,(1.696) & \textbf{4.322} \,(1.599) & 4.930 \,(2.096) \\
  \cmidrule{2-10}
  & \multirow{5}{*}{0.05} 
    & \ref{ModelL1} & 7.796 \,(6.118) & 5.539 \,(3.241) & 5.606 \,(2.709) & 6.094 \,(2.594) & 4.798 \,(2.058) & \textbf{4.686} \,(1.949) & 5.458 \,(2.761) \\
  & & \ref{ModelL2} & 7.267 \,(5.702) & 5.314 \,(3.125) & 5.256 \,(2.336) & 5.715 \,(2.489) & 4.617 \,(1.842) & \textbf{4.488} \,(1.742) & 5.274 \,(2.407) \\
  & & \ref{ModelL3} & 7.491 \,(5.746) & 5.133 \,(2.736) & 5.469 \,(2.727) & 5.677 \,(2.524) & 4.576 \,(1.805) & \textbf{4.507} \,(1.672) & 5.129 \,(2.156) \\
  & & \ref{ModelL4} & 7.370 \,(5.442) & 5.333 \,(3.252) & 5.339 \,(2.674) & 5.790 \,(2.860) & 4.545 \,(1.743) & \textbf{4.447} \,(1.680) & 5.023 \,(2.147) \\
  & & \ref{ModelL5} & 7.404 \,(6.130) & 5.246 \,(2.836) & 5.913 \,(3.071) & 5.857 \,(2.530) & 4.707 \,(4.856) & \textbf{4.416} \,(1.532) & 5.037 \,(2.049) \\
  \cmidrule{2-10}
  & \multirow{5}{*}{0.10} 
    & \ref{ModelL1} & 9.382 \,(8.015) & 6.461 \,(4.345) & 6.802 \,(3.843) & 7.039 \,(3.079) & 5.643 \,(2.679) & \textbf{5.615} \,(2.581) & 6.679 \,(3.665) \\
  & & \ref{ModelL2} & 7.245 \,(5.129) & 5.522 \,(3.036) & 5.661 \,(2.567) & 6.055 \,(2.659) & 5.184 \,(2.191) & \textbf{5.056} \,(2.070) & 5.654 \,(2.588) \\
  & & \ref{ModelL3} & 7.293 \,(5.170) & 5.375 \,(2.930) & 5.658 \,(2.787) & 5.816 \,(2.607) & 5.203 \,(2.079) & \textbf{5.161} \,(2.116) & 6.025 \,(2.773) \\
  & & \ref{ModelL4} & 7.238 \,(5.755) & 5.309 \,(3.004) & 5.457 \,(2.708) & 5.866 \,(2.707) & 4.640 \,(1.818) & \textbf{4.525} \,(1.643) & 5.150 \,(2.192) \\
  & & \ref{ModelL5} & 7.491 \,(6.107) & 5.487 \,(3.057) & 6.085 \,(3.369) & 5.957 \,(2.613) & 4.763 \,(4.707) & \textbf{4.521} \,(1.753) & 5.236 \,(2.400) \\
\midrule

\multirow{15}{*}{\rotatebox{90}{\textbf{$\mathrm{IVAR}\cdot 10^3$}}} 
  & \multirow{5}{*}{0.01} 
    & \ref{ModelL1} & 3.225 \,(1.269) & \textbf{3.196} \,(1.022) & 3.694 \,(1.368) & 4.868 \,(2.118) & 3.582 \,(1.182) & 3.508 \,(1.083) & 3.863 \,(1.482) \\
  & & \ref{ModelL2} & 3.284 \,(1.352) & \textbf{3.171} \,(1.020) & 3.618 \,(1.281) & 4.718 \,(2.161) & 3.535 \,(1.196) & 3.461 \,(1.110) & 3.818 \,(1.423) \\
  & & \ref{ModelL3} & 3.169 \,(1.284) & \textbf{3.105} \,(0.975) & 3.608 \,(1.296) & 4.525 \,(1.888) & 3.565 \,(1.217) & 3.473 \,(1.128) & 3.816 \,(1.464) \\
  & & \ref{ModelL4} & 3.148 \,(1.213) & \textbf{3.100} \,(0.996) & 3.593 \,(1.437) & 4.790 \,(2.243) & 3.516 \,(1.207) & 3.428 \,(1.129) & 3.832 \,(1.492) \\
  & & \ref{ModelL5} & 3.227 \,(1.305) & \textbf{3.138} \,(0.980) & 3.631 \,(1.328) & 4.796 \,(2.103) & 3.559 \,(1.197) & 3.440 \,(1.092) & 3.836 \,(1.399) \\
  \cmidrule{2-10}
  & \multirow{5}{*}{0.05} 
    & \ref{ModelL1} & \textbf{3.142} \,(1.220) & 3.149 \,(1.024) & 3.851 \,(1.409) & 5.227 \,(2.337) & 3.566 \,(1.218) & 3.487 \,(1.138) & 3.800 \,(1.501) \\
  & & \ref{ModelL2} & 3.326 \,(1.441) & \textbf{3.231} \,(1.083) & 3.815 \,(1.383) & 4.709 \,(2.048) & 3.708 \,(1.360) & 3.647 \,(1.290) & 4.057 \,(1.691) \\
  & & \ref{ModelL3} & 3.201 \,(1.224) & \textbf{3.182} \,(1.050) & 3.672 \,(1.260) & 4.376 \,(1.890) & 3.684 \,(1.367) & 3.596 \,(1.230) & 4.035 \,(1.653) \\
  & & \ref{ModelL4} & 3.209 \,(1.330) & \textbf{3.097} \,(0.973) & 3.607 \,(1.262) & 4.961 \,(2.651) & 3.605 \,(1.247) & 3.496 \,(1.110) & 3.800 \,(1.454) \\
  & & \ref{ModelL5} & 3.300 \,(1.331) & \textbf{3.224} \,(1.070) & 3.765 \,(1.429) & 4.946 \,(2.292) & 3.752 \,(4.551) & 3.531 \,(1.136) & 3.908 \,(1.478) \\
  \cmidrule{2-10}
  & \multirow{5}{*}{0.10} 
    & \ref{ModelL1} & \textbf{3.113} \,(1.164) & 3.177 \,(1.025) & 4.025 \,(1.621) & 5.662 \,(2.753) & 3.663 \,(1.313) & 3.577 \,(1.197) & 3.956 \,(1.562) \\
  & & \ref{ModelL2} & 3.511 \,(1.658) & \textbf{3.430} \,(1.276) & 4.146 \,(1.619) & 4.967 \,(2.223) & 4.259 \,(1.809) & 4.183 \,(1.760) & 4.510 \,(2.108) \\
  & & \ref{ModelL3} & 3.430 \,(1.432) & \textbf{3.410} \,(1.193) & 3.966 \,(1.479) & 4.508 \,(1.841) & 4.295 \,(1.687) & 4.243 \,(1.716) & 4.735 \,(2.195) \\
  & & \ref{ModelL4} & 3.260 \,(1.235) & \textbf{3.196} \,(1.019) & 3.731 \,(1.354) & 5.128 \,(2.581) & 3.676 \,(1.353) & 3.580 \,(1.202) & 3.971 \,(1.565) \\
  & & \ref{ModelL5} & \textbf{3.273} \,(1.328) & 3.305 \,(1.149) & 3.837 \,(1.460) & 4.866 \,(2.170) & 3.776 \,(4.340) & 3.574 \,(1.228) & 3.999 \,(1.645) \\
\midrule

\multirow{15}{*}{\rotatebox{90}{\textbf{$\mathrm{ISB}\cdot 10^3$}}} 
  & \multirow{5}{*}{0.01} 
    & \ref{ModelL1} & 4.420 \,(6.194) & 1.940 \,(2.591) & 1.495 \,(2.026) & \textbf{0.713} \,(1.032) & 0.903 \,(1.253) & 0.904 \,(1.253) & 1.218 \,(1.770) \\
  & & \ref{ModelL2} & 4.024 \,(5.533) & 1.840 \,(2.580) & 1.536 \,(2.213) & 0.899 \,(1.236) & 0.883 \,(1.207) & \textbf{0.853} \,(1.152) & 1.155 \,(1.705) \\
  & & \ref{ModelL3} & 4.213 \,(5.869) & 2.098 \,(3.103) & 1.723 \,(2.224) & \textbf{0.903} \,(1.224) & 0.921 \,(1.295) & 0.927 \,(1.285) & 1.046 \,(1.523) \\
  & & \ref{ModelL4} & 3.828 \,(5.505) & 1.920 \,(2.669) & 1.573 \,(2.146) & \textbf{0.717} \,(1.045) & 0.921 \,(1.264) & 0.918 \,(1.302) & 1.047 \,(1.548) \\
  & & \ref{ModelL5} & 4.032 \,(5.328) & 1.861 \,(2.585) & 1.508 \,(2.178) & \textbf{0.750} \,(0.992) & 0.870 \,(1.237) & 0.882 \,(1.213) & 1.093 \,(1.664) \\
  \cmidrule{2-10}
  & \multirow{5}{*}{0.05} 
    & \ref{ModelL1} & 4.654 \,(6.097) & 2.389 \,(3.186) & 1.754 \,(2.377) & \textbf{0.866} \,(1.185) & 1.231 \,(1.711) & 1.199 \,(1.647) & 1.657 \,(2.411) \\
  & & \ref{ModelL2} & 3.941 \,(5.572) & 2.083 \,(3.034) & 1.440 \,(1.940) & 1.006 \,(1.424) & 0.909 \,(1.249) & \textbf{0.841} \,(1.159) & 1.217 \,(1.766) \\
  & & \ref{ModelL3} & 4.290 \,(5.713) & 1.950 \,(2.687) & 1.797 \,(2.435) & 1.300 \,(1.795) & \textbf{0.892} \,(1.227) & 0.911 \,(1.295) & 1.093 \,(1.625) \\
  & & \ref{ModelL4} & 4.160 \,(5.435) & 2.236 \,(3.250) & 1.732 \,(2.370) & \textbf{0.829} \,(1.150) & 0.940 \,(1.321) & 0.951 \,(1.397) & 1.223 \,(1.723) \\
  & & \ref{ModelL5} & 4.104 \,(6.096) & 2.022 \,(2.719) & 2.148 \,(2.822) & 0.910 \,(1.213) & 0.954 \,(1.299) & \textbf{0.885} \,(1.156) & 1.128 \,(1.553) \\
  \cmidrule{2-10}
  & \multirow{5}{*}{0.10} 
    & \ref{ModelL1} & 6.268 \,(7.990) & 3.283 \,(4.275) & 2.776 \,(3.429) & \textbf{1.377} \,(1.594) & 1.980 \,(2.337) & 2.038 \,(2.304) & 2.722 \,(3.253) \\
  & & \ref{ModelL2} & 3.734 \,(5.001) & 2.092 \,(2.838) & 1.515 \,(2.095) & 1.087 \,(1.545) & 0.924 \,(1.298) & \textbf{0.873} \,(1.218) & 1.144 \,(1.584) \\
  & & \ref{ModelL3} & 3.863 \,(5.084) & 1.964 \,(2.786) & 1.692 \,(2.364) & 1.308 \,(1.949) & \textbf{0.907} \,(1.273) & 0.918 \,(1.337) & 1.289 \,(1.895) \\
  & & \ref{ModelL4} & 3.977 \,(5.675) & 2.112 \,(2.922) & 1.726 \,(2.396) & \textbf{0.738} \,(1.016) & 0.964 \,(1.292) & 0.945 \,(1.238) & 1.178 \,(1.630) \\
  & & \ref{ModelL5} & 4.218 \,(6.058) & 2.181 \,(2.931) & 2.247 \,(3.048) & 1.090 \,(1.646) & 0.987 \,(1.325) & \textbf{0.946} \,(1.292) & 1.236 \,(1.774) \\
\bottomrule
        \end{tabular}}
\caption{Robust location estimation (Section~\ref{Supplement: LocationEstimation}): Comparison of performance measures for Models~\ref{ModelL1}--\ref{ModelL5} with $n=100$. The table reports $\mathrm{MISE}\cdot 10^3$, $\mathrm{IVAR}\cdot 10^3$, and $\mathrm{ISB}\cdot 10^3$, as defined in~\eqref{eq: MISEiSB}, for various contamination levels $\varepsilon\in\set{0.01,0.05,0.1}$. For each experimental setting, the deepest function is identified and the corresponding quantities are evaluated. Reported values are Monte Carlo averages with standard deviations in parentheses, based on $1\,000$ simulations.}
\label{tab: Location_n100}
\end{table}

%% file: Tables_applications/Location_n1000.tex
\begin{table}
\centering
\resizebox{\textwidth}{!}{%
\begin{tabular}{cccccccccc}
\toprule
Metric & $\varepsilon$ & Model & $\RPD_{0.1}$ & $\RPD_{0.5}$ & $\RHD_{0.1}$ & $\RHD_{0.5}$ & \FD{} & \MBD{} & \ID{} \\
\midrule

\multirow{15}{*}{\rotatebox{90}{\textbf{$\mathrm{MISE}\cdot 10^3$}}}
  & \multirow{5}{*}{0.01}
    & \ref{ModelL1} & 4.258 \,(2.945) & 2.806 \,(0.989) & 3.101 \,(0.992) & 3.552 \,(1.282) & 2.477 \,(0.557) & \textbf{2.414} \,(0.489) & 2.853 \,(0.807) \\
  & & \ref{ModelL2} & 4.146 \,(2.997) & 2.776 \,(0.961) & 3.146 \,(1.008) & 3.495 \,(1.237) & 2.474 \,(0.551) & \textbf{2.424} \,(0.496) & 2.842 \,(0.783) \\
  & & \ref{ModelL3} & 4.068 \,(2.779) & 2.801 \,(1.015) & 3.168 \,(1.064) & 3.498 \,(1.284) & 2.502 \,(0.591) & \textbf{2.431} \,(0.517) & 2.832 \,(0.810) \\
  & & \ref{ModelL4} & 4.098 \,(2.610) & 2.811 \,(1.020) & 3.100 \,(1.026) & 3.595 \,(1.510) & 2.467 \,(0.529) & \textbf{2.411} \,(0.480) & 2.813 \,(0.773) \\
  & & \ref{ModelL5} & 4.126 \,(2.765) & 2.780 \,(0.955) & 3.078 \,(0.985) & 3.464 \,(1.213) & 2.470 \,(0.553) & \textbf{2.413} \,(0.485) & 2.850 \,(0.819) \\
  \cmidrule{2-10}
  & \multirow{5}{*}{0.05}
    & \ref{ModelL1} & 4.695 \,(3.344) & 3.095 \,(1.378) & 3.437 \,(1.164) & 3.655 \,(1.221) & 2.750 \,(0.747) & \textbf{2.678} \,(0.690) & 3.124 \,(1.042) \\
  & & \ref{ModelL2} & 4.365 \,(2.936) & 2.897 \,(1.030) & 3.252 \,(1.047) & 3.528 \,(1.179) & 2.642 \,(0.651) & \textbf{2.585} \,(0.607) & 3.011 \,(0.906) \\
  & & \ref{ModelL3} & 4.266 \,(2.867) & 2.890 \,(1.010) & 3.265 \,(1.079) & 3.623 \,(1.325) & 2.667 \,(0.672) & \textbf{2.623} \,(0.613) & 3.121 \,(0.986) \\
  & & \ref{ModelL4} & 4.181 \,(2.611) & 2.854 \,(0.989) & 3.208 \,(0.991) & 3.576 \,(1.235) & 2.516 \,(0.562) & \textbf{2.463} \,(0.509) & 2.893 \,(0.822) \\
  & & \ref{ModelL5} & 4.303 \,(2.822) & 2.858 \,(1.087) & 3.355 \,(1.146) & 3.677 \,(1.361) & 2.513 \,(0.575) & \textbf{2.444} \,(0.492) & 2.895 \,(0.834) \\
  \cmidrule{2-10}
  & \multirow{5}{*}{0.10}
    & \ref{ModelL1} & 6.594 \,(5.522) & 4.164 \,(2.188) & 4.593 \,(1.595) & 4.475 \,(1.393) & 3.639 \,(1.214) & \textbf{3.624} \,(1.172) & 4.236 \,(1.575) \\
  & & \ref{ModelL2} & 4.541 \,(3.074) & 3.077 \,(1.091) & 3.405 \,(1.053) & 3.655 \,(1.189) & 3.120 \,(0.928) & \textbf{3.062} \,(0.860) & 3.547 \,(1.249) \\
  & & \ref{ModelL3} & 4.234 \,(2.578) & \textbf{3.068} \,(1.160) & 3.365 \,(1.087) & 3.669 \,(1.302) & 3.225 \,(0.972) & 3.254 \,(0.972) & 3.772 \,(1.289) \\
  & & \ref{ModelL4} & 4.301 \,(3.047) & 2.906 \,(1.110) & 3.271 \,(1.106) & 3.613 \,(1.314) & 2.568 \,(0.599) & \textbf{2.505} \,(0.542) & 2.940 \,(0.842) \\
  & & \ref{ModelL5} & 4.304 \,(2.726) & 3.054 \,(1.264) & 3.464 \,(1.201) & 3.799 \,(1.403) & 2.578 \,(0.601) & \textbf{2.525} \,(0.557) & 2.954 \,(0.871) \\
\midrule

\multirow{15}{*}{\rotatebox{90}{\textbf{$\mathrm{IVAR}\cdot 10^3$}}}
  & \multirow{5}{*}{0.01}
    & \ref{ModelL1} & 2.152 \,(0.677) & \textbf{2.097} \,(0.461) & 2.716 \,(0.883) & 3.493 \,(1.293) & 2.263 \,(0.504) & 2.206 \,(0.446) & 2.542 \,(0.728) \\
  & & \ref{ModelL2} & 2.134 \,(0.620) & \textbf{2.127} \,(0.495) & 2.728 \,(0.879) & 3.371 \,(1.263) & 2.248 \,(0.517) & 2.197 \,(0.472) & 2.539 \,(0.715) \\
  & & \ref{ModelL3} & 2.136 \,(0.662) & \textbf{2.104} \,(0.476) & 2.703 \,(0.884) & 3.393 \,(1.297) & 2.278 \,(0.545) & 2.203 \,(0.480) & 2.521 \,(0.735) \\
  & & \ref{ModelL4} & 2.166 \,(0.643) & \textbf{2.102} \,(0.457) & 2.729 \,(0.943) & 3.525 \,(1.528) & 2.264 \,(0.505) & 2.199 \,(0.453) & 2.509 \,(0.698) \\
  & & \ref{ModelL5} & 2.133 \,(0.657) & \textbf{2.106} \,(0.463) & 2.697 \,(0.877) & 3.391 \,(1.231) & 2.261 \,(0.513) & 2.205 \,(0.448) & 2.537 \,(0.729) \\
  \cmidrule{2-10}
  & \multirow{5}{*}{0.05}
    & \ref{ModelL1} & 2.155 \,(0.604) & \textbf{2.120} \,(0.465) & 2.886 \,(0.961) & 3.459 \,(1.206) & 2.284 \,(0.522) & 2.227 \,(0.463) & 2.531 \,(0.717) \\
  & & \ref{ModelL2} & 2.245 \,(0.753) & \textbf{2.191} \,(0.522) & 2.878 \,(0.943) & 3.380 \,(1.176) & 2.405 \,(0.606) & 2.365 \,(0.574) & 2.685 \,(0.832) \\
  & & \ref{ModelL3} & 2.225 \,(0.683) & \textbf{2.178} \,(0.517) & 2.836 \,(0.952) & 3.424 \,(1.317) & 2.443 \,(0.645) & 2.404 \,(0.585) & 2.759 \,(0.879) \\
  & & \ref{ModelL4} & 2.148 \,(0.633) & \textbf{2.119} \,(0.472) & 2.854 \,(0.927) & 3.505 \,(1.257) & 2.299 \,(0.527) & 2.243 \,(0.474) & 2.562 \,(0.727) \\
  & & \ref{ModelL5} & 2.192 \,(0.691) & \textbf{2.160} \,(0.502) & 2.788 \,(0.931) & 3.575 \,(1.379) & 2.288 \,(0.534) & 2.232 \,(0.462) & 2.567 \,(0.736) \\
  \cmidrule{2-10}
  & \multirow{5}{*}{0.10}
    & \ref{ModelL1} & \textbf{2.114} \,(0.623) & 2.128 \,(0.485) & 3.025 \,(1.047) & 3.686 \,(1.352) & 2.315 \,(0.543) & 2.255 \,(0.477) & 2.547 \,(0.704) \\
  & & \ref{ModelL2} & 2.417 \,(0.928) & \textbf{2.387} \,(0.680) & 3.097 \,(0.974) & 3.530 \,(1.203) & 2.883 \,(0.906) & 2.816 \,(0.834) & 3.225 \,(1.187) \\
  & & \ref{ModelL3} & 2.396 \,(0.777) & \textbf{2.391} \,(0.650) & 3.015 \,(0.984) & 3.493 \,(1.279) & 3.004 \,(0.953) & 3.024 \,(0.950) & 3.394 \,(1.213) \\
  & & \ref{ModelL4} & 2.216 \,(0.685) & \textbf{2.163} \,(0.511) & 2.866 \,(1.000) & 3.534 \,(1.332) & 2.338 \,(0.562) & 2.283 \,(0.527) & 2.610 \,(0.769) \\
  & & \ref{ModelL5} & \textbf{2.264} \,(0.723) & 2.272 \,(0.596) & 2.894 \,(0.990) & 3.652 \,(1.392) & 2.344 \,(0.565) & 2.297 \,(0.517) & 2.612 \,(0.762) \\
\midrule

\multirow{15}{*}{\rotatebox{90}{\textbf{$\mathrm{ISB}\cdot 10^3$}}}
  & \multirow{5}{*}{0.01}
    & \ref{ModelL1} & 2.125 \,(2.980) & 0.727 \,(0.988) & 0.409 \,(0.534) & \textbf{0.090} \,(0.139) & 0.234 \,(0.297) & 0.227 \,(0.292) & 0.333 \,(0.472) \\
  & & \ref{ModelL2} & 2.031 \,(3.002) & 0.668 \,(0.962) & 0.442 \,(0.591) & \textbf{0.153} \,(0.209) & 0.247 \,(0.299) & 0.247 \,(0.302) & 0.326 \,(0.443) \\
  & & \ref{ModelL3} & 1.950 \,(2.788) & 0.716 \,(1.002) & 0.489 \,(0.655) & \textbf{0.135} \,(0.192) & 0.244 \,(0.311) & 0.248 \,(0.314) & 0.333 \,(0.458) \\
  & & \ref{ModelL4} & 1.951 \,(2.611) & 0.727 \,(1.031) & 0.396 \,(0.515) & \textbf{0.101} \,(0.147) & 0.223 \,(0.285) & 0.231 \,(0.310) & 0.327 \,(0.450) \\
  & & \ref{ModelL5} & 2.012 \,(2.792) & 0.693 \,(0.940) & 0.406 \,(0.532) & \textbf{0.103} \,(0.150) & 0.229 \,(0.293) & 0.228 \,(0.290) & 0.336 \,(0.484) \\
  \cmidrule{2-10}
  & \multirow{5}{*}{0.05}
    & \ref{ModelL1} & 2.559 \,(3.338) & 0.994 \,(1.360) & 0.577 \,(0.703) & \textbf{0.227} \,(0.236) & 0.486 \,(0.567) & 0.471 \,(0.574) & 0.616 \,(0.816) \\
  & & \ref{ModelL2} & 2.140 \,(2.929) & 0.726 \,(1.008) & 0.399 \,(0.542) & \textbf{0.178} \,(0.263) & 0.258 \,(0.327) & 0.241 \,(0.309) & 0.351 \,(0.516) \\
  & & \ref{ModelL3} & 2.061 \,(2.867) & 0.732 \,(1.007) & 0.455 \,(0.604) & \textbf{0.229} \,(0.310) & 0.246 \,(0.315) & 0.240 \,(0.315) & 0.387 \,(0.581) \\
  & & \ref{ModelL4} & 2.051 \,(2.602) & 0.754 \,(1.024) & 0.379 \,(0.512) & \textbf{0.102} \,(0.147) & 0.238 \,(0.305) & 0.240 \,(0.317) & 0.354 \,(0.507) \\
  & & \ref{ModelL5} & 2.131 \,(2.821) & 0.717 \,(1.077) & 0.592 \,(0.759) & \textbf{0.134} \,(0.194) & 0.245 \,(0.312) & 0.232 \,(0.290) & 0.351 \,(0.522) \\
  \cmidrule{2-10}
  & \multirow{5}{*}{0.10}
    & \ref{ModelL1} & 4.499 \,(5.543) & 2.056 \,(2.182) & 1.595 \,(1.280) & \textbf{0.822} \,(0.465) & 1.344 \,(1.069) & 1.389 \,(1.077) & 1.712 \,(1.447) \\
  & & \ref{ModelL2} & 2.146 \,(2.982) & 0.712 \,(1.000) & 0.335 \,(0.475) & \textbf{0.156} \,(0.222) & 0.263 \,(0.305) & 0.271 \,(0.331) & 0.351 \,(0.523) \\
  & & \ref{ModelL3} & 1.859 \,(2.528) & 0.699 \,(1.047) & 0.377 \,(0.525) & \textbf{0.207} \,(0.292) & 0.248 \,(0.322) & 0.257 \,(0.338) & 0.408 \,(0.540) \\
  & & \ref{ModelL4} & 2.104 \,(3.063) & 0.762 \,(1.116) & 0.431 \,(0.589) & \textbf{0.110} \,(0.154) & 0.251 \,(0.309) & 0.243 \,(0.306) & 0.353 \,(0.476) \\
  & & \ref{ModelL5} & 2.061 \,(2.700) & 0.802 \,(1.224) & 0.595 \,(0.766) & \textbf{0.179} \,(0.252) & 0.255 \,(0.330) & 0.249 \,(0.323) & 0.365 \,(0.545) \\
\bottomrule
\end{tabular}}
\caption{Robust location estimation (Section~\ref{Supplement: LocationEstimation}): Comparison of performance measures for Models~\ref{ModelL1}--\ref{ModelL5} with $n=1000$. The table reports $\mathrm{MISE}\cdot 10^3$, $\mathrm{IVAR}\cdot 10^3$, and $\mathrm{ISB}\cdot 10^3$, as defined in~\eqref{eq: MISEiSB}, for various contamination levels $\varepsilon\in\set{0.01,0.05,0.1}$. For each experimental setting, the deepest function is identified and the corresponding quantities are evaluated. Reported values are Monte Carlo averages with standard deviations in parentheses, based on $1\,000$ simulations.}
\label{tab: Location_n1000}
\end{table}

%% file: CSDA_Combined.bbl
\begin{thebibliography}{53}
\expandafter\ifx\csname natexlab\endcsname\relax\def\natexlab#1{#1}\fi
\providecommand{\url}[1]{\texttt{#1}}
\providecommand{\href}[2]{#2}
\providecommand{\path}[1]{#1}
\providecommand{\DOIprefix}{doi:}
\providecommand{\ArXivprefix}{arXiv:}
\providecommand{\URLprefix}{URL: }
\providecommand{\Pubmedprefix}{pmid:}
\providecommand{\doi}[1]{\href{http://dx.doi.org/#1}{\path{#1}}}
\providecommand{\Pubmed}[1]{\href{pmid:#1}{\path{#1}}}
\providecommand{\bibinfo}[2]{#2}
\ifx\xfnm\relax \def\xfnm[#1]{\unskip,\space#1}\fi
\bibitem[{Am\'endola et~al.(2020)Am\'endola, Engstr\"om and
  Haase}]{Amendola_etal2020}
\bibinfo{author}{Am\'endola, C.}, \bibinfo{author}{Engstr\"om, A.},
  \bibinfo{author}{Haase, C.}, \bibinfo{year}{2020}.
\newblock \bibinfo{title}{Maximum number of modes of {G}aussian mixtures}.
\newblock \bibinfo{journal}{Inf. Inference} \bibinfo{volume}{9},
  \bibinfo{pages}{587--600}.
\bibitem[{Arribas-Gil and Romo(2014)}]{Arribas_Romo2014}
\bibinfo{author}{Arribas-Gil, A.}, \bibinfo{author}{Romo, J.},
  \bibinfo{year}{2014}.
\newblock \bibinfo{title}{Shape outlier detection and visualization for
  functional data: the outliergram}.
\newblock \bibinfo{journal}{Biostatistics} \bibinfo{volume}{15},
  \bibinfo{pages}{603--619}.
\bibitem[{Boente et~al.(2014)Boente, Salibi\'an~Barrera and
  Tyler}]{Boente_etal2014}
\bibinfo{author}{Boente, G.}, \bibinfo{author}{Salibi\'an~Barrera, M.},
  \bibinfo{author}{Tyler, D.E.}, \bibinfo{year}{2014}.
\newblock \bibinfo{title}{A characterization of elliptical distributions and
  some optimality properties of principal components for functional data}.
\newblock \bibinfo{journal}{J. Multivariate Anal.} \bibinfo{volume}{131},
  \bibinfo{pages}{254--264}.
\bibitem[{Bo\v{c}inec et~al.(2026)Bo\v{c}inec, Nagy and
  Yeon}]{Bocinec_etal2026}
\bibinfo{author}{Bo\v{c}inec, F.}, \bibinfo{author}{Nagy, S.},
  \bibinfo{author}{Yeon, H.}, \bibinfo{year}{2026}.
\newblock \bibinfo{title}{Projection depth for functional data: Theoretical
  properties}.
\newblock \bibinfo{note}{ArXiv:2512.20452}.
\bibitem[{Briend et~al.(2025)Briend, Lugosi and Oliveira}]{Briend_etal2025}
\bibinfo{author}{Briend, S.}, \bibinfo{author}{Lugosi, G.},
  \bibinfo{author}{Oliveira, R.I.}, \bibinfo{year}{2025}.
\newblock \bibinfo{title}{On the quality of randomized approximations of
  {T}ukey's depth}.
\newblock \bibinfo{journal}{SIAM J. Math. Data Sci.} \bibinfo{volume}{7},
  \bibinfo{pages}{1441--1464}.
\bibitem[{Chakraborty and Chaudhuri(2014a)}]{Chakraborty_Chaudhuri2014b}
\bibinfo{author}{Chakraborty, A.}, \bibinfo{author}{Chaudhuri, P.},
  \bibinfo{year}{2014}a.
\newblock \bibinfo{title}{The deepest point for distributions in infinite
  dimensional spaces}.
\newblock \bibinfo{journal}{Stat. Methodol.} \bibinfo{volume}{20},
  \bibinfo{pages}{27--39}.
\bibitem[{Chakraborty and Chaudhuri(2014b)}]{Chakraborty_Chaudhuri2014}
\bibinfo{author}{Chakraborty, A.}, \bibinfo{author}{Chaudhuri, P.},
  \bibinfo{year}{2014}b.
\newblock \bibinfo{title}{On data depth in infinite dimensional spaces}.
\newblock \bibinfo{journal}{Ann. Inst. Statist. Math.} \bibinfo{volume}{66},
  \bibinfo{pages}{303--324}.
\bibitem[{Chaudhuri(1996)}]{Chaudhuri1996}
\bibinfo{author}{Chaudhuri, P.}, \bibinfo{year}{1996}.
\newblock \bibinfo{title}{On a geometric notion of quantiles for multivariate
  data}.
\newblock \bibinfo{journal}{J. Amer. Statist. Assoc.} \bibinfo{volume}{91},
  \bibinfo{pages}{862--872}.
\bibitem[{Chenouri and Small(2012)}]{CS12}
\bibinfo{author}{Chenouri, S.}, \bibinfo{author}{Small, C.G.},
  \bibinfo{year}{2012}.
\newblock \bibinfo{title}{A nonparametric multivariate multisample test based
  on data depth}.
\newblock \bibinfo{journal}{Electron. J. Statist.} \bibinfo{volume}{6},
  \bibinfo{pages}{760--782}.
\bibitem[{Cram\'er(1946)}]{Cramer1946}
\bibinfo{author}{Cram\'er, H.}, \bibinfo{year}{1946}.
\newblock \bibinfo{title}{Mathematical Methods of Statistics}. volume
  \bibinfo{volume}{vol. 9} of \textit{\bibinfo{series}{Princeton Mathematical
  Series}}.
\newblock \bibinfo{publisher}{Princeton University Press, Princeton, NJ}.
\bibitem[{Cuevas et~al.(2007)Cuevas, Febrero and Fraiman}]{Cuevas_etal2007}
\bibinfo{author}{Cuevas, A.}, \bibinfo{author}{Febrero, M.},
  \bibinfo{author}{Fraiman, R.}, \bibinfo{year}{2007}.
\newblock \bibinfo{title}{Robust estimation and classification for functional
  data via projection-based depth notions}.
\newblock \bibinfo{journal}{Comput. Statist.} \bibinfo{volume}{22},
  \bibinfo{pages}{481--496}.
\bibitem[{Dai and Genton(2018)}]{DaiGenton2018}
\bibinfo{author}{Dai, W.}, \bibinfo{author}{Genton, M.G.},
  \bibinfo{year}{2018}.
\newblock \bibinfo{title}{Multivariate functional data visualization and
  outlier detection}.
\newblock \bibinfo{journal}{J. Comput. Graph. Statist.} \bibinfo{volume}{27},
  \bibinfo{pages}{923--934}.
\bibitem[{Dai et~al.(2020)Dai, Mrkvi\v{c}ka, Sun and Genton}]{Dai_etal2020}
\bibinfo{author}{Dai, W.}, \bibinfo{author}{Mrkvi\v{c}ka, T.},
  \bibinfo{author}{Sun, Y.}, \bibinfo{author}{Genton, M.G.},
  \bibinfo{year}{2020}.
\newblock \bibinfo{title}{Functional outlier detection and taxonomy by
  sequential transformations}.
\newblock \bibinfo{journal}{Comput. Statist. Data Anal.} \bibinfo{volume}{149},
  \bibinfo{pages}{106960, 17}.
\bibitem[{Donoho(1982)}]{Donoho1982}
\bibinfo{author}{Donoho, D.L.}, \bibinfo{year}{1982}.
\newblock \bibinfo{title}{Breakdown properties of multivariate location
  estimators}.
\newblock \bibinfo{note}{Qualifying paper, Harvard University}.
\bibitem[{Dudley(2002)}]{Dudley2002}
\bibinfo{author}{Dudley, R.M.}, \bibinfo{year}{2002}.
\newblock \bibinfo{title}{Real Analysis and Probability}.
  volume~\bibinfo{volume}{74} of \textit{\bibinfo{series}{Cambridge Studies in
  Advanced Mathematics}}.
\newblock \bibinfo{publisher}{Cambridge University Press},
  \bibinfo{address}{Cambridge}.
\newblock \bibinfo{note}{Revised reprint of the 1989 original}.
\bibitem[{Dyckerhoff(2004)}]{Dyckerhoff2004}
\bibinfo{author}{Dyckerhoff, R.}, \bibinfo{year}{2004}.
\newblock \bibinfo{title}{Data depths satisfying the projection property}.
\newblock \bibinfo{journal}{Allg. Stat. Arch.} \bibinfo{volume}{88},
  \bibinfo{pages}{163--190}.
\bibitem[{Eddelbuettel et~al.(2024)Eddelbuettel, Francois, Bates, Ni and
  Sanderson}]{RcppArmadillo}
\bibinfo{author}{Eddelbuettel, D.}, \bibinfo{author}{Francois, R.},
  \bibinfo{author}{Bates, D.}, \bibinfo{author}{Ni, B.},
  \bibinfo{author}{Sanderson, C.}, \bibinfo{year}{2024}.
\newblock \bibinfo{title}{RcppArmadillo: `Rcpp' integration for the `Armadillo'
  templated linear algebra library}.
\newblock \URLprefix \url{https://CRAN.R-project.org/package=RcppArmadillo}.
  \bibinfo{note}{{R} package version 14.2.2-1}.
\bibitem[{Elmore et~al.(2006)Elmore, Hettmansperger and Xuan}]{Elmore_etal2006}
\bibinfo{author}{Elmore, R.T.}, \bibinfo{author}{Hettmansperger, T.P.},
  \bibinfo{author}{Xuan, F.}, \bibinfo{year}{2006}.
\newblock \bibinfo{title}{Spherical data depth and a multivariate median}, in:
  \bibinfo{booktitle}{Data Depth: Robust Multivariate Analysis, Computational
  Geometry and Applications}. \bibinfo{publisher}{Amer. Math. Soc., Providence,
  RI}. volume~\bibinfo{volume}{72} of \textit{\bibinfo{series}{DIMACS Ser.
  Discrete Math. Theoret. Comput. Sci.}}, pp. \bibinfo{pages}{87--101}.
\bibitem[{Fanaee-T and Gama(2014)}]{Fanaee_etal2014}
\bibinfo{author}{Fanaee-T, H.}, \bibinfo{author}{Gama, J.},
  \bibinfo{year}{2014}.
\newblock \bibinfo{title}{Event labeling combining ensemble detectors and
  background knowledge}.
\newblock \bibinfo{journal}{Progress in Artificial Intelligence}
  \bibinfo{volume}{2}, \bibinfo{pages}{113--127}.
\bibitem[{Fraiman and Muniz(2001)}]{Fraiman_Muniz2001}
\bibinfo{author}{Fraiman, R.}, \bibinfo{author}{Muniz, G.},
  \bibinfo{year}{2001}.
\newblock \bibinfo{title}{Trimmed means for functional data}.
\newblock \bibinfo{journal}{Test} \bibinfo{volume}{10},
  \bibinfo{pages}{419--440}.
\bibitem[{Helander et~al.(2020)Helander, Van~Bever, Rantala and
  Ilmonen}]{Helander_etal2020}
\bibinfo{author}{Helander, S.}, \bibinfo{author}{Van~Bever, G.},
  \bibinfo{author}{Rantala, S.}, \bibinfo{author}{Ilmonen, P.},
  \bibinfo{year}{2020}.
\newblock \bibinfo{title}{Pareto depth for functional data}.
\newblock \bibinfo{journal}{Statistics} \bibinfo{volume}{54},
  \bibinfo{pages}{182--204}.
\bibitem[{Hlubinka et~al.(2015)Hlubinka, Gijbels, Omelka and
  Nagy}]{Hlubinka_etal2015}
\bibinfo{author}{Hlubinka, D.}, \bibinfo{author}{Gijbels, I.},
  \bibinfo{author}{Omelka, M.}, \bibinfo{author}{Nagy, S.},
  \bibinfo{year}{2015}.
\newblock \bibinfo{title}{Integrated data depth for smooth functions and its
  application in supervised classification}.
\newblock \bibinfo{journal}{Comput. Statist.} \bibinfo{volume}{30},
  \bibinfo{pages}{1011--1031}.
\bibitem[{Hsing and Eubank(2015)}]{Hsing_Eubank2015}
\bibinfo{author}{Hsing, T.}, \bibinfo{author}{Eubank, R.},
  \bibinfo{year}{2015}.
\newblock \bibinfo{title}{Theoretical Foundations of Functional Data Analysis,
  with an Introduction to Linear Operators}.
\newblock Wiley Series in Probability and Statistics, \bibinfo{publisher}{John
  Wiley \& Sons, Ltd., Chichester}.
\bibitem[{Huber(1985)}]{Huber1985}
\bibinfo{author}{Huber, P.J.}, \bibinfo{year}{1985}.
\newblock \bibinfo{title}{Projection pursuit}.
\newblock \bibinfo{journal}{Ann. Statist.} \bibinfo{volume}{13},
  \bibinfo{pages}{435--525}.
\bibitem[{Huber and Ronchetti(2009)}]{Huber_Ronchetti2009}
\bibinfo{author}{Huber, P.J.}, \bibinfo{author}{Ronchetti, E.M.},
  \bibinfo{year}{2009}.
\newblock \bibinfo{title}{Robust Statistics}.
\newblock Wiley Series in Probability and Statistics. \bibinfo{edition}{second}
  ed., \bibinfo{publisher}{John Wiley \& Sons, Inc., Hoboken, NJ}.
\bibitem[{Hubert et~al.(2017)Hubert, Rousseeuw and Segaert}]{Hubert_etal2017}
\bibinfo{author}{Hubert, M.}, \bibinfo{author}{Rousseeuw, P.},
  \bibinfo{author}{Segaert, P.}, \bibinfo{year}{2017}.
\newblock \bibinfo{title}{Multivariate and functional classification using
  depth and distance}.
\newblock \bibinfo{journal}{Adv. Data Anal. Classif.} \bibinfo{volume}{11},
  \bibinfo{pages}{445--466}.
\bibitem[{Ieva et~al.(2019)Ieva, Paganoni, Romo and
  Tarabelloni}]{roahd_package}
\bibinfo{author}{Ieva, F.}, \bibinfo{author}{Paganoni, A.M.},
  \bibinfo{author}{Romo, J.}, \bibinfo{author}{Tarabelloni, N.},
  \bibinfo{year}{2019}.
\newblock \bibinfo{title}{roahd package: Robust analysis of high dimensional
  data}.
\newblock \bibinfo{journal}{{The R Journal}} \bibinfo{volume}{11},
  \bibinfo{pages}{291--307}.
\bibitem[{Jim\'enez-Var\'on et~al.(2024)Jim\'enez-Var\'on, Harrou and
  Sun}]{Jimenez-Varon_etal2024}
\bibinfo{author}{Jim\'enez-Var\'on, C.F.}, \bibinfo{author}{Harrou, F.},
  \bibinfo{author}{Sun, Y.}, \bibinfo{year}{2024}.
\newblock \bibinfo{title}{Pointwise data depth for univariate and multivariate
  functional outlier detection}.
\newblock \bibinfo{journal}{Environmetrics} \bibinfo{volume}{35},
  \bibinfo{pages}{Paper No. e2851, 19}.
\bibitem[{Kong and Zuo(2010)}]{Kong_Zuo2010}
\bibinfo{author}{Kong, L.}, \bibinfo{author}{Zuo, Y.}, \bibinfo{year}{2010}.
\newblock \bibinfo{title}{Smooth depth contours characterize the underlying
  distribution}.
\newblock \bibinfo{journal}{J. Multivariate Anal.} \bibinfo{volume}{101},
  \bibinfo{pages}{2222--2226}.
\bibitem[{Kuo(1975)}]{Kuo1975}
\bibinfo{author}{Kuo, H.H.}, \bibinfo{year}{1975}.
\newblock \bibinfo{title}{Gaussian Measures in {B}anach Spaces}. volume
  \bibinfo{volume}{Vol. 463} of \textit{\bibinfo{series}{Lecture Notes in
  Mathematics}}.
\newblock \bibinfo{publisher}{Springer-Verlag, Berlin-New York}.
\bibitem[{Li et~al.(2012)Li, Cuesta-Albertos and Liu}]{Li_etal2012}
\bibinfo{author}{Li, J.}, \bibinfo{author}{Cuesta-Albertos, J.A.},
  \bibinfo{author}{Liu, R.Y.}, \bibinfo{year}{2012}.
\newblock \bibinfo{title}{{$DD$}-classifier: nonparametric classification
  procedure based on {$DD$}-plot}.
\newblock \bibinfo{journal}{J. Amer. Statist. Assoc.} \bibinfo{volume}{107},
  \bibinfo{pages}{737--753}.
\bibitem[{Liu(1990)}]{Liu1990}
\bibinfo{author}{Liu, R.Y.}, \bibinfo{year}{1990}.
\newblock \bibinfo{title}{On a notion of data depth based on random simplices}.
\newblock \bibinfo{journal}{Ann. Statist.} \bibinfo{volume}{18},
  \bibinfo{pages}{405--414}.
\bibitem[{Liu and Zuo(2014)}]{Liu_Zuo2014b}
\bibinfo{author}{Liu, X.}, \bibinfo{author}{Zuo, Y.}, \bibinfo{year}{2014}.
\newblock \bibinfo{title}{Computing projection depth and its associated
  estimators}.
\newblock \bibinfo{journal}{Stat. Comput.} \bibinfo{volume}{24},
  \bibinfo{pages}{51--63}.
\bibitem[{L{\'o}pez-Pintado and Romo(2009)}]{Lopez_Romo2009}
\bibinfo{author}{L{\'o}pez-Pintado, S.}, \bibinfo{author}{Romo, J.},
  \bibinfo{year}{2009}.
\newblock \bibinfo{title}{On the concept of depth for functional data}.
\newblock \bibinfo{journal}{J. Amer. Statist. Assoc.} \bibinfo{volume}{104},
  \bibinfo{pages}{718--734}.
\bibitem[{Mendro\v{s} and Nagy(2025)}]{Mendros_Nagy2025}
\bibinfo{author}{Mendro\v{s}, E.}, \bibinfo{author}{Nagy, S.},
  \bibinfo{year}{2025}.
\newblock \bibinfo{title}{The spherical depth for functional data}, in:
  \bibinfo{booktitle}{New Trends in Functional Statistics and Related Fields}.
  \bibinfo{publisher}{Springer, Cham}. Contrib. Stat., pp.
  \bibinfo{pages}{401--408}.
\bibitem[{Mosler(2002)}]{Mosler2002}
\bibinfo{author}{Mosler, K.}, \bibinfo{year}{2002}.
\newblock \bibinfo{title}{Multivariate Dispersion, Central Regions and Depth:
  The Lift Zonoid Approach}. volume \bibinfo{volume}{165} of
  \textit{\bibinfo{series}{Lecture Notes in Statistics}}.
\newblock \bibinfo{publisher}{Springer-Verlag}, \bibinfo{address}{Berlin}.
\bibitem[{Mosler(2013)}]{Mosler2013}
\bibinfo{author}{Mosler, K.}, \bibinfo{year}{2013}.
\newblock \bibinfo{title}{Depth statistics}, in: \bibinfo{editor}{Becker, C.},
  \bibinfo{editor}{Fried, R.}, \bibinfo{editor}{Kuhnt, S.} (Eds.),
  \bibinfo{booktitle}{Robustness and Complex Data Structures}.
  \bibinfo{publisher}{Springer, Heidelberg}, pp. \bibinfo{pages}{17--34}.
\bibitem[{Nagy et~al.(2017)Nagy, Gijbels and Hlubinka}]{Nagy_etal2017}
\bibinfo{author}{Nagy, S.}, \bibinfo{author}{Gijbels, I.},
  \bibinfo{author}{Hlubinka, D.}, \bibinfo{year}{2017}.
\newblock \bibinfo{title}{Depth-based recognition of shape outlying functions}.
\newblock \bibinfo{journal}{J. Comput. Graph. Statist.} \bibinfo{volume}{26},
  \bibinfo{pages}{883--893}.
\bibitem[{Nagy et~al.(2016)Nagy, Gijbels, Omelka and Hlubinka}]{Nagy_etal2016}
\bibinfo{author}{Nagy, S.}, \bibinfo{author}{Gijbels, I.},
  \bibinfo{author}{Omelka, M.}, \bibinfo{author}{Hlubinka, D.},
  \bibinfo{year}{2016}.
\newblock \bibinfo{title}{Integrated depth for functional data: Statistical
  properties and consistency}.
\newblock \bibinfo{journal}{ESAIM Probab. Stat.} \bibinfo{volume}{20},
  \bibinfo{pages}{95--130}.
\bibitem[{Nagy et~al.(2026)Nagy, Mrkvi{\v{c}}ka and El{\'\i}as}]{Nagy_etal2025}
\bibinfo{author}{Nagy, S.}, \bibinfo{author}{Mrkvi{\v{c}}ka, T.},
  \bibinfo{author}{El{\'\i}as, A.}, \bibinfo{year}{2026}.
\newblock \bibinfo{title}{Which depth to use to construct functional boxplots?}
\newblock \bibinfo{journal}{Statist. Sci.} \bibinfo{note}{To appear}.
\bibitem[{Narisetty and Nair(2016)}]{Narisetty_Nair2016}
\bibinfo{author}{Narisetty, N.N.}, \bibinfo{author}{Nair, V.N.},
  \bibinfo{year}{2016}.
\newblock \bibinfo{title}{Extremal depth for functional data and applications}.
\newblock \bibinfo{journal}{J. Amer. Statist. Assoc.} \bibinfo{volume}{111},
  \bibinfo{pages}{1705--1714}.
\bibitem[{Ojo et~al.(2022)Ojo, Fern\'andez~Anta, Lillo and
  Sguera}]{Ojo_etal2022}
\bibinfo{author}{Ojo, O.T.}, \bibinfo{author}{Fern\'andez~Anta, A.},
  \bibinfo{author}{Lillo, R.E.}, \bibinfo{author}{Sguera, C.},
  \bibinfo{year}{2022}.
\newblock \bibinfo{title}{Detecting and classifying outliers in big functional
  data}.
\newblock \bibinfo{journal}{Adv. Data Anal. Classif.} \bibinfo{volume}{16},
  \bibinfo{pages}{725--760}.
\bibitem[{Ojo et~al.(2023)Ojo, Lillo and Anta}]{fdaoutlier_package}
\bibinfo{author}{Ojo, O.T.}, \bibinfo{author}{Lillo, R.E.},
  \bibinfo{author}{Anta, A.F.}, \bibinfo{year}{2023}.
\newblock \bibinfo{title}{fdaoutlier package: Outlier detection tools for
  functional data analysis}.
\newblock \DOIprefix\doi{10.32614/CRAN.package.fdaoutlier}. \bibinfo{note}{{R}
  package version 0.2.1}.
\bibitem[{Pokotylo et~al.(2024)Pokotylo, Mozharovskyi, Dyckerhoff and
  Nagy}]{ddalpha}
\bibinfo{author}{Pokotylo, O.}, \bibinfo{author}{Mozharovskyi, P.},
  \bibinfo{author}{Dyckerhoff, R.}, \bibinfo{author}{Nagy, S.},
  \bibinfo{year}{2024}.
\newblock \bibinfo{title}{ddalpha: Depth-based classification and calculation
  of data depth}.
\newblock \DOIprefix\doi{10.32614/CRAN.package.ddalpha}. \bibinfo{note}{{R}
  package version 1.3.16}.
\bibitem[{Ramsay and Chenouri(2024)}]{RamsayChenouri2024}
\bibinfo{author}{Ramsay, K.}, \bibinfo{author}{Chenouri, S.},
  \bibinfo{year}{2024}.
\newblock \bibinfo{title}{Robust nonparametric hypothesis tests for differences
  in the covariance structure of functional data}.
\newblock \bibinfo{journal}{Canad. J. Statist.} \bibinfo{volume}{52},
  \bibinfo{pages}{43--78}.
\bibitem[{Schwarz(1978)}]{Schwarz1978}
\bibinfo{author}{Schwarz, G.}, \bibinfo{year}{1978}.
\newblock \bibinfo{title}{Estimating the dimension of a model}.
\newblock \bibinfo{journal}{Ann. Statist.} \bibinfo{volume}{6},
  \bibinfo{pages}{461--464}.
\bibitem[{Scrucca et~al.(2025)Scrucca, Fraley, Murphy and
  Raftery}]{mclust_package}
\bibinfo{author}{Scrucca, L.}, \bibinfo{author}{Fraley, C.},
  \bibinfo{author}{Murphy, T.B.}, \bibinfo{author}{Raftery, A.E.},
  \bibinfo{year}{2025}.
\newblock \bibinfo{title}{mclust: Gaussian mixture modelling for model-based
  clustering, classification, and density estimation}.
\newblock \DOIprefix\doi{10.32614/CRAN.package.mclust}. \bibinfo{note}{{R}
  package version 6.1.2}.
\bibitem[{Shao et~al.(2022)Shao, Zuo and Luo}]{Shao_etal2022}
\bibinfo{author}{Shao, W.}, \bibinfo{author}{Zuo, Y.}, \bibinfo{author}{Luo,
  J.}, \bibinfo{year}{2022}.
\newblock \bibinfo{title}{Employing the {MCMC} technique to compute the
  projection depth in high dimensions}.
\newblock \bibinfo{journal}{J. Comput. Appl. Math.} \bibinfo{volume}{411},
  \bibinfo{pages}{Paper No. 114278, 16}.
\bibitem[{Stahel(1981)}]{Stahel1981}
\bibinfo{author}{Stahel, W.A.}, \bibinfo{year}{1981}.
\newblock \bibinfo{title}{Robuste {S}ch\"atzungen: {I}nfinitesimale
  {O}ptimalit\"at und {S}ch\"atzungen von {K}ovarianzmatrizen}.
\newblock \bibinfo{note}{Ph.D. thesis, ETH Z\"urich}.
\bibitem[{Sun and Genton(2011)}]{SunGenton2011}
\bibinfo{author}{Sun, Y.}, \bibinfo{author}{Genton, M.G.},
  \bibinfo{year}{2011}.
\newblock \bibinfo{title}{Functional boxplots}.
\newblock \bibinfo{journal}{J. Comput. Graph. Statist.} \bibinfo{volume}{20},
  \bibinfo{pages}{316--334}.
\bibitem[{Tukey(1975)}]{Tukey1975}
\bibinfo{author}{Tukey, J.W.}, \bibinfo{year}{1975}.
\newblock \bibinfo{title}{Mathematics and the picturing of data}, in:
  \bibinfo{booktitle}{Proceedings of the {I}nternational {C}ongress of
  {M}athematicians ({V}ancouver, {B}. {C}., 1974), {V}ol. 2},
  \bibinfo{publisher}{Canad. Math. Congress, Montreal, Que}. pp.
  \bibinfo{pages}{523--531}.
\bibitem[{Yeon et~al.(2025)Yeon, Dai and L\'opez-Pintado}]{Yeon_etal2025}
\bibinfo{author}{Yeon, H.}, \bibinfo{author}{Dai, X.},
  \bibinfo{author}{L\'opez-Pintado, S.}, \bibinfo{year}{2025}.
\newblock \bibinfo{title}{Regularized halfspace depth for functional data}.
\newblock \bibinfo{journal}{J. R. Stat. Soc. Ser. B. Stat. Methodol.}
  \bibinfo{volume}{87}, \bibinfo{pages}{1553--1575}.
\bibitem[{Zuo(2003)}]{Zuo2003}
\bibinfo{author}{Zuo, Y.}, \bibinfo{year}{2003}.
\newblock \bibinfo{title}{Projection-based depth functions and associated
  medians}.
\newblock \bibinfo{journal}{Ann. Statist.} \bibinfo{volume}{31},
  \bibinfo{pages}{1460--1490}.

\end{thebibliography}


\begin{thebibliography}{5}
\expandafter\ifx\csname natexlab\endcsname\relax\def\natexlab#1{#1}\fi
\providecommand{\url}[1]{\texttt{#1}}
\providecommand{\href}[2]{#2}
\providecommand{\path}[1]{#1}
\providecommand{\DOIprefix}{doi:}
\providecommand{\ArXivprefix}{arXiv:}
\providecommand{\URLprefix}{URL: }
\providecommand{\Pubmedprefix}{pmid:}
\providecommand{\doi}[1]{\href{http://dx.doi.org/#1}{\path{#1}}}
\providecommand{\Pubmed}[1]{\href{pmid:#1}{\path{#1}}}
\providecommand{\bibinfo}[2]{#2}
\ifx\xfnm\relax \def\xfnm[#1]{\unskip,\space#1}\fi
\bibitem[{Bo\v{c}inec et~al.(2026)Bo\v{c}inec, Nagy and
  Yeon}]{Bocinec_etal2026}
\bibinfo{author}{Bo\v{c}inec, F.}, \bibinfo{author}{Nagy, S.},
  \bibinfo{author}{Yeon, H.}, \bibinfo{year}{2026}.
\newblock \bibinfo{title}{Projection depth for functional data: Theoretical
  properties}.
\newblock \bibinfo{note}{ArXiv:2512.20452}.
\bibitem[{Dudley(2002)}]{Dudley2002}
\bibinfo{author}{Dudley, R.M.}, \bibinfo{year}{2002}.
\newblock \bibinfo{title}{Real Analysis and Probability}.
  volume~\bibinfo{volume}{74} of \textit{\bibinfo{series}{Cambridge Studies in
  Advanced Mathematics}}.
\newblock \bibinfo{publisher}{Cambridge University Press},
  \bibinfo{address}{Cambridge}.
\newblock \bibinfo{note}{Revised reprint of the 1989 original}.
\bibitem[{Dyckerhoff(2004)}]{Dyckerhoff2004}
\bibinfo{author}{Dyckerhoff, R.}, \bibinfo{year}{2004}.
\newblock \bibinfo{title}{Data depths satisfying the projection property}.
\newblock \bibinfo{journal}{Allg. Stat. Arch.} \bibinfo{volume}{88},
  \bibinfo{pages}{163--190}.
\bibitem[{Serfling(1980)}]{Serfling1980}
\bibinfo{author}{Serfling, R.}, \bibinfo{year}{1980}.
\newblock \bibinfo{title}{Approximation theorems of mathematical statistics}.
\newblock \bibinfo{publisher}{John Wiley \& Sons Inc.}, \bibinfo{address}{New
  York}.
\newblock \bibinfo{note}{Wiley Series in Probability and Mathematical
  Statistics}.
\bibitem[{Sinova et~al.(2018)Sinova, Gonz\'alez-Rodr\'iguez and
  Van~Aelst}]{Sinova_etal2018}
\bibinfo{author}{Sinova, B.}, \bibinfo{author}{Gonz\'alez-Rodr\'iguez, G.},
  \bibinfo{author}{Van~Aelst, S.}, \bibinfo{year}{2018}.
\newblock \bibinfo{title}{M-estimators of location for functional data}.
\newblock \bibinfo{journal}{Bernoulli} \bibinfo{volume}{24},
  \bibinfo{pages}{2328--2357}.

\end{thebibliography}
